\documentclass{amsart}
\usepackage[margin=1.0in]{geometry}
\usepackage[foot]{amsaddr}
\usepackage{amsthm}

\usepackage[utf8]{inputenc}
\usepackage{latexsym,array,delarray,epsfig,amssymb}%eufrak
\usepackage{caption}
\usepackage{subcaption}
\usepackage{amscd}
\usepackage{verbatim}
\usepackage{tikz}
\usetikzlibrary{decorations.markings}
\usetikzlibrary{arrows.meta}
\usetikzlibrary{shapes}
\usepackage{url}
\pgfdeclarelayer{nodelayer}
\pgfdeclarelayer{edgelayer}
\pgfsetlayers{main,edgelayer,nodelayer}
\usepackage{blkarray}
\usepackage{hyperref}

 \theoremstyle{plain}
 \newtheorem{theorem}{Theorem}[section]
 \newtheorem{lemma}[theorem]{Lemma}
 \newtheorem{proposition}[theorem]{Proposition}
 \newtheorem{corollary}[theorem]{Corollary}

 \theoremstyle{definition}
 \newtheorem{definition}[theorem]{Definition}
 \newtheorem{example}[theorem]{Example}

 \theoremstyle{remark}
 \newtheorem{remark}[theorem]{Remark}

\newcommand\xqed[1]{%
  \leavevmode\unskip\penalty9999 \hbox{}\nobreak\hfill
  \quad\hbox{#1}}

\newcommand{\proofbox}{\ensuremath{\square}}

\usepackage{enumitem}
\setlist[enumerate]{leftmargin=.5in}
\setlist[itemize]{leftmargin=.5in}

 \usepackage[capitalise]{cleveref}

 \crefname{section}{Section}{Sections}
 \crefname{figure}{Figure}{Figures}
 \crefname{table}{Table}{Tables}
 \crefname{lemma}{Lemma}{Lemmas}
 \crefname{theorem}{Theorem}{Theorems}
 \crefname{corollary}{Corollary}{Corollaries}
 \crefname{observation}{Observation}{Observations}
 \crefname{proposition}{Proposition}{Propositions}
 \crefname{conjecture}{Conjecture}{Conjectures}
 \crefname{example}{Example}{Examples}
 \crefname{definition}{Definition}{Definitions}
 \crefname{remark}{Remark}{Remarks}
 \crefname{algorithm}{Algorithm}{Algorithms}
 \crefname{question}{Question}{Questions}
 \crefname{problem}{Problem}{Problems}
\title[On Full Identifiability of Phylogenetic Networks]{On Tree-Network Distinguishability and Full Identifiability of Phylogenetic Networks}

 \author{Jari Brits$^1$}
 \address{$^1$Delft Institute of Applied Mathematics, Delft University of Technology, Delft, The Netherlands}
 \email{j.brits@student.tudelft.nl}

 \author{Niels Holtgrefe$^1$}
 \email{n.a.l.holtgrefe@tudelft.nl}

 \author{Leo van Iersel$^1$}
 \email{l.j.j.vaniersel@tudelft.nl}

 \author{Samuel Martin$^2$}
 \address{$^2$European Molecular Biology Laboratory, European Bioinformatics Institute (EMBL-EBI), Wellcome Genome Campus, Cambridgeshire, UK}
 \email{samuel.martin@ebi.ac.uk}
 \date{\today}

\begin{document}

\maketitle

\begin{abstract}
Phylogenetic networks generalize phylogenetic trees to evolutionary histories that include reticulate events such as recombination, horizontal gene transfer, and hybridization. Under a Markov model of nucleotide substitution, a phylogenetic network determines a distribution of leaf-patterns. Here, we study the identifiability of the network topology from this distribution under the Jukes-Cantor (JC), Kimura 2-parameter (K2P), and Kimura 3-parameter (K3P) models.
Our first result is that the semi-directed network parameter of a level-1 phylogenetic network (modulo redirecting triangles) is fully identifiable under all three models, on
a biologically reasonable 
parameter space in which substitution rates are probabilistic and mixing parameters are non-trivial (i.e., not 0 or 1). In contrast to the generic identifiability established in prior work, this holds at every point of the parameter space, not merely off of a measure-zero subset.
Our second result distinguishes phylogenetic networks from phylogenetic trees, on the same parameter space, under JC.
We prove that no phylogenetic network and phylogenetic tree can induce the same leaf-pattern distribution unless the network is a tree, possibly augmented with certain substructures called $2$-blobs. This means the presence of reticulate evolution creates, in most cases, a detectable signature in the leaf-pattern distribution. More broadly, these results have consequences for identifiability beyond the models and network classes studied here, including for several coalescent-based models.
\end{abstract}

\section{Introduction}
In evolutionary biology, phylogenetic trees are trees (i.e., connected graphs with no cycles) that describe the evolution of a set of species or other taxa. However, trees are unable to capture reticulate evolution events such as recombination, horizontal gene transfer, and hybridization, which are now known to have occurred across the tree of life. Phylogenetic networks are graphs that generalize phylogenetic trees, and are able to describe these reticulate events.

We place a statistical model of nucleotide evolution on a phylogenetic network in the form of a Markov model. These models allow us to formulate the probabilities of observing patterns of nucleotides at the species represented by the leaves in the phylogenetic network, which we call \emph{leaf-patterns} (also called \emph{site-patterns}), and which we observe in real life through DNA sequencing and subsequent sequence alignment. By viewing the set of all leaf-patterns as the image of a polynomial map from the parameter space to the appropriate probability simplex (called the \emph{parameterization map}), and then allowing all parameters to vary freely over $\mathbb{C}$, we can view the observable part of our model as an algebraic variety. This is the fundamental observation of algebraic statistics. Here, we restrict the parameter space to a semi-algebraic set that corresponds to biologically reasonable parameters, and study the image of this set under the parameterization map.

Understanding which parts of a model are identifiable from data is essential if the model is to be used for statistical inference. We focus on identifying the phylogenetic network topology (i.e., the underlying graph without edge lengths) from the leaf-pattern distribution of the model. Generic identifiability of phylogenetic networks has been well studied for the Jukes-Cantor (JC) model~\cite{jukes1969evolution}, the Kimura 2-parameter (K2P) model~\cite{kimura1980simple}, and the Kimura 3-parameter (K3P) model~\cite{kimura1981estimation} on a class of phylogenetic networks called \emph{level-1} \cite{gross2018distinguishing, gross2021distinguishing,hollering2021identifiability}, and has been extended in some cases to the larger class of arbitrary group-based models through dimension arguments \cite{cox2025group,gross2024dimensions}. Recently, generic identifiability results on level-2 phylogenetic networks have been obtained for the JC model \cite{englander2025identifiability}. Generic identifiability of a class of phylogenetic networks tells us that, within that class, the phylogenetic network topology is identifiable except on a subset of Lebesgue measure zero. In algebraic-geometric terms, this is equivalent to the intersection of the varieties corresponding to any two different phylogenetic networks in the class (under the same substitution model) having strictly smaller dimension than that of either model.

For methods such as \cite{barton2026statistical,cummings2026pfaffian,kong2025inference,martin2025algebraic} that attempt to infer a phylogenetic network from aligned DNA sequence data through analysis of the distribution of leaf-patterns, it is important to know exactly if and where the corresponding varieties intersect. If, for a class of models, the varieties do not intersect at all on some region of interest, we say that the model is \emph{fully identifiable} on this region. We use the adjective \emph{fully} to emphasize the distinction between this and generic identifiability. Here, we extend the full identifiability results of \cite{englander2025identifiability} by giving further full identifiability results for the JC and K2P models, and extending these results to the K3P model. 
%We say that the network parameter of a class of phylogenetic networks is \emph{fully identifiable} on a parameter set $\Theta$, if given any two networks in the class, the intersection of the images of $\Theta$ under the parameterization maps is empty, that is $\psi_{\mathcal{N}_1}(\Theta)\cap\psi_{\mathcal{N}_2}(\Theta)=\emptyset$, where $\psi_{\mathcal{N}_i}$ is the parameterization map of $\mathcal{N}_i$, for $i=1,2$. 

In Section~\ref{sec:level1identifiability} we give a series of full identifiability results for small level-1 phylogenetic networks, which we then use to give a full identifiability result for all level-1 phylogenetic networks on a set of biologically reasonable parameters (defined in Section~\ref{sec:prelim}). We find that on this set, two level-1 phylogenetic network models with $n$ leaves intersect exactly on the models of their maximal shared displayed networks (Corollary~\ref{cor:level1identifiability}).

In Section~\ref{sec:levelktrinet}, we consider tree-network distinguishability. Roughly speaking, this means we aim to distinguish trees from networks, based on leaf-pattern distributions. We show that we can do that, under JC, unless the network is a tree augmented with $2$-blobs, which are reticulated components of the network with only two incident edges (see Figure~\ref{fig:tree_vs_network} for an example and Section~\ref{subsec:networks} for a formal definition). We start by giving
%e go on to give further 
identifiability results for small level-$k$ phylogenetic networks with 3 leaves (called \emph{trinets}) for arbitrary $k\in\mathbb{N}$. 
In particular, we show that under JC, the three-leaf star tree is distinguishable from every trinet with at least one reticulation and no $2$-blob, positively answering Conjecture~2.16 in \cite{englander2025identifiability}. We then generalize these results to networks with more leaves, showing that
%, modulo suppressing certain small substructures called \emph{2-blobs},
reticulate evolution is almost always statistically visible (Corollary~\ref{cor:trees_vs_networks}) 
in the sense that no phylogenetic tree can induce the same leaf-pattern distribution as a phylogenetic network, unless the network is a tree possibly augmented with $2$-blobs.

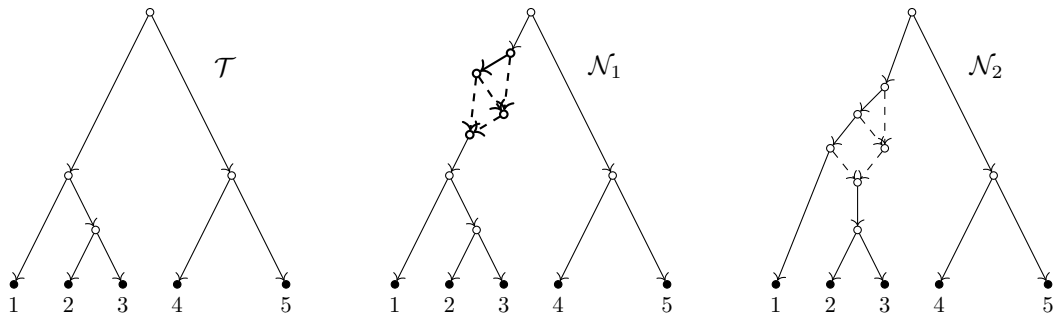
\begin{figure}[h!]
        \centering
    \begin{tikzpicture}[scale=0.36]
    \tikzstyle{node}=[circle, draw=black, fill=black, scale=0.25]
    \tikzstyle{edge}=[]
    \tikzstyle{dashed_edge}=[dashed]
    \tikzstyle{arc}=[->]
    \tikzstyle{dashed_arc}=[-{Latex[scale=1.2]}, dashed]
    \tikzstyle{none}=[]
    \tikzstyle{medium_label}=[none,scale=1.0]
    \tikzstyle{main_label}=[none,scale=0.8]
    \tikzstyle{internal_node}=[shape=circle,draw=black,scale=0.3]

    \tikzstyle{special_edge3}=[->,thick] %non reticulation edge in the 2-blob
    \tikzstyle{special_edge2}=[->,thick,dashed] %reticulation edge in the 2-blob
    \tikzstyle{special_node}=[shape=circle,draw=black,scale=0.3,thick] %node in the 2-blob

    \tikzstyle{special_node2}=[shape=circle,draw=black,scale=0.3] %node in the 3-blob
    \tikzstyle{special_edge6}=[->,dashed] %reticulation edge in the 3-blob
    \tikzstyle{special_edge7}=[->] %non reticulation edge in the 3-blob
    
    \tikzstyle{leaf_node}=[shape=circle,draw=black,scale=0.3,fill=black] %0.275
    \tikzstyle{ret_arc}=[-{>[scale=.8]}, dashed]
    \tikzstyle{small_label}=[scale=0.85]

	\begin{pgfonlayer}{nodelayer}
		\node [style={medium_label}] (529) at (8.5, 8.75) {$\mathcal{T}$};
		\node [style={leaf_node}] (530) at (2.75, 0.75) {};
		\node [style={leaf_node}] (531) at (4.75, 0.75) {};
		\node [style={leaf_node}] (532) at (0.75, 0.75) {};
		\node [style={leaf_node}] (533) at (6.75, 0.75) {};
		\node [style={leaf_node}] (534) at (10.75, 0.75) {};
		\node [style={internal_node}] (535) at (3.75, 2.75) {};
		\node [style={internal_node}] (536) at (2.75, 4.75) {};
		\node [style={internal_node}] (537) at (5.75, 10.75) {};
		\node [style={internal_node}] (538) at (8.75, 4.75) {};
		\node [style={leaf_node}] (539) at (16.75, 0.75) {};
		\node [style={leaf_node}] (540) at (18.75, 0.75) {};
		\node [style={leaf_node}] (541) at (14.75, 0.75) {};
		\node [style={leaf_node}] (542) at (20.75, 0.75) {};
		\node [style={leaf_node}] (543) at (24.75, 0.75) {};
		\node [style={internal_node}] (544) at (17.75, 2.75) {};
		\node [style={internal_node}] (545) at (16.75, 4.75) {};
		\node [style={internal_node}] (546) at (19.75, 10.75) {};
		\node [style={internal_node}] (547) at (22.75, 4.75) {};
		\node [style={special_node}] (548) at (17.5, 6.25) {};
		\node [style={special_node}] (549) at (17.75, 8.5) {};
		\node [style={special_node}] (550) at (18.75, 7) {};
		\node [style={special_node}] (551) at (19, 9.25) {};
		\node [style={leaf_node}] (552) at (30.75, 0.75) {};
		\node [style={leaf_node}] (553) at (32.75, 0.75) {};
		\node [style={leaf_node}] (554) at (28.75, 0.75) {};
		\node [style={leaf_node}] (555) at (34.75, 0.75) {};
		\node [style={leaf_node}] (556) at (38.75, 0.75) {};
		\node [style={internal_node}] (557) at (31.75, 2.75) {};
		\node [style={internal_node}] (559) at (33.75, 10.75) {};
		\node [style={internal_node}] (560) at (36.75, 4.75) {};
		\node [style={special_node2}] (561) at (32.75, 8) {};
		\node [style={special_node2}] (562) at (31.75, 4.5) {};
		\node [style={special_node2}] (563) at (31.75, 7) {};
		\node [style={special_node2}] (564) at (32.75, 5.75) {};
		\node [style={special_node2}] (565) at (30.75, 5.75) {};
		\node [style={small_label}] (566) at (28.75, 0) {$1$};
		\node [style={small_label}] (567) at (30.75, 0) {$2$};
		\node [style={small_label}] (568) at (32.75, 0) {$3$};
		\node [style={small_label}] (569) at (34.75, 0) {$4$};
		\node [style={small_label}] (570) at (38.75, 0) {$5$};
		\node [style={small_label}] (571) at (14.75, 0) {$1$};
		\node [style={small_label}] (572) at (16.75, 0) {$2$};
		\node [style={small_label}] (573) at (18.75, 0) {$3$};
		\node [style={small_label}] (574) at (20.75, 0) {$4$};
		\node [style={small_label}] (575) at (24.75, 0) {$5$};
		\node [style={small_label}] (576) at (0.75, 0) {$1$};
		\node [style={small_label}] (577) at (2.75, 0) {$2$};
		\node [style={small_label}] (578) at (4.75, 0) {$3$};
		\node [style={small_label}] (579) at (6.75, 0) {$4$};
		\node [style={small_label}] (580) at (10.75, 0) {$5$};
		\node [style={medium_label}] (581) at (22.5, 8.75) {$\mathcal{N}_1$};
		\node [style={medium_label}] (582) at (36.5, 8.75) {$\mathcal{N}_2$};
	\end{pgfonlayer}
	\begin{pgfonlayer}{edgelayer}
		\draw [style=arc] (535) to (530);
		\draw [style=arc] (535) to (531);
		\draw [style=arc] (536) to (535);
		\draw [style=arc] (536) to (532);
		\draw [style=arc] (537) to (538);
		\draw [style=arc] (538) to (533);
		\draw [style=arc] (538) to (534);
		\draw [style=arc] (537) to (536);
		\draw [style=arc] (544) to (539);
		\draw [style=arc] (544) to (540);
		\draw [style=arc] (545) to (544);
		\draw [style=arc] (545) to (541);
		\draw [style=arc] (546) to (547);
		\draw [style=arc] (547) to (542);
		\draw [style=arc] (547) to (543);
		\draw [style=arc] (548) to (545);
		\draw [style=arc] (546) to (551);
		\draw [style={special_edge2}] (550) to (548);
		\draw [style={special_edge2}] (549) to (548);
		\draw [style={special_edge2}] (549) to (550);
		\draw [style={special_edge2}] (551) to (550);
		\draw [style={special_edge3}] (551) to (549);
		\draw [style=arc] (557) to (552);
		\draw [style=arc] (557) to (553);
		\draw [style=arc] (559) to (560);
		\draw [style=arc] (560) to (555);
		\draw [style=arc] (560) to (556);
		\draw [style=arc] (565) to (554);
		\draw [style=arc] (562) to (557);
		\draw [style=arc] (559) to (561);
		\draw [style={special_edge6}] (565) to (562);
		\draw [style={special_edge6}] (564) to (562);
		\draw [style={special_edge6}] (563) to (564);
		\draw [style={special_edge6}] (561) to (564);
		\draw [style={special_edge7}] (561) to (563);
		\draw [style={special_edge7}] (563) to (565);
	\end{pgfonlayer}
\end{tikzpicture}

        \caption{Three rooted phylogenetic networks~$\mathcal{T}, \mathcal{N}_1$, and~$\mathcal{N}_2$ on the set of taxa $\mathcal{X} = \{1,2,3,4,5\}$. $\mathcal{T}$ is also a rooted phylogenetic tree, whereas $\mathcal{N}_1$ and $\mathcal{N}_2$ are not. $\mathcal{N}_1$ can be obtained from $\mathcal{T}$ by augmenting it with a 2-blob (in bold). $\mathcal{N}_2$ cannot be obtained by augmenting a rooted phylogenetic tree with $2$-blobs (since it contains a $3$-blob). Hence, by Corollary~\ref{cor:trees_vs_networks} we can distinguish $\mathcal{T}$ and $\mathcal{N}_1$ from $\mathcal{N}_2$, but not $\mathcal{T}$ from $\mathcal{N}_1$.}
        \label{fig:tree_vs_network}
    \end{figure}

%Other lines of work study identifiability under different evolutionary models. In particular, 
In Section~\ref{sec:discussion}, we discuss some consequences of our results under different evolutionary models. Coalescent-based models describe how gene trees arise within a phylogenetic network through the multispecies coalescent, thereby accounting for incomplete lineage sorting (ILS); see \cite{allman2025beyond,allman2024identifiability} and references therein. The substitution models we consider here do not incorporate ILS. A reticulation is instead treated as a mixture over its parent edges, so that the leaf-pattern distribution of the network is a convex combination of the leaf-pattern distributions of its displayed trees. Nevertheless, some of our combinatorial results
imply similar results for models with ILS based on
\cite{allman2025beyond}.
 
%-------------------------------------------------------------------------------------------------
%                                    Preliminaries
%
%-------------------------------------------------------------------------------------------------
 \section{Preliminaries}\label{sec:prelim}
 
In this section, we lay out the background needed for the main results. We review phylogenetic networks and Markov models of evolution (in particular, group-based models).

\subsection{Phylogenetic Networks}\label{subsec:networks}
Our first definitions concern phylogenetic networks. These are graphs that generalize phylogenetic trees by allowing edges that describe reticulate events such as recombination, hybridization, and gene transfer.

\begin{definition}\label{def:network}
     A  \emph{binary rooted phylogenetic network} $\mathcal{N}$ on a set of taxa $\mathcal{X}$ is a directed acyclic graph with no parallel edges that satisfies the following.
 \begin{itemize}
 \item There is a distinguished vertex called the \emph{root} that has indegree~0 and outdegree~2.
 \item All vertices of outdegree 0 have indegree 1. These vertices are called \emph{leaves} and are in one-to-one correspondence with the elements of $\mathcal{X}$.
 \item All other vertices have either indegree 1 and outdegree 2 (called \emph{tree vertices}), or indegree 2 and outdegree 1 (called \emph{reticulation vertices}).  The incoming edges of a reticulation vertex are called \emph{reticulation edges}.
 \item The root is the only vertex that is on all directed paths from the root to a leaf.
 \end{itemize}
\end{definition}
In applications, the set $\mathcal{X}$ is a set of taxa (e.g., species, sub-species, populations) that label the leaves. Here, we will frequently use the set $[n]=\{1,\ldots,n\}$ to label the leaves of $\mathcal{N}$. We identify each leaf with its label. The internal vertices of $\mathcal{N}$ represent ancestral taxa to the taxa at the leaves. Thus $\mathcal{N}$ represents a possible evolutionary history of the taxa in $\mathcal{X}$, and is therefore an \emph{explicit} phylogenetic network (see e.g. \cite{huson2010phylogenetic}). The motivation for the last condition in Definition~\ref{def:network} is that the evolutionary history before the last stable ancestor (the lowest vertex that is on all root-leaf paths) is generally not expected to be recoverable from leaf-data. Note that a binary rooted phylogenetic tree is a binary phylogenetic network with no reticulation vertices. 

A \emph{split} of $\mathcal{X}$ is a partition of $\mathcal{X}$ into two subsets, typically written $A\,|\,B$, where $A,B\subsetneq\mathcal{X
}$ and $A\sqcup B = \mathcal{X}$. A \emph{cut edge} is an edge whose removal disconnects the graph. A split $A\,|\,B$ is said to be \emph{displayed} on a phylogenetic tree or network if there exists a cut edge $e$ that splits $A$ and $B$; that is, the edge~$e$ lies on every path between any element of $A$ and any element of $B$. A split $A\,|\,B$ is said to be \emph{trivial} if either $A$ or $B$ consists of a single leaf. For a given set of taxa $\mathcal{X}$, all trivial splits are displayed on all phylogenetic trees and networks on $\mathcal{X}$.

The class of all phylogenetic networks is exceedingly large, so analyses are often restricted to certain classes of phylogenetic networks.
%The \emph{level} of a phylogenetic network $\mathcal{N}$ is the maximum number of reticulation vertices in any %2-connected component of $\mathcal{N}$
A phylogenetic network is \emph{level}-$k$ if each
connected subgraph without cut edges has at most~$k$ reticulation vertices, see Figure~\ref{fig:rootedNetworkExample}(a).
In Section \ref{sec:level1identifiability}, we will restrict our analysis to level-1 phylogenetic networks. In Section \ref{sec:levelktrinet}, we will consider level-$k$ phylogenetic networks for arbitrarily large $k\in\mathbb{N}$. A phylogenetic network is \emph{strictly} level-$k$ if the maximum number of reticulation vertices in a connected subgraph without cut edges is equal to~$k$.

For the evolutionary models we will consider here, it is only possible to identify, at most, the \emph{semi-directed} phylogenetic network from the leaf-pattern distribution of the network. This is a mixed graph (i.e. a graph that may contain both undirected and directed edges, but no parallel edges) that is obtained from a phylogenetic network $\mathcal{N}$ by undirecting all edges except the reticulation edges, suppressing the root and exhaustively suppressing any resulting parallel edges and degree two vertices.
We call $\mathcal{N}$ a \emph{rooting} of the semi-directed phylogenetic network. Rootings of a semi-directed phylogenetic network are limited; the root cannot be placed on any edge of a semi-directed network, since the reticulation edges must be directed away from the root. Note that two distinct phylogenetic networks may have the same semi-directed network. Indeed, it is well known for the substitution models we study here that the leaf-pattern distribution of all rootings of a fixed semi-directed network is the same. Therefore, the best we can hope to identify from a leaf-pattern distribution is a semi-directed phylogenetic network. In the case that a phylogenetic network is a phylogenetic tree, then the corresponding semi-directed network is the corresponding unrooted tree. Semi-directed networks with 3 or~4 leaves are called \emph{trinets} and \emph{quarnets}, respectively. The notion of level naturally extends to semi-directed phylogenetic networks. See Figure~\ref{fig:rootedNetworkExample}(b) for an example. A semi-directed level-$0$ phylogenetic network is also called an \emph{unrooted binary phylogenetic tree}.

It is helpful to consider the semi-directed phylogenetic networks that have no non-trivial cut edges, where a cut edge is trivial if it induces a trivial split. These form the building blocks for larger networks, and for level-1 they are 3-star trees and sunlets. Here, an \emph{$n$-sunlet} is a semi-directed phylogenetic network consisting of a single (undirected) cycle of $n$ vertices, and one leaf adjacent to each vertex in the cycle. Exactly one of the cycle vertices is a reticulation vertex, and the corresponding reticulation edges are the incident edges that form part of the cycle.  See $\mathcal{N}_2$ of Figure~\ref{fig:networkExample} for an example of a 4-sunlet. For level-2 (and higher), the picture is more complex (see e.g. \cite{englander2025identifiability}).

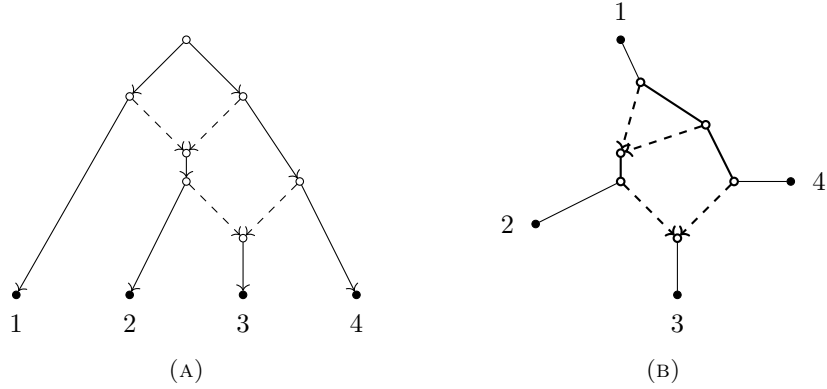
\begin{figure}[t]
        \centering
        \begin{subfigure}[b]{0.45\linewidth}
            \centering
            \begin{tikzpicture}[scale = .75]
    
            \node[shape=circle,draw=black,scale=0.3] (A) at (0,5) {};
            \node[shape=circle,draw=black,scale=0.3] (B) at (-1,4) {};
            \node[shape=circle,draw=black,scale=0.3] (C) at (1,4) {};
            \node[shape=circle,draw=black,scale=0.3] (D) at (0,3) {};
            \node[shape=circle,draw=black,scale=0.3] (E) at (0,2.5) {};
            \node[shape=circle,draw=black,scale=0.3] (F) at (2,2.5) {};
            \node[shape=circle,draw=black,scale=0.3] (G) at (1,1.5) {};
            \node[shape=circle,draw=black,scale=0.3,fill=black] (1) at (-3,0.5) {};
            \node[shape=circle,draw=black,scale=0.3,fill=black] (2) at (-1,0.5) {}; 
            \node[shape=circle,draw=black,scale=0.3,fill=black] (3) at (1,0.5) {};
            \node[shape=circle,draw=black,scale=0.3,fill=black] (4) at (3,0.5) {}; 

            \node (t1) at (-3,0) {$1$};
            \node (t2) at (-1,0) {$2$}; 
            \node (t3) at (1,0) {$3$};
            \node (t4) at (3,0) {$4$}; 
    
            \draw [->](B)--(1);
            \draw [->](F)--(4);
            \draw [->](G)--(3);
            \draw [->](E)--(2);
            \draw [dashed,->] (E)--(G);
            \draw [dashed,->] (F)--(G);
            \draw [dashed,->] (B)--(D);
            \draw [dashed,->] (C)--(D);
            \draw [->](D)--(E);
            \draw [->](C)--(F);
            \draw [->](A)--(B);
            \draw [->](A)--(C);
            
            \end{tikzpicture}
            \subcaption[]{}
        \end{subfigure}
        \begin{subfigure}[b]{0.3\linewidth}
            \centering
            \begin{tikzpicture}[scale = .75]
            
            \node[shape=circle,draw=black,scale=0.3,thick] (B) at (0.35,4.25) {};
            \node[shape=circle,draw=black,scale=0.3,thick] (C) at (1.5,3.5) {};
            \node[shape=circle,draw=black,scale=0.3,thick] (D) at (0,3) {};
            \node[shape=circle,draw=black,scale=0.3,thick] (E) at (0,2.5) {};
            \node[shape=circle,draw=black,scale=0.3,thick] (F) at (2,2.5) {};
            \node[shape=circle,draw=black,scale=0.3,thick] (G) at (1,1.5) {};
            \node[shape=circle,draw=black,scale=0.3,fill=black] (1) at (0,5) {};
            \node[shape=circle,draw=black,scale=0.3,fill=black] (2) at (-1.5,1.75) {}; 
            \node[shape=circle,draw=black,scale=0.3,fill=black] (3) at (1,0.5) {};
            \node[shape=circle,draw=black,scale=0.3,fill=black] (4) at (3,2.5) {}; 

            \node (t1) at (0,5.5) {$1$};
            \node (t2) at (-2,1.75) {$2$}; 
            \node (t3) at (1,0) {$3$};
            \node (t4) at (3.5,2.5) {$4$}; 
    
            \draw (B)--(1);
            \draw (F)--(4);
            \draw (G)--(3);
            \draw (E)--(2);
            \draw [thick, dashed,-{>[scale=.8]}] (E)--(G);
            \draw [thick, dashed,-{>[scale=.8]}] (F)--(G);
            \draw [thick, dashed,-{>[scale=.8]}] (B)--(D);
            \draw [thick, dashed,-{>[scale=.8]}] (C)--(D);
            \draw [thick] (D)--(E);
            \draw [thick] (C)--(F);
            \draw [thick] (B)--(C);
            
            \end{tikzpicture}
            \subcaption[]{}
        \end{subfigure} 
        \caption{\textbf{(a)} A 4-leaf, level-2 binary rooted phylogenetic network on the set of taxa $\mathcal{X}=\{1,2,3,4\}$. Reticulation edges are drawn as dashed lines. \textbf{(b)} The corresponding semi-directed phylogenetic network. In this case, this consists of a single 4-blob (containing all internal vertices, highlighted in bold) and leaves adjacent to the blob.}
        \label{fig:rootedNetworkExample}
    \end{figure}

A \emph{blob} of a (semi-directed) phylogenetic network~$\mathcal{N}$ is a maximal connected subgraph of $\mathcal{N}$ without any cut edges, see Figures~\ref{fig:tree_vs_network} and~\ref{fig:rootedNetworkExample}. A blob is called an $m$-blob for $m\in\mathbb{N}$ when there are exactly $m$ edges incident to the blob but not in the blob. A blob is \emph{non-trivial} if it contains at least $2$ vertices. For level-1 (semi-directed) phylogenetic networks, the only non-trivial blobs are subgraphs consisting of a single cycle. For example, an $n$-sunlet consists of a single $n$-blob with one leaf adjacent to each vertex in the blob.
A \emph{triangle} is an undirected cycle of three vertices and three edges.
The \emph{tree-of-blobs} of a semi-directed phylogenetic network is the
%(multifurcating) phylogenetic 
tree obtained by contracting each blob to a single vertex and then suppressing degree-2 vertices. Note that the tree-of-blobs is an unrooted (not necessarily binary) phylogenetic tree.
%(also called an unrooted  multifurcating phylogenetic tree).

An \emph{up-down} path between two vertices $i$ and $j$ in a semi-directed phylogenetic network is a path in which, for some $\ell\in\mathbb{N}$, the first $\ell$ edges are directed towards $i$ or undirected, and the remaining edges are directed towards $j$ or undirected. Similarly, a \emph{semi-directed path} from~$i$ to~$j$ in a semi-directed phylogenetic network is a path in which all edges are directed towards~$j$ or undirected. Throughout this work, one of the main tools we will use is \emph{restriction}. This enables us to answer questions about a phylogenetic network by looking at smaller networks.
 
\begin{definition}
Let $\mathcal{N}$ be a semi-directed phylogenetic network on a leaf set $\mathcal{X}$, and let $\mathcal{Y}\subset\mathcal{X}$. The \emph{restricted network} $\mathcal{N}|_{\mathcal{Y}}$ of $\mathcal{N}$ (also called \emph{the subnetwork of $\mathcal{N}$ induced by~$\mathcal{Y}$}) is the semi-directed phylogenetic network on the leaf set $\mathcal{Y}$ obtained from $\mathcal{N}$ by taking all up-down paths on $\mathcal{N}$ between all pairs of leaves in $\mathcal{Y}$ and exhaustively suppressing any degree two vertices and identifying parallel edges.
\end{definition}

It is clear that the level of $\mathcal{N}|_{\mathcal{Y}}$ is less than or equal to the level of $\mathcal{N}$. Thus, in some cases the restricted network $\mathcal{N}|_{\mathcal{Y}}$ may be a phylogenetic tree. Note that here, unlike \cite[Definition~2.4]{englander2025identifiability}, we do not necessarily suppress 2-blobs in a restricted network. In the case that $\mathcal{N}$ is a level-1 semi-directed phylogenetic network, the only 2-blobs that can appear during the restriction process are parallel edges, and for the models we will consider (which are closed under convex combinations), suppressing parallel edges does not affect the leaf-pattern distribution of the network, so these 2-blobs are suppressed. For higher levels however, 2-blobs can become more complex and need to be considered carefully. We investigate higher-level 2-blobs in Section~\ref{sec:levelktrinet} (see also \cite{allman2023treeofblobs, englander2025identifiability,holtgrefe2025distinguishing, sullivant2025phylogenetic}).

\begin{definition}\label{def:displayedNetwork}
    Let $\mathcal{N}$ and $\mathcal{N}'$ be two semi-directed phylogenetic networks on a leaf set  $\mathcal{X}$. We say that $\mathcal{N}'$ is \emph{displayed} by $\mathcal{N}$ if it can be obtained from $\mathcal{N}$ by removing the reticulation edges~$e_1,\ldots, e_k$, where no two $e_i$ and $e_j$ are directed into the same reticulation vertex, then taking the union of all up-down paths between leaves in $\mathcal{X}$, and then exhaustively suppressing degree-2 vertices and identifying parallel edges. We call $\mathcal{N}'$ the \emph{displayed network of $\mathcal{N}$ induced by~$e_1,\ldots, e_k$}, and denote it by $\mathcal{N} - \{e_1,\ldots, e_k\}$. In the case $k=1$, we denote it by  $\mathcal{N} - e_1$. We allow the set of reticulation edges to be empty, so that $\mathcal{N}$ displays itself.
\end{definition}

Observe that the relation between a network and a displayed network is transitive, that is, if $\mathcal{N}'$ is displayed by $\mathcal{N}$, and $\mathcal{N}''$ is displayed by $\mathcal{N}'$, then $\mathcal{N}''$ is displayed by $\mathcal{N}$. Furthermore, the number of reticulation vertices of a displayed network $\mathcal{N}'$ is strictly less than the number of reticulation vertices of $\mathcal{N}$, except for the case that $\mathcal{N}=\mathcal{N}'$. Thus the relation is also antisymmetric, and by definition it is reflexive. We therefore have a partial order on the set of displayed networks of~$\mathcal{N}$.

Let $\mathcal{N}$ be a semi-directed phylogenetic network with $r$ reticulation vertices, and let $e_1,\ldots,e_r$ be a set of reticulation edges, one from each reticulation vertex. Then the displayed network $\mathcal{N}-\{e_1,\ldots, e_r\}$ is an unrooted phylogenetic tree, called a \emph{displayed tree} of $\mathcal{N}$. The displayed trees are minimal with respect to the partial order on the displayed networks of $\mathcal{N}$.

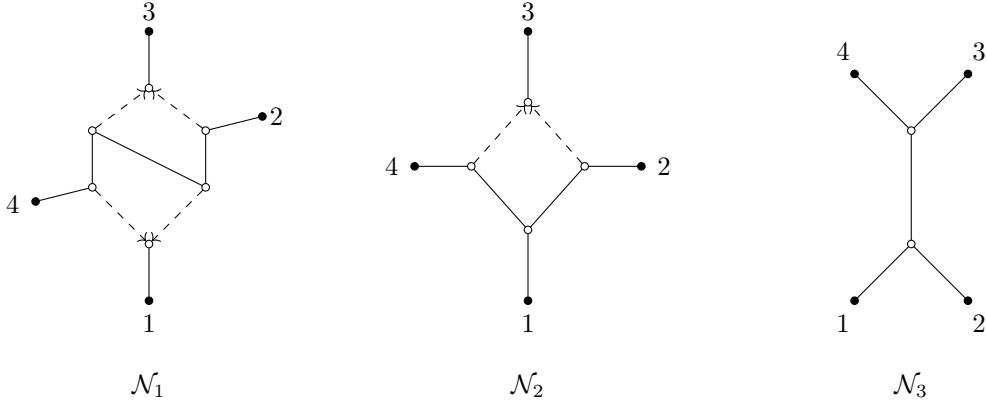
\begin{figure}[h!]
        \centering
        \begin{subfigure}[b]{0.3\linewidth}
            \centering
            \begin{tikzpicture}[scale = .75]
    
            \node[shape=circle,draw=black,scale=0.3] (A) at (0,4.75) {};
            \node[shape=circle,draw=black,scale=0.3] (B) at (1,3) {};
            \node[shape=circle,draw=black,scale=0.3] (C) at (-1,3) {};
            \node[shape=circle,draw=black,scale=0.3] (D) at (0,2) {};
            \node[shape=circle,draw=black,scale=0.3] (E) at (-1,4) {};
            \node[shape=circle,draw=black,scale=0.3] (F) at (1,4) {};
            \node[shape=circle,draw=black,scale=0.3,fill=black] (1) at (0,5.75) {};
            \node[shape=circle,draw=black,scale=0.3,fill=black] (2) at (2,4.25) {}; 
            \node[shape=circle,draw=black,scale=0.3,fill=black] (3) at (-2,2.75) {};
            \node[shape=circle,draw=black,scale=0.3,fill=black] (4) at (0,1) {}; 

            \node (t1) at (0,6.1) {$3$};
            \node (t2) at (2.25,4.25) {$2$}; 
            \node (t3) at (-2.4,2.7) {$4$};
            \node (t4) at (0,0.6) {$1$}; 
    
            \draw (1)--(A);
            \draw [dashed,->] (E)--(A);
            \draw [dashed,->] (F)--(A);
            \draw (F)--(B);
            \draw (E)--(B);
            \draw (E)--(C);
            \draw [dashed,->] (B)--(D);
            \draw [dashed,->] (C)--(D);
            \draw (F)--(2);
            \draw (C)--(3);
            \draw (D)--(4);

            \node (t5) at (0,-0.5) {$\mathcal{N}_1$};
            
            \end{tikzpicture}
        \end{subfigure}
        \begin{subfigure}[b]{0.3\linewidth}
            \centering
            \begin{tikzpicture}[scale = .75]
            \node[shape=circle,draw=black,scale=0.3] (A) at (0,4.5) {};
            \node[shape=circle,draw=black,scale=0.3] (B) at (0,2.25) {};
            \node[shape=circle,draw=black,scale=0.3] (E) at (-1,3.37) {};
            \node[shape=circle,draw=black,scale=0.3] (F) at (1,3.37) {};
            \node[shape=circle,draw=black,scale=0.3,fill=black] (1) at (0,5.75) {};
            \node[shape=circle,draw=black,scale=0.3,fill=black] (2) at (2,3.37) {}; 
            \node[shape=circle,draw=black,scale=0.3,fill=black] (3) at (-2,3.37) {};
            \node[shape=circle,draw=black,scale=0.3,fill=black] (4) at (0,1) {}; 

            \node (t1) at (0,6.1) {$3$};
            \node (t2) at (2.4,3.37) {$2$}; 
            \node (t3) at (-2.4,3.37) {$4$};
            \node (t4) at (0,0.6) {$1$}; 
    
            \draw (1)--(A);
            \draw [dashed,->] (E)--(A);
            \draw [dashed,->] (F)--(A);
            \draw (F)--(B);
            \draw (E)--(B);
            \draw (E)--(3);
            \draw (F)--(2);
            \draw (B)--(4);
            
            \node (t5) at (0,-0.5) {$\mathcal{N}_2$};
            
            \end{tikzpicture}
        \end{subfigure}
        \begin{subfigure}[b]{0.3\linewidth}
            \centering
            \begin{tikzpicture}[scale = .75]
            \node[shape=circle,draw=black,scale=0.3] (B) at (0,2) {};
            \node[shape=circle,draw=black,scale=0.3] (E) at (0,4) {};
            \node[shape=circle,draw=black,scale=0.3,fill=black] (1) at (1,5.) {};
            \node[shape=circle,draw=black,scale=0.3,fill=black] (2) at (1,1) {}; 
            \node[shape=circle,draw=black,scale=0.3,fill=black] (3) at (-1,5) {};
            \node[shape=circle,draw=black,scale=0.3,fill=black] (4) at (-1,1) {}; 

            \node (t1) at (1.2,5.4) {$3$};
            \node (t2) at (1.2,0.6) {$2$}; 
            \node (t3) at (-1.2,5.4) {$4$};
            \node (t4) at (-1.2,0.6) {$1$}; 
    
            \draw (1)--(E);
            \draw (E)--(B);
            \draw (B)--(2);
            \draw (E)--(3);
            \draw (B)--(4);   

            \node (t5) at (0,-0.5) {$\mathcal{N}_3$}; 
            
            \end{tikzpicture}
        \end{subfigure}   
        \caption{Three semi-directed phylogenetic networks on the set of taxa $\mathcal{X}=\{1,2,3,4\}$. Reticulation edges are drawn with dashed lines. The network $\mathcal{N}_1$ is a level-2 network, $\mathcal{N}_2$ is level-1, and $\mathcal{N}_3$ is level-0 (i.e., a tree). Both $\mathcal{N}_2$ and $\mathcal{N}_3$ are displayed networks of $\mathcal{N}_1$, whilst $\mathcal{N}_3$ is also a displayed network of $\mathcal{N}_2$. In terms of the partial order, we have $\mathcal{N}_1 > \mathcal{N}_2 > \mathcal{N}_3$. Since it is minimal, the network $\mathcal{N}_3$ is a displayed tree of $\mathcal{N}_1$ and $\mathcal{N}_2$.}
        \label{fig:networkExample}
    \end{figure}
    
\subsection{Markov Models of Evolution}\label{sec:markov}

To study evolution at the molecular level, we can place a type of Markov model, called a substitution model, on a phylogenetic tree or network. In this work we consider three specific DNA substitution models, called the Jukes-Cantor (JC) model~\cite{jukes1969evolution}, the Kimura 2-parameter (K2P) model~\cite{kimura1980simple}, and the Kimura 3-parameter (K3P) model~\cite{kimura1981estimation}. These models are examples of group-based models, which are particularly amenable to study from the algebraic perspective. We will model evolution on a phylogenetic network using a displayed tree model. This is a mixture model of the corresponding substitution model on the displayed trees of the network.

To place a substitution model on a rooted phylogenetic tree $\mathcal{T}$ with $n$ leaves, we model the evolution of a nucleotide at a single genomic site on $\mathcal{T}$. Thus, our state space is the set of nucleotides $\Omega=\{\rm{A,C,G,T}\}$. First, we label the leaves with the integers $1,\ldots, n$. We associate a random variable $X_v \in\Omega$ to each vertex $v\in V(\mathcal{T})$, and a $4 \times 4$ transition matrix~$M^e$ to each directed edge $e=(u, v) \in E(\mathcal{T})$ such that $M_{i,j}^e = P(X_v = j | X_u = i)$. Lastly, we associate a distribution $\pi$ of states to the root $\rho$ of $\mathcal{T}$.  This enables us to write down the probability of an assignment of states to the vertices of $\mathcal{T}$. If $X = \{x_v \in \Omega\ |\ v \in V(\mathcal{T})\}$ is such an assignment (sometimes called a \emph{character}), then
\[
P(X_v = x_v\ \forall\  v \in V(\mathcal{T})) = \pi_{x_\rho}\prod_{(u,v) \in E(\mathcal{T})}M_{x_u, x_v}^{(u,v)}.
\]
In applications, we typically only observe the random variables at the leaves of $\mathcal{T}$, which usually represent extant species. We can determine the probability of observing a joint state $(x_1, \ldots, x_n) \in \Omega^n$ at the leaves of $\mathcal{T}$, where $x_i$ is the state observed at leaf $i$, by marginalizing over the internal nodes of $\mathcal{T}$. We write down an expression for this in the following manner. Let $X(x_1,\ldots, x_n) \subset \Omega^{|V(\mathcal{T})|}$ be the set of all possible assignments of states at all vertices in $\mathcal{T}$ that have states $x_1,\dots, x_n$ assigned to leaves $1,\ldots,n$ respectively. For $x\in X(x_1,\ldots, x_n)$ and $u\in V(\mathcal{T})$, let $x_u$ denote the state assigned to vertex $u$. Then the probability of observing the joint state $(x_1, \ldots, x_n) \in \Omega^n$ is
\[
p_{x_1 x_2 \ldots x_n} := P(X_1 = x_1, \ldots, X_n = x_n) ~= 
\sum_{x \in X(x_1,\ldots, x_n)}\pi_{x_\rho}\prod_{(u,v) \in E(\mathcal{T})}M_{x_u, x_v}^{(u,v)}, 
\]
where $X_1,\ldots,X_n$ are the random variables associated to the leaves of $\mathcal{T}$. We call an assignment of states to the leaves of $\mathcal{T}$ a \emph{leaf-pattern}.

Let $\theta$ be a set of numerical parameters associated to our model (i.e., the entries of the transition matrices $M^e$ and the distribution of states $\pi_\rho$ at the root), and let $\Theta$ be the space of all such parameter sets. The \emph{parameterization map} $\varphi_{\mathcal{T}}$ associated to our model is the map
\begin{align*}
    \varphi_{\mathcal{T}}: \Theta &\longrightarrow \Delta^{4^n-1} \\
                        \theta &\longmapsto (p_{x_1\ldots x_n}\ |\ x_i \in \Omega,\ i=1,\ldots, n), 
\end{align*}
which takes the parameter space $\Theta$ of our model to the $(4^n-1)$-dimensional
% $4^n$-dimensional
probability simplex~$\Delta^{4^n-1}$. In algebraic phylogenetics it is typical to drop the restrictions on parameters and think of them as independent variables over $\mathbb{C}$. Then, the map $\varphi_{\mathcal{T}}$ is viewed as a rational map from $\mathbb{C}^m$ to 
% $\mathbb{C}^{4^n-1}$
$\mathbb{C}^{4^n}$, where $m$ is the number of numerical parameters in the model. The Zariski closure of the image of this map is called the variety associated to $\mathcal{T}$, which we will denote $\mathcal{V}_{\mathcal{T}}$, and this object is studied from an algebraic perspective (see e.g. \cite{allman2007phylogenetic, allman2008phylogenetic,draisma2009ideals,eriksson2005phylogenetic,raicu2012secant}).

Different substitution models place restrictions on the form of the transition matrices. For example, for the Jukes-Cantor model, transition matrices have the form
\begin{equation}\label{eqn:JC}
\begin{pmatrix}
    1-3\alpha & \alpha & \alpha & \alpha \\
    \alpha & 1-3\alpha & \alpha & \alpha \\
    \alpha & \alpha & 1-3\alpha & \alpha \\
    \alpha & \alpha & \alpha & 1-3\alpha
\end{pmatrix}.
\end{equation}

%A key property of substitution models we will require is \emph{multiplicative closure}, that is, if $M_1$ and $M_2$ are transition matrices drawn from a substitution model, then so is their product $M_1 M_2$. This enables us to suppress vertices of degree 2 on a phylogenetic tree without worrying that the transition matrices for the resulting tree are no longer in the model. A sufficient condition for this is that the model is a Lie Markov model \cite{sumner2012lie}.

To place a substitution model on a phylogenetic network $\mathcal{N}$ with leaf set $\mathcal{X}$, we will use a displayed tree model.
This is a mixture model of the models for the trees displayed by the network. First, label the reticulation vertices  of $\mathcal{N}$ as $1,\ldots, r$. Assign a random variable $X_v \in\Omega$ to each vertex $v\in V(\mathcal{N})$, a transition matrix $M^e$ to each edge $e\in E(\mathcal{N})$, a distribution of states at the root $\pi_\rho$, and to each reticulation vertex $i$ a mixing parameter $\delta_i \in [0,1]$. We encode a choice of reticulation edge for each reticulation vertex with a vector $\sigma \in \{0,1\}^r$, and denote the resulting displayed tree by $\mathcal{T}_\sigma$. Then the parameterization map associated to~$\mathcal{N}$ is
\begin{equation}\label{eqn:networkParam}
    \varphi_{\mathcal{N}} = \sum_{\sigma \in \{0,1\}^r} \big(\prod_{i = 1}^r \delta_i^{\sigma_i}(1-\delta_i)^{1-\sigma_i}\big)\varphi_{\mathcal{T}_\sigma}.
\end{equation}
Here, the parameters assigned to the model of the displayed tree $\mathcal{T}_{\sigma}$ are inherited from $\mathcal{N}$. In this work we only consider group-based substitution models (see next section), and these are multiplicatively closed, splittable, and closed under convex combinations \cite{sullivant2025phylogenetic}. Therefore, when suppressing vertices of degree 2 and identifying parallel edges in the process of obtaining displayed trees, there are transition matrices within the model that we can assign to the new edges that give the same leaf-pattern distribution as if we had not performed these operations \cite[Proposition~4.6, Proposition~5.4]{sullivant2025phylogenetic}. In the following sections, we specifically do not allow mixing parameters $\delta_i$ to be $0$ or $1$ (see next section).

As in the case for trees, we can view $\varphi_{\mathcal{N}}$ as a rational map between two affine spaces over~$\mathbb{C}$ and consider the Zariski closure of the image of this map. This is the variety associated to~$\mathcal{N}$, and we denote it $\mathcal{V}_{\mathcal{N}}$.

\subsection{Group-based Models}
A group-based substitution model is a time-reversible substitution model for which the state space is associated with an abelian group. The JC, K2P, and K3P models are all group-based models, where the state space $\{{\rm A,C,G,T}\}$ is associated with the Klein-4 group $G = \mathbb{Z}/2\mathbb{Z}\times\mathbb{Z}/2\mathbb{Z}$. This association enables the use of a discrete Fourier transformation that greatly simplifies the parameterization \cite{evans1993invariants, hendy1996complete}, and thereby makes them much more amenable to study from an algebraic perspective. Often, probabilistic requirements, such as the rows of transition matrices summing to $1$ and transition matrix entries being real numbers in $[0,1]$, are relaxed when viewed from the algebraic perspective.

The discrete Fourier transformation is applied to the parameterization of the model. The transition matrices of the model are diagonalized, and so the new parameters are the eigenvalues of the transition matrices. For an edge $e$ in a phylogenetic tree or network, we will denote these new parameters by $a^e_{\rm A}, a^e_{\rm C}, a^e_{\rm G},$ and~$a^e_{\rm T}$.

The transformed probability $p_{x_1,\ldots, x_n}$ is typically denoted by $q_{x_1,\ldots,x_n}$. For a phylogenetic tree $\mathcal{T}$, in the new coordinates it is given by the monomial expression
\begin{equation}\label{eqn:groupparam}
q_{x_1,\ldots,x_n} = \begin{cases}
    \displaystyle\prod_{e\in E(\mathcal{T})} a^e_{x_e} & x_1 + \cdots + x_n = 0\\
    \quad 0 & \text{otherwise,}
\end{cases}
\end{equation}
where each edge $e$ is labelled by an element $x_e \in G$, defined to be the sum (in $G$) of all leaf states $x_i$ for those leaves $i$ in the subtree below $e$. The qualifier \lq below\rq\ is determined by the direction of the edge $e$, but for the group $\mathbb{Z}/2\mathbb{Z}\times\mathbb{Z}/2\mathbb{Z}$ this direction has no effect on the parameterization, because the sum of states assigned to all leaves is $0$, and each element of this group is self-inverse. In the context of group-based models, leaf-patterns $(x_1,\ldots, x_n)$ for which $x_1+\cdots +x_n = 0$ (i.e., those where $q_{x_1,\ldots,x_n}\neq 0$) are often called \emph{consistent leaf labellings}. For further details the reader may consult \cite[Section~15.3]{sullivant2018algebraic}. As in equation (\ref{eqn:networkParam}), the phylogenetic network parameterization map in Fourier coordinates, denoted $\psi_\mathcal{N}$, is given as a mixture of the Fourier parameterizations of the displayed trees, as follows.
\begin{equation}\label{eqn:fourierNetworkParam}
    \psi_{\mathcal{N}} = \sum_{\sigma \in \{0,1\}^r} \big(\prod_{i = 1}^r \delta_i^{\sigma_i}(1-\delta_i)^{1-\sigma_i}\big)\psi_{\mathcal{T}_\sigma}.
\end{equation}

\begin{figure}[h!]

    \centering
    \begin{tikzpicture}[scale = .5]
        
        \node[shape=circle,draw=black,scale=0.3] (B) at (0,2) {};
        \node[shape=circle,draw=black,scale=0.3] (C) at (1.8,5) {};
        \node[shape=circle,draw=black,scale=0.3] (D) at (-1.8,5) {};
        \node[shape=circle,draw=black,scale=0.3,fill=black] (2) at (3,6.5) {};
        \node[shape=circle,draw=black,scale=0.3,fill=black] (1) at (-3,6.5) {}; 
        \node[shape=circle,draw=black,scale=0.3,fill=black] (3) at (0,0) {};
        \node at (3.2,7.2) {$2$};
        \node at (-3.2,7.2) {$1$};
        \node at (0,-1) {$3$};

        \draw [dashed,->](C)--(B) node[midway,right=0.1] {$e$};
        \draw [dashed,->](D)--(B) node[midway,left=0.1] {$d$};
        \draw [->](D)--(C) node[midway,above=0.05] {$f$};
        \draw [->](C)--(2) node[midway,below right] {$b$};
        \draw [->](D)--(1) node[midway,below left] {$a$};
        \draw [->](B)--(3) node[midway,left=0.05] {$c$};

    \end{tikzpicture}
    \caption{A leaf and edge-labelled 3-sunlet network. }
    \label{fig:3sunlet}
\end{figure}
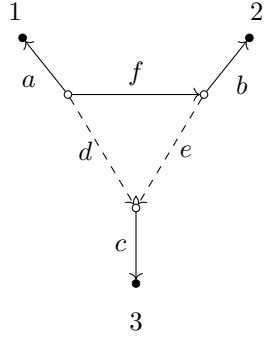

\begin{example}\label{ex:3sunlet}
    Consider the semi-directed phylogenetic network $\mathcal{N}$ in Figure~\ref{fig:3sunlet}. This has a single reticulation vertex and therefore has two displayed trees, each a 3-star tree. Let $\mathcal{T}_1$ be the 3-star tree got by deleting the edge~$e$, and let $\mathcal{T}_2$ be the 3-star tree got by deleting the edge~$d$. For ease of notation, we will not suppress vertices of degree 2 here. In $\mathcal{T}_1$, the Fourier coordinate $q_{x_1x_2x_3}$ for any $x_1, x_2, x_3\in G$ with $x_1+x_2+x_3 = 0$, is given by
    $$ q_{x_1x_2x_3}=a_{x_1}b_{x_2}c_{x_3}d_{x_3}f_{x_2},$$
    where we are reusing edge labels as parameter labels. For $\mathcal{T}_2$ we have
    $$ q_{x_1x_2x_3}=a_{x_1}b_{x_2}c_{x_3}e_{x_3}f_{x_2+x_3}.$$
    The Fourier parameterization of $\mathcal{N}$ is therefore given by
    \begin{align*} q_{x_1x_2x_3} &= \delta a_{x_1}b_{x_2}c_{x_3}d_{x_3}f_{x_2} + (1-\delta)a_{x_1}b_{x_2}c_{x_3}e_{x_3}f_{x_2+x_3} \\ 
                &=a_{x_1}b_{x_2}c_{x_3}(\delta d_{x_3}f_{x_2} + (1-\delta)e_{x_3}f_{x_2+x_3}).
    \end{align*}
    \xqed{$\triangle$}
\end{example}

We have described phylogenetic network models as mixture models of the corresponding displayed tree models. Equivalently, one can also think of them recursively as mixture models of the displayed networks, as the following lemma demonstrates.

\begin{lemma}\label{lem:displayedNetworkParam}
    Let $\mathcal{N}$ be a semi-directed phylogenetic network. For any reticulation vertex~$v$ in~$\mathcal{N}$ with corresponding reticulation edges $e_1$ and $e_2$, and mixing parameter $\delta$,  let $\mathcal{N}_1$ and~$\mathcal{N}_2$ be the displayed networks of $\mathcal{N}$, got by removing reticulation edges $e_2$ and $e_1$, respectively. Then 
    $$\psi_{\mathcal{N}} = \delta\psi_{\mathcal{N}_1} + (1-\delta)\psi_{\mathcal{N}_2}.$$
\end{lemma}
\begin{proof}
    First observe that the set of displayed trees of $\mathcal{N}$ is equal to the union of the sets of displayed trees of $\mathcal{N}_1$ and $\mathcal{N}_2$, which we denote $T_1$ and $T_2$ respectively. For each displayed tree~$\mathcal{T}$, let $\delta_{\mathcal{T}}$ be the product of mixing parameters for that tree (that is, if $\mathcal{T} = \mathcal{T}_{\sigma}$ as in equation (\ref{eqn:fourierNetworkParam}), then we have $\delta_{\mathcal{T}} = \prod_{i=1}^r\delta_i^{\sigma_i}(1-\delta_i)^{1-\sigma_i}$). Since $\delta$ is the mixing parameter for $v$, for trees in $T_1$ we can write $\delta_{\mathcal{T}}=\delta\delta'_{\mathcal{T}}$, and similarly for trees in $T_2$ we can write $\delta_{\mathcal{T}}=(1-\delta)\delta'_{\mathcal{T}}$. Here, $\delta'$ is the product of mixing parameters for all reticulation vertices except $v$, i.e., those in $\mathcal{N}_1$ and $\mathcal{N}_2$. Then we have
    \begin{align*}
        \psi_{\mathcal{N}} &= \sum_{\mathcal{T}\in T_1\cup T_2} \delta_{\mathcal{T}} \psi_{\mathcal{T}} \\
            &= \delta \sum_{\mathcal{T}\in T_1} \delta'_{\mathcal{T}} \psi_{\mathcal{T}} + (1-\delta)\sum_{\mathcal{T}\in T_2} \delta'_{\mathcal{T}} \psi_{\mathcal{T}}\\
            &=\delta\psi_{\mathcal{N}_1} + (1-\delta)\psi_{\mathcal{N}_2}.
    \end{align*}
\end{proof}

As in \cite{englander2025identifiability}, we will only consider parameters from a restricted parameter set. First, we will require that our transition matrices are indeed transition matrices, that is, they have real entries in $[0,1]$ and rows that sum to $1$. Under this requirement, each transition matrix has $1$ as an eigenvalue, so for each edge $e$ we set $a^e_{\rm A} = 1$. We further do not allow transition matrices to be the identity matrix (which would correspond to a branch of length zero in the corresponding tree or network), and require that transition matrices are positive definite. (These restrictions are those that we would have on transition matrices from a continuous-time Markov model with symmetric rate matrices.) Then, the remaining eigenvalues $a^e_{\rm C}, a^e_{\rm G},$ and $a^e_{\rm T}$ all lie in the open interval $(0,1)$. Finally, we require that all mixing parameters $\delta_i$ lie in the open interval $(0,1)$. We denote by $\Theta_0(\mathcal{N})$ this set of restricted parameters for the phylogenetic network $\mathcal{N}$. We will simply use $\Theta_0$ when $\mathcal{N}$ is clear from context. 

The requirement that transition matrices are positive definite places reasonable restrictions on the probabilities of mutation events along edges in the network. For example, in the JC model, this requirement is equivalent to the restriction that $0 \leq \alpha < \frac{1}{4}$, where $\alpha$ is as in~(\ref{eqn:JC}), although here we further require that $\alpha\neq 0$. This means that the probability of a mutation event occurring along a particular edge is between $0$ and $\frac{3}{4}$, which certainly contains the biologically plausible range.
Under the JC model, for any transition matrix, there is only a single eigenvalue distinct from $1$, that is $a^e_{\rm C} = a^e_{\rm G} = a^e_{\rm T}$ for each edge $e$. Under the K2P model, we have two eigenvalues distinct from~$1$. Here, we choose to identify those corresponding to ${\rm C}$ and ${\rm T}$, that is, $a^e_{\rm C} = a^e_{\rm T}$ for each edge~$e$.

For a semi-directed phylogenetic network $\mathcal{N}$ on $n$ leaves and a fixed substitution model, we define the model $\mathcal{M}_{\mathcal{N}}$ to be the image of the parameterization map $\psi_\mathcal{N}$ on the restricted set of parameters $\Theta_0$,
$$\mathcal{M}_{\mathcal{N}} = \psi_{\mathcal{N}}(\Theta_0) \subset\Delta^{4^n-1}.$$
Note that this differs from other works (e.g. \cite{gross2018distinguishing}), where $\mathcal{M}_{\mathcal{N}}$ denotes the image of the parameterization map over $\mathbb{C}^{m}$, i.e., no restrictions are placed on the numerical parameters. Observe that for any phylogenetic network $\mathcal{N}$ and any of the JC, K2P, or K3P substitution models, the set $\Theta_0$ is semi-algebraic. Then, as a consequence of the Tarski-Seidenberg theorem (see e.g. \cite{bochnak1998real}), since $\psi_{\mathcal{N}}$ is a polynomial mapping, the set $\mathcal{M}_{\mathcal{N}}$ is also semi-algebraic.

%\begin{remark}\label{rem:dense}
%    Observe that $\mathcal{M}_{\mathcal{N}}$ is dense (with respect to the Zariski topology) in $\mathcal{V}_{\mathcal{N}}$. This is because $\Theta_0$, as a product of $m$ non-trivial intervals of real numbers, is dense in $\mathbb{C}^m$ (where $m$ is the number of parameters). We can view  $\psi_{\mathcal{N}}$ as a polynomial map from $\mathbb{C}^m$ to $\mathbb{C}^{4^n}$, and therefore a morphism of varieties. It follows that $\psi_{\mathcal{N}}(\overline{\Theta_0}) \subseteq \overline{\psi_{\mathcal{N}}(\Theta_0)}$ and so $\mathcal{V}_{\mathcal{N}}=\overline{{\rm im}\,\psi_{\mathcal{N}}} = \overline{\psi_{\mathcal{N}}(\Theta_0)}$.
%\end{remark}

We say that the network parameter of a class of phylogenetic networks is \emph{fully identifiable} on a parameter set $\Theta$, if given any two networks $\mathcal{N}_1$ and $\mathcal{N}_2$ in the class, the intersection of the images of $\Theta(\mathcal{N}_1)$ and $\Theta(\mathcal{N}_2)$ under the respective parameterization maps is empty, that is $\psi_{\mathcal{N}_1}(\Theta)\cap\psi_{\mathcal{N}_2}(\Theta)=\emptyset$.

\begin{lemma}\label{lem:paramlemma}
    Let $\mathcal{N}$ be a semi-directed phylogenetic network on a set of $n$ taxa $\mathcal{X} = [n]$, with a reticulation vertex $r$ and corresponding reticulation edges $e'$ and $e''$, and a leaf (say, the leaf labelled $1$) adjacent to $r$. Let $\mathcal{N}'$ and $\mathcal{N}''$ be the displayed networks of $\mathcal{N}$ got by removing reticulation edges $e''$ and $e'$ respectively. For a consistent leaf labelling $x_1,\ldots, x_n$, let $q_{x_1,\ldots, x_n},q'_{x_1,\ldots, x_n}$ and $q''_{x_1,\ldots,x_n}$ be the Fourier coordinates corresponding to $\mathcal{N},\mathcal{N}'$ and $\mathcal{N}''$ respectively for the JC, K2P, or K3P substitution model. Then on the parameter space $\Theta_0(\mathcal{N})$, we have 
    $$q_{{\rm A},x_2\ldots, x_n}  = q'_{{\rm A},x_2\ldots, x_n} = q''_{{\rm A},x_2\ldots,x_n},$$
    for all $x_2,\ldots,x_n \in G$. 
\end{lemma}
\begin{proof}
    By Lemma~\ref{lem:displayedNetworkParam} we can write the Fourier coordinates of $\mathcal{N}$ as
    $$q_{{\rm A},x_2\ldots, x_n} = \delta q'_{{\rm A},x_2\ldots, x_n} + (1-\delta)q''_{{\rm A},x_2\ldots, x_n},$$
    where $\delta$ is the mixing parameter associated to $r$. Thus, it is sufficient to show that $q'_{{\rm A},x_2\ldots, x_n} = q''_{{\rm A},x_2\ldots,x_n}$.
        
    Let $\mathcal{N}_1, \ldots,\mathcal{N}_{\ell}$ be the displayed networks of $\mathcal{N}$ whose only reticulation vertex is~$r$. For each $i=1,\ldots,\ell$ let $\mathcal{T}_i'$ and $\mathcal{T}_i''$ be the displayed trees of $\mathcal{N}_i$, got by removing reticulation edges $e''$ and $e'$ respectively. It follows that $\mathcal{T}_1',\ldots,\mathcal{T}_{\ell}'$ are the displayed trees of $\mathcal{N}'$, and $\mathcal{T}_1'',\ldots,\mathcal{T}_{\ell}''$ are the displayed trees of $\mathcal{N}''$. Therefore the Fourier coordinates of $\mathcal{N}'$ and $\mathcal{N}''$ are given by
    \begin{equation}\label{eqn:displayedTreeParams}
    q'_{{\rm A},x_2\ldots, x_n} = \sum_{i=1}^{\ell} \delta_i q'^i_{{\rm A},x_2\ldots, x_n},\quad q''_{{\rm A},x_2\ldots, x_n} = \sum_{i=1}^{\ell} \delta_i q''^i_{{\rm A},x_2\ldots, x_n},
    \end{equation}
    where $\delta_i$ is the product of mixing parameters associated to the removed reticulation edges in $\mathcal{N}_i$, and $q'^i_{{\rm A},x_2\ldots, x_n}$ and $q''^i_{{\rm A},x_2\ldots, x_n}$ are the Fourier coordinates of $\mathcal{T}_i'$ and $\mathcal{T}_i''$ respectively.

    Next, for each pair of trees $\mathcal{T}_i'$ and $\mathcal{T}_i''$, let $E_i$ be the set of shared edges, so that $E(\mathcal{T}_i') = E_i\cup\{e'\}$ and $E(\mathcal{T}_i'') = E_i\cup\{e''\}$. Then for a fixed leaf labelling $x_1={\rm A}, x_2, \ldots, x_n$, the Fourier coordinates of $\mathcal{T}_i'$ and $\mathcal{T}_i''$ are given by 
    $$q'^{i}_{{\rm A},x_2,\ldots,x_n} = a_{\rm A}^{e'}\prod_{e\in E_i}a^e_{x_e'}$$
    and
    $$q''^{i}_{{\rm A},x_2,\ldots,x_n} = a_{\rm A}^{e''}\prod_{e\in E_i}a^e_{x_e''}$$
    respectively.  Observe that for each edge $e\in E_i$, the split induced by $e$ in $\mathcal{T}_i'$ differs from the split induced by $e$ in $\mathcal{T}_i''$ only by the placement of leaf 1. That is, if $A\,|\,B$ is the split induced by $e$ in $\mathcal{T}_i'$, and supposing (without loss of generality) that $1\in A$, the split induced by $e$ in $\mathcal{T}_i''$ is either $A|B$ or $(A\setminus\{1\})|(B\cup\{1\})$. In both cases, since $x_1 = {\rm A}$ (the identity of $G$), we must have $x_e'=\sum_{i\in A}x_i = x_e''$.
    
    Finally, since in $\Theta_0$ we have $a_{\rm A}^{e'} = a_{\rm A}^{e''} = 1$, the Fourier coordinates $q'^{i}_{{\rm A},x_2,\ldots,x_n}$ and $q''^{i}_{{\rm A},x_2,\ldots,x_n}$ are equal. The result then follows by equation (\ref{eqn:displayedTreeParams}).
\end{proof}
\begin{remark}
It is clear that Lemma~\ref{lem:paramlemma} holds for any group-based model for an arbitrary group $G$. The only condition is that the parameters $a^e_0$ are $1$ for each edge $e$, where $0$ is the identity element of the group. Even if this condition is not satisfied, the above proof shows that we can take a factor of $(\delta a_0^{e'} + (1-\delta)a_0^{e''})$ out of $q_{0,x_2,\ldots,x_n}$. For networks with a single reticulation, this makes the Fourier coordinate essentially monomial (cf. the proof of \cite[Proposition~12]{gross2024dimensions}).
\end{remark}
\section{Restrictions}

In this section we give a fundamental result that we use to obtain identifiability results in later sections. Our main strategy is to reduce the question of identifiability to networks with a small number of leaves by looking at restricted networks. We begin with a definition.

\begin{definition}
	Let $\mathbb{C}^{4^n}$ be indexed by leaf-patterns for $n$ taxa. The $i^{\text{th}}$ \emph{marginalization map} $\mathrm{m}^n_i$ on $\mathbb{C}^{4^n}$  is the map
$$\mathrm{m}^n_i\,:\ \mathbb{C}^{4^n} \longrightarrow \mathbb{C}^{4^{n-1}}$$
given by marginalizing over the $i^{\text{th}}$ leaf, that is, the $(x_1,\ldots,x_{i-1},x_{i+1},\ldots,x_n)$-component of $\mathrm{m}^n_i(\mathbf{p})$ is given by
$$\big(\mathrm{m}^n_i(\mathbf{p})\big)_{x_1,\ldots, x_{i-1}, x_{i+1},\ldots, x_n} = \sum_{x_i \in \Omega} p_{x_1,\ldots, x_i,\ldots, x_n},$$
for $\mathbf{p}\in \mathbb{C}^{4^n}$ with components $p_{x_1,\ldots, x_n}$. When it is clear from context, we will drop the superscript $n$ and simply write $\mathrm{m}_i$.

For a subset of leaves $\mathcal{Y} = \{y_1,\ldots,y_k\}$, the \emph{marginalization map over} $\mathcal{Y}$ is the map
$$\mathrm{m}^n_\mathcal{Y}\,:\ \mathbb{C}^{4^n} \longrightarrow \mathbb{C}^{4^{n-k}}$$
given by $\mathrm{m}^n_\mathcal{Y} = \mathrm{m}^{n-k+1}_{y_k}\circ\cdots\circ\mathrm{m}^{n}_{y_1}$.
\end{definition}

It is clear that when $\mathbf{p} \in \Delta^{4^n-1}$ we have $\mathrm{m}^n_{\mathcal{Y}}(\mathbf{p})\in\Delta^{4^{n-k}-1}$.

\begin{remark}\label{rem:commute}
    The marginalization maps for two distinct leaves $i$ and $j$ commute, in the sense that $\mathrm{m}_i^{n-1}\circ\mathrm{m}_j^{n} = \mathrm{m}_j^{n-1}\circ\mathrm{m}_i^{n}$. Explicitly, we have
    \begin{align*}     
    (\mathrm{m}_i^{n-1}&\circ\mathrm{m}_j^{n}(\mathbf{p}))_{x_1,\ldots,x_{i-1},x_{i+1},\ldots,x_{j-1}x_{j+1},\ldots,x_n} \\
    &= \sum_{x_i\in \Omega}(\mathrm{m}_j^{n}(\mathbf{p}))_{x_1,\ldots,x_{j-1},x_{j+1},\ldots,x_n} \\ &= \sum_{x_i\in \Omega}\sum_{x_j\in \Omega}p_{x_1,\ldots,x_n} = \sum_{x_j\in \Omega}\sum_{x_i\in \Omega}p_{x_1,\ldots,x_n} \\
    &=\sum_{x_j\in \Omega}(\mathrm{m}_i^{n}(\mathbf{p}))_{x_1,\ldots,x_{i-1},x_{i+1},\ldots,x_n} \\
    &= (\mathrm{m}_j^{n-1}\circ\mathrm{m}_i^{n}(\mathbf{p}))_{x_1,\ldots,x_{i-1},x_{i+1},\ldots,x_{j-1}x_{j+1},\ldots,x_n},
    \end{align*}
    for all $x_k\in \Omega$ with $k\in[n]\setminus\{i,j\}$. Here we have assumed, without loss of generality, that $i<j$. It follows that the ordering of leaves in the subset $\mathcal{Y}$ does not matter in the definition of $\mathrm{m}^n_{\mathcal{Y}}$.
\end{remark}
Observe that the marginalization maps are linear, so that when we marginalize over a leaf in a phylogenetic network with a displayed tree model, we obtain the mixture of the marginal points from the displayed trees. The core result we will need says that marginalizing over a leaf in a tree gives us a point in the model of the tree in which that leaf has been pruned. This result is well-known, but we give a full proof in the appendix for completeness.

\begin{proposition}\label{prop:restriction}
	Let $\mathcal{T}$ be an unrooted binary phylogenetic tree on a set of taxa $\mathcal{X}$, and let $\mathcal{M}$ be the corresponding model for a multiplicatively closed substitution model. For a leaf $x\in\mathcal{X}$, let $\mathcal{T}' = \mathcal{T}|_{\mathcal{X}\setminus\{x\}}$ be the restricted tree obtained from $\mathcal{T}$ by pruning the leaf $x$ and suppressing the resulting degree two vertex, and let $\mathcal{M}'$ be the corresponding model for the same substitution model. If $\mathbf{p}\in\mathcal{M}$ then the point $\mathbf{p}'$ obtained from $\mathbf{p}$ by marginalizing over the leaf $x$ is in $\mathcal{M}'$. \xqed{\proofbox}
\end{proposition}
A natural consequence is that we can perform this restriction for many leaves at once.
\begin{corollary}\label{cor:restriction}
Let $\mathcal{T}$ be an unrooted binary phylogenetic tree on a set of taxa $\mathcal{X}$, and let $\mathcal{M}$ be the corresponding model for a multiplicatively closed substitution model. For a subset $\mathcal{Y}\subset\mathcal{X}$, let $\mathcal{T}' = \mathcal{T}|_{\mathcal{X}\setminus\mathcal{Y}}$ be the restricted tree of $\mathcal{T}$, obtained  by pruning all leaves in $\mathcal{Y}$ and suppressing any resulting degree two vertices, and let $\mathcal{M}'$ be the corresponding model for the same substitution model. If $\mathbf{p}\in\mathcal{M}$ then the point $\mathbf{p}'$ obtained from $\mathbf{p}$ by marginalizing over all leaves $y\in\mathcal{Y}$ is in $\mathcal{M}'$.
\end{corollary}
\begin{proof}
The proof is by induction on the size of the set $\mathcal{Y}$, with the base case $|\mathcal{Y}|=0$ being trivial. For the induction step, first observe that both marginalization and tree pruning are commutative operations (i.e., they can be done in any order to give the same result, see Remark~\ref{rem:commute}). Write $\mathcal{Y} = \mathcal{Y}_0\cup\{y\}$ for a subset $\mathcal{Y}_0 \subset\mathcal{Y}$ with $|\mathcal{Y}_0| = |\mathcal{Y}| - 1$. By the induction hypothesis, if $\mathbf{p}\in\mathcal{M}$, then marginalizing over all leaves in $\mathcal{Y}_0$ gives a point $\mathrm{m}_{\mathcal{Y}_0}(\mathbf{p}) = \mathbf{p}_0$ in the model $\mathcal{M}_0$ of the tree $\mathcal{T}_0$ obtained from $\mathcal{T}$ by pruning the leaves in $\mathcal{Y}_0$ (and suppressing any resulting degree two vertices). Now, by Proposition~\ref{prop:restriction}, marginalizing $\mathbf{p}_0$ over the leaf $y$ gives a point in the model $\mathcal{M}'$.
\end{proof}

\begin{corollary}
Let $\mathcal{T}_1$ and $\mathcal{T}_2$ be two binary, phylogenetic trees on a set of taxa $\mathcal{X}$, and let $\mathcal{M}_1$ and $\mathcal{M}_2$ be the corresponding models for a multiplicatively closed substitution model. Consider the models $\mathcal{M}_1'$ and $\mathcal{M}_2'$ on the restricted trees $\mathcal{T}_1' = \mathcal{T}_1|_{\mathcal{X}\setminus \mathcal{Y}}$ and $\mathcal{T}_2' = \mathcal{T}_2|_{\mathcal{X}\setminus \mathcal{Y}}$ for some subset  $\mathcal{Y}\subset\mathcal{X}$. If $\mathcal{M}_1'\cap\mathcal{M}_2' = \emptyset$, then $\mathcal{M}_1\cap\mathcal{M}_2 = \emptyset$.
\end{corollary}
\begin{proof}

Suppose not, and that there exists a point $\mathbf{p}\in\mathcal{M}_1\cap\mathcal{M}_2 $. Then by Corollary~\ref{cor:restriction} the marginal point $\mathrm{m}_{\mathcal{Y}}(\mathbf{p})$ is contained in $\mathcal{M}_1'\cap\mathcal{M}_2'$, a contradiction.
\end{proof}
Next we apply our results to phylogenetic networks.

\begin{theorem}\label{thm:restriction_network}
Let $\mathcal{N}_1$ and $\mathcal{N}_2$ be two phylogenetic networks on a set of taxa $\mathcal{X}$, and let $\mathcal{M}_1$ and $\mathcal{M}_2$ be the corresponding models for a multiplicatively closed substitution model.  Consider the models $\mathcal{M}_1'$ and $\mathcal{M}_2'$ on the restricted networks $\mathcal{N}_1' = \mathcal{N}_1|_{\mathcal{X}\setminus \mathcal{Y}}$ and $\mathcal{N}_2' = \mathcal{N}_2|_{\mathcal{X}\setminus \mathcal{Y}}$ for some subset  $\mathcal{Y}\subset\mathcal{X}$. If $\mathcal{M}_1'\cap\mathcal{M}_2' = \emptyset$, then $\mathcal{M}_1\cap\mathcal{M}_2 = \emptyset$.
\end{theorem}
\begin{proof}
Suppose not, and that we have a point $\mathbf{p} \in\mathcal{M}_1\cap\mathcal{M}_2$. For each network $\mathcal{N}_i$ for $i=1,2$ we can write $\mathbf{p}$ as a sum of points coming from the displayed trees $\mathcal{T}_j^i$ of $\mathcal{N}_i$ for a fixed set of parameters and mixing weights
$$ \mathbf{p} = \sum_{j}\delta^i_j \mathbf{p}^i_j,$$
where $\mathbf{p}^i_j$ is the point coming from tree $\mathcal{T}_j^i$.
Now restricting to the leaf set $\mathcal{X}\setminus\mathcal{Y}$ by marginalizing over the leaves in $\mathcal{Y}$ we obtain the point
$$\mathrm{m}_{\mathcal{Y}}(\mathbf{p}) = \sum_{j} \delta^i_j \mathrm{m}_{\mathcal{Y}}(\mathbf{p}^i_j ),$$
where  $\mathrm{m}_{\mathcal{Y}}(\mathbf{p}^i_j)$ is the marginal point from the tree $\mathcal{T}_j^i$, which lies in the model of the restricted tree $\mathcal{T}_j^i |_{\mathcal{X}\setminus\mathcal{Y}}$ by Proposition~\ref{prop:restriction}. Now, since the displayed trees of a restricted network are precisely the restrictions of the displayed trees of the original network \cite[Proposition~A.1]{frohn2026bounds}, it is clear that $\mathrm{m}_{\mathcal{Y}}(\mathbf{p})$ lies in the model $\mathcal{M}'_i$. Thus we have $\mathrm{m}_{\mathcal{Y}}(\mathbf{p}) \in \mathcal{M}_1'\cap\mathcal{M}_2'$, a contradiction.
\end{proof}

\section{Full Identifiability of Level-1 Phylogenetic Networks}\label{sec:level1identifiability}
In this section we give identifiability results on level-1 phylogenetic networks, restricting our attention to group-based models on the parameter space~$\Theta_0$, for which transition matrices are probabilistic and can be obtained by the usual continuous-time Markov process. This subset of transition matrices forms a Lie-Markov model~\cite{sumner2012lie}, and so is multiplicatively closed. We may therefore apply the restriction results from the previous section. In Section~\ref{sec:identifiability} we use Theorem~\ref{thm:restriction_network} to reduce arbitrarily large level-1 phylogenetic networks to the small cases dealt with in Sections~\ref{sec:trinet} and~\ref{sec:quarnets}.

\subsection{Trinet Inequality (Level-1)}\label{sec:trinet}
In \cite[Proposition~2.15]{englander2025identifiability} the authors showed that, for the JC substitution model, the intersection of a 3-star tree model and a 3-sunlet model on the restricted parameter set $\Theta_0$ is empty. They called this result the \lq trinet inequality\rq, since 3-sunlets are the only (strictly) level-1 trinets. Here, we extend the trinet inequality to the K2P and K3P models. For the JC and K2P models, we further generalize these inequalities to trinets of arbitrary level in Section~\ref{sec:levelktrinet}.

\begin{lemma}\label{lem:trinetk2p}
    Let $\mathcal{N}_1$ be the 3-star unrooted tree and $\mathcal{N}_2$ be a 3-sunlet semi-directed network, and let $\mathcal{M}_1$ and $\mathcal{M}_2$ be the associated models under the K2P substitution model with parameter values in $\Theta_0(\mathcal{N}_1)$ and $\Theta_0(\mathcal{N}_2)$ respectively. Then the polynomial
    $$ Q = q_{\rm AGG}q_{\rm GAG}q_{\rm CCA}^2 - q_{\rm AAA}q_{\rm GGA}q_{\rm TCG}^2$$
    is zero on $\mathcal{M}_1$ and strictly positive on $\mathcal{M}_2$. In particular, $\mathcal{M}_1\cap\mathcal{M}_2=\emptyset$.\xqed{\proofbox}
\end{lemma}

The proof of Lemma \ref{lem:trinetk2p} is analogous to the proof of \cite[Proposition~2.15]{englander2025identifiability} and given in the appendix. For the K3P model we could not find a single invariant like $Q$ that distinguishes between the two models. However, we can find several invariants for the 3-star tree that cannot be simultaneously zero for a point on the 3-sunlet model under the restricted parameters.

\begin{lemma}\label{lem:trinetk3p}
    Let $\mathcal{N}_1$ be the 3-star unrooted tree and $\mathcal{N}_2$ be a 3-sunlet semi-directed network, and let $\mathcal{M}_1$ and $\mathcal{M}_2$ be the associated models under the K3P substitution model with parameter values in $\Theta_0(\mathcal{N}_1)$ and $\Theta_0(\mathcal{N}_2)$ respectively. Then $\mathcal{M}_1\cap\mathcal{M}_2=\emptyset$.
\end{lemma}
\begin{proof}
    Consider the polynomials
    \begin{align*}
        Q_1 &= q_{\rm AAA}q_{\rm CTG}q_{\rm TGC} - q_{\rm AGG}q_{\rm CAC}q_{\rm TTA},\\
        Q_2 &= q_{\rm AAA}q_{\rm GTC}q_{\rm TCG} - q_{\rm ACC}q_{\rm GAG}q_{\rm TTA}, \\
        Q_3 &= q_{\rm ATT}q_{\rm CCA}q_{\rm TGC} - q_{\rm ACC}q_{\rm CGT}q_{\rm TTA}, \\
        Q_4 &= q_{\rm ACC}q_{\rm GGA}q_{\rm TAT} - q_{\rm AAA}q_{\rm GCT}q_{\rm TGC}.
    \end{align*}

    The reader can check that each of these polynomials lies in $I_{\mathcal{N}_1}$ by looking at the multisets of labels for each leaf position, as in the proof of Lemma \ref{lem:trinetk2p} (see appendix). We will show that for any point $x\in\mathcal{M}_2$ we cannot have $Q_i(x)=0$ for all $i=1,2,3,4$, and therefore $x\not\in\mathcal{M}_1$.

    As in the proof of Lemma~\ref{lem:trinetk2p}, we label the edges of the 3-sunlet as in Figure~\ref{fig:3sunlet}. For each edge $a$, we have the associated parameters $a_{\rm A}, a_{\rm C}, a_{\rm G}$, and $a_{\rm T}$, which are subject to the constraint $a_{\rm A} = 1$. We substitute a generic point $x$ into $Q$ using the parameterization of the model, which is given by
    $$q_{x_1x_2x_3}=a_{x_1}b_{x_2}c_{x_3}(\delta d_{x_3}f_{x_2} + (1-\delta)e_{x_3}f_{x_2+x_3}),$$
    for $x_1, x_2, x_3 \in G$ with $x_1+x_2+x_3 = 0$. We have
    \begin{align*}
        Q_1 &= q_{\rm AAA}q_{\rm CTG}q_{\rm TGC} - q_{\rm AGG}q_{\rm CAC}q_{\rm TTA} \\
        &=a_{\rm C}a_{\rm T}b_{\rm G}b_{\rm T}c_{\rm C}c_{\rm G}\delta(1-\delta)(f_{\rm C}f_{\rm G}-f_{\rm T})(d_{\rm C}e_{\rm G} - d_{\rm G}e_{\rm C}f_{\rm T})
    \end{align*}
    and
    \begin{align*}
        Q_2 &= q_{\rm AAA}q_{\rm GTC}q_{\rm TCG} - q_{\rm ACC}q_{\rm GAG}q_{\rm TTA} \\
        &= a_{\rm G}a_{\rm T}b_{\rm C}b_{\rm T}c_{\rm C}c_{\rm G}\delta(1-\delta)(f_{\rm C}f_{\rm G}-f_{\rm T})(d_{\rm G}e_{\rm C}-d_{\rm C}e_{\rm G}f_{\rm T}).
    \end{align*}
    Since all parameters lie in $(0,1)$, we have $Q_1(x)=0$ if and only if either $f_{\rm C}f_{\rm G}-f_{\rm T} = 0$, or $d_{\rm C}e_{\rm G} - d_{\rm G}e_{\rm C}f_{\rm T} = 0$. On the other hand $Q_2(x) = 0$ if and only if either $f_{\rm C}f_{\rm G}-f_{\rm T} = 0$ or $d_{\rm G}e_{\rm C}-d_{\rm C}e_{\rm G}f_{\rm T} = 0$. Now if $f_{\rm C}f_{\rm G}-f_{\rm T} \neq 0$ then 
    \begin{align*}
        d_{\rm C}e_{\rm G} - d_{\rm G}e_{\rm C}f_{\rm T} &= 0\quad \text{ and }\\
        d_{\rm G}e_{\rm C} - d_{\rm C}e_{\rm G}f_{\rm T} & = 0.
    \end{align*}
    Then $d_{\rm C}e_{\rm G}(1-f_{\rm T}^2) = 0$, but this is impossible since all parameters lie in $(0,1)$, so we must have $f_{\rm C}f_{\rm G}-f_{\rm T} = 0$.

    Next consider $Q_3$ and $Q_4$. We have 
        \begin{align*}
        Q_3 &= q_{\rm ATT}q_{\rm CCA}q_{\rm TGC} - q_{\rm ACC}q_{\rm CGT}q_{\rm TTA} \\
        &=a_{\rm C}a_{\rm T}b_{\rm C}b_{\rm G}b_{\rm T}c_{\rm C}c_{\rm T}\delta(1-\delta)(f_{\rm C}f_{\rm T}-f_{\rm G})(d_{\rm T}e_{\rm C}f_{\rm T}-d_{\rm C}e_{\rm T}f_{\rm C})
    \end{align*}
    and
    \begin{align*}
        Q_4 &= q_{\rm ACC}q_{\rm GGA}q_{\rm TAT} - q_{\rm AAA}q_{\rm GCT}q_{\rm TGC} \\
        &= a_{\rm G}a_{\rm T}b_{\rm C}b_{\rm G}c_{\rm C}c_{\rm T}\delta(1-\delta)(f_{\rm G}-f_{\rm C}f_{\rm T})(d_{\rm T}e_{\rm C}-d_{\rm C}e_{\rm T}f_{\rm G}).
    \end{align*}
    As before we have $Q_3(x)=0$ if and only if $f_{\rm C}f_{\rm T}-f_{\rm G}=0$ or $d_{\rm T}e_{\rm C}f_{\rm T}-d_{\rm C}e_{\rm T}f_{\rm C}=0$. Similarly $Q_4(x) = 0$ if and only if $f_{\rm C}f_{\rm T}-f_{\rm G}=0$ or $d_{\rm T}e_{\rm C}-d_{\rm C}e_{\rm T}f_{\rm G} = 0$. 

    Now if $f_{\rm C}f_{\rm T}-f_{\rm G}=0$, since we also have $f_{\rm C}f_{\rm G}-f_{\rm T} = 0$ we obtain $f_{\rm G}(f_{\rm C}^2 - 1) = 0$. But this is not possible since $f_{\rm C},f_{\rm G}\in(0,1)$. Then we must have $d_{\rm T}e_{\rm C}f_{\rm T}-d_{\rm C}e_{\rm T}f_{\rm C}=0$ and $d_{\rm T}e_{\rm C}-d_{\rm C}e_{\rm T}f_{\rm G} = 0$. But then $d_{\rm C}e_{\rm T}(f_{\rm G}f_{\rm T}-f_{\rm C}) = 0$ so $f_{\rm G}f_{\rm T}-f_{\rm C} = 0$. But combining this with $f_{\rm C}f_{\rm G}-f_{\rm T} = 0$ we get $f_{\rm C}(f_{\rm G}^2 - 1) = 0$. Again, this is impossible. Thus there are no points in $\mathcal{M}_2$ that simultaneously satisfy $Q_1, Q_2, Q_3$ and $Q_4$. It follows that $\mathcal{M}_1\cap\mathcal{M}_2=\emptyset$.
\end{proof}

\subsection{Level-1 Quarnets}\label{sec:quarnets}
In this section we determine identifiability results for level-1 semi-directed phylogenetic networks with four leaves, or \emph{level-1 quarnets}. Our first result shows that, for the K3P model on the parameter space $\Theta_0$, the three unrooted quartet trees are identifiable, extending \cite[Proposition~2.9]{englander2025identifiability}.

\begin{lemma}\label{lem:k3pquartets}
    Let $\mathcal{T}_1, \mathcal{T}_2$, and $\mathcal{T}_3$ be the three unrooted quartet trees with splits $12|34, 13|24$, and $14|23$ respectively. Let $\mathcal{M}_1, \mathcal{M}_2$, and $\mathcal{M}_3$ be the corresponding models under the K3P substitution model on the parameter sets $\Theta_0(\mathcal{N}_1)$, $\Theta_0(\mathcal{N}_2)$ and $\Theta_0(\mathcal{N}_3)$ respectively. Then the polynomial
    $$Q_{12|34}=q_{\rm AAAA}q_{\rm TTTT} - q_{\rm AATT}q_{\rm TTAA}$$
    evaluates to zero on $\mathcal{M}_1$ and is strictly positive on $\mathcal{M}_2$ and $\mathcal{M}_3$.
\end{lemma}
\begin{proof}
We substitute into $Q_{12|34}$ generic points $x$ using the parameterizations of the models. For $\mathcal{T}_1$ we have 
$$q_{x_1x_2x_3x_4} = a^1_{x_1}a^2_{x_2}a^3_{x_3}a^4_{x_4}a^5_{x_1 + x_2},$$
giving
\begin{align*}
    Q_{12|34}(x) &= (a^1_{\rm A}a^2_{\rm A}a^3_{\rm A}a^4_{\rm A}a^5_{\rm A})(a^1_{\rm T}a^2_{\rm T}a^3_{\rm T}a^4_{\rm T}a^5_{\rm A}) - (a^1_{\rm A}a^2_{\rm A}a^3_{\rm T}a^4_{\rm T}a^5_{\rm A})(a^1_{\rm T}a^2_{\rm T}a^3_{\rm A}a^4_{\rm A}a^5_{\rm A}) \\
    &= a^1_{\rm T}a^2_{\rm T}a^3_{\rm T}a^4_{\rm T} - a^1_{\rm T}a^2_{\rm T}a^3_{\rm T}a^4_{\rm T} \\
    & = 0,
\end{align*}
where we make the usual identification $a^i_{\rm A}=1$. On the other hand, for $\mathcal{T}_2$ we have
$$q_{g_1g_2g_3g_4} = a^1_{g_1}a^2_{g_2}a^3_{g_3}a^4_{g_4}a^5_{g_1 + g_3},$$
giving
\begin{align*}
    Q_{12|34}(x) &= (a^1_{\rm A}a^2_{\rm A}a^3_{\rm A}a^4_{\rm A}a^5_{\rm A})(a^1_{\rm T}a^2_{\rm T}a^3_{\rm T}a^4_{\rm T}a^5_{\rm A}) - (a^1_{\rm A}a^2_{\rm A}a^3_{\rm T}a^4_{\rm T}a^5_{\rm T})(a^1_{\rm T}a^2_{\rm T}a^3_{\rm A}a^4_{\rm A}a^5_{\rm T}) \\
    &=a^1_{\rm T}a^2_{\rm T}a^3_{\rm T}a^4_{\rm T}(1-(a^5_{\rm T})^2).
\end{align*}
Since $a^i_{\rm T} \in (0,1)$, we have that $Q_{12|34}(x)$ is strictly positive. Similarly, $Q_{12|34}$ is strictly positive on $\mathcal{M}_3$.
\end{proof}
\begin{remark}\label{rem:quartetpermutation}
    By permuting the leaf labels we can obtain invariants for $\mathcal{T}_2$ and $\mathcal{T}_3$ that are strictly positive on the models of the other quartet trees. We have
    $$Q_{13|24}=q_{\rm AAAA}q_{\rm TTTT} - q_{\rm ATAT}q_{\rm TATA}$$
    and
    $$Q_{14|23}=q_{\rm AAAA}q_{\rm TTTT} - q_{\rm ATTA}q_{\rm TAAT}.$$
    Thus we have that $\mathcal{M}_i\cap\mathcal{M}_j=\emptyset$ for all $i,j\in\{1,2,3\}$ with $i\neq j$.
\end{remark}

In \cite{englander2025identifiability} the authors use invariants for the quartet trees to distinguish between networks with different displayed quartets. The quartet tree invariants they found are linear, and so when applied to points in $\mathcal{M}_{\mathcal{N}}$, for some network $\mathcal{N}$ with 4 leaves, they apply linearly across the mixture of points from the displayed quartet models, and are therefore able to pick out the displayed quartets. Here however, our K3P invariants are quadratic (there are no linear invariants for K3P, because it is a general group-based model for $\mathbb{Z}/2\mathbb{Z}\times\mathbb{Z}/2\mathbb{Z}$), and so when applied to mixtures of points from the quartet tree models we get cross-terms that do not cancel. Thus, we cannot use them to pick out displayed quartets.

\begin{figure}[h!]
\centering
    \begin{subfigure}{0.23\linewidth}
        \centering
        \begin{tikzpicture}[scale = .4]

        \node[shape=circle,draw=black,scale=0.3] (A) at (0,0) {};
        \node[shape=circle,draw=black,scale=0.3] (B) at (0,3) {};
        \node[shape=circle,draw=black,scale=0.3,fill=black] (1) at (2,5) {};
        \node[shape=circle,draw=black,scale=0.3,fill=black] (2) at (-2,5) {};
        \node[shape=circle,draw=black,scale=0.3,fill=black] (3) at (2,-2) {};
        \node[shape=circle,draw=black,scale=0.3,fill=black] (4) at (-2,-2) {};
        \draw (A)--(B);
        \draw (1)--(B);
        \draw (2)--(B);
        \draw (3)--(A);
        \draw (4)--(A);
        
        \end{tikzpicture}
    \end{subfigure}
    \begin{subfigure}{0.23\linewidth}
        \centering
        \begin{tikzpicture}[scale = .35]
        
        \node[shape=circle,draw=black,scale=0.3] (A) at (0,0) {};
        \node[shape=circle,draw=black,scale=0.3] (B) at (0,3) {};
        \node[shape=circle,draw=black,scale=0.3] (C) at (1.8,5) {};
        \node[shape=circle,draw=black,scale=0.3] (D) at (-1.8,5) {};
        \node[shape=circle,draw=black,scale=0.3,fill=black] (1) at (3,6.5) {};
        \node[shape=circle,draw=black,scale=0.3,fill=black] (2) at (-3,6.5) {}; 
        \node[shape=circle,draw=black,scale=0.3,fill=black] (3) at (2,-2) {};
        \node[shape=circle,draw=black,scale=0.3,fill=black] (4) at (-2,-2) {};  
        
        \draw (A)--(B);
        \draw (B)--(C);
        \draw (B)--(D);
        \draw (D)--(C);
        \draw (C)--(1);
        \draw (D)--(2);
        \draw (A)--(3);
        \draw (A)--(4);

        \end{tikzpicture}      
    \end{subfigure}
    \begin{subfigure}{0.23\linewidth}
        \centering
        \begin{tikzpicture}[scale = .3]
        
        \node[shape=circle,draw=black,scale=0.3] (A) at (0,0) {};
        \node[shape=circle,draw=black,scale=0.3] (B) at (0,3) {};
        \node[shape=circle,draw=black,scale=0.3] (C) at (1.8,5) {};
        \node[shape=circle,draw=black,scale=0.3] (D) at (-1.8,5) {};
        \node[shape=circle,draw=black,scale=0.3] (E) at (1.8,-2) {};
        \node[shape=circle,draw=black,scale=0.3] (F) at (-1.8,-2) {};
        \node[shape=circle,draw=black,scale=0.3,fill=black] (1) at (3,6.5) {};
        \node[shape=circle,draw=black,scale=0.3,fill=black] (2) at (-3,6.5) {}; 
        \node[shape=circle,draw=black,scale=0.3,fill=black] (3) at (3,-3.5) {};
        \node[shape=circle,draw=black,scale=0.3,fill=black] (4) at (-3,-3.5) {};  
        
        \draw (A)--(B);
        \draw (B)--(C);
        \draw (B)--(D);
        \draw (D)--(C);
        \draw (C)--(1);
        \draw (D)--(2);
        
        \draw (A)--(E);
        \draw (A)--(F);
        \draw (E)--(F);       
        \draw (E)--(3);
        \draw (F)--(4);
        
        \end{tikzpicture}      
    \end{subfigure}
    \begin{subfigure}{0.23\linewidth}
        \centering
        \begin{tikzpicture}[scale = .7]

        \node[shape=circle,draw=black,scale=0.3] (A) at (0,4) {};
        \node[shape=circle,draw=black,scale=0.3] (B) at (1,3) {};
        \node[shape=circle,draw=black,scale=0.3] (C) at (-1,3) {};
        \node[shape=circle,draw=black,scale=0.3] (D) at (0,2) {};
        \node[shape=circle,draw=black,scale=0.3,fill=black] (1) at (0,5) {};
        \node[shape=circle,draw=black,scale=0.3,fill=black] (2) at (2,3) {}; 
        \node[shape=circle,draw=black,scale=0.3,fill=black] (3) at (-2,3) {};
        \node[shape=circle,draw=black,scale=0.3,fill=black] (4) at (0,1) {};

        \draw (1)--(A);
        \draw (A)--(B);
        \draw (A)--(C);
        \draw [dashed,->] (B)--(D);
        \draw [dashed,->] (C)--(D);
        \draw (B)--(2);
        \draw (C)--(3);
        \draw (D)--(4);
        
        \end{tikzpicture}
    \end{subfigure}
    
    \caption{The four level-1 quarnet topologies. From left to right: a quartet tree, a single triangle, a double triangle, a 4-sunlet. Note that reticulation edges are not identifiable in triangles. In the 4-sunlet, reticulation edges are dashed. Only the 4-sunlet does not have a non-trivial split.}
    \label{fig:4SunletAndTrees}
\end{figure}
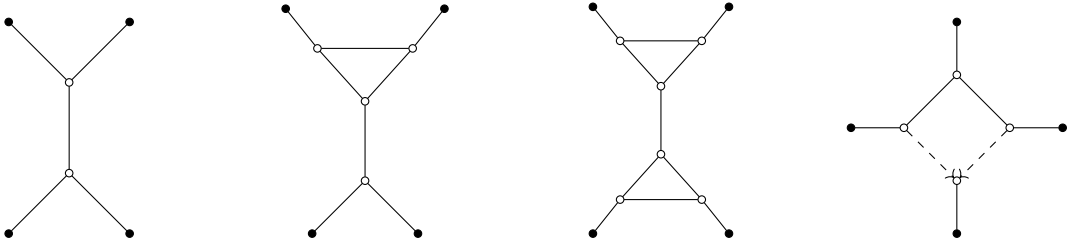

Given a quarnet, we are able to determine its topology (i.e. the underlying graph without leaf labels) by considering its induced trinets. We have four distinct level-1 quarnet topologies: the quartet tree, the quartet with a single triangle, the quartet with a double triangle, and the 4-sunlet (see Figure~\ref{fig:4SunletAndTrees}). In each case we have 4 ways of restricting to 3 leaves. The topologies of the induced trinets can be determined using the trinet inequalities from Section \ref{sec:trinet}, and each level-1 quarnet has a different induced trinet signature (see Table~\ref{tab:trinets}).

\begin{lemma}\label{lem:quarnettopologies}
    Let $\mathcal{N}_1$ and $\mathcal{N}_2$ be level-1 quarnets with different topologies, and let $\mathcal{M}_1$ and $\mathcal{M}_2$ be the corresponding models on the parameter sets $\Theta_0(\mathcal{N}_1)$ and $\Theta_0(\mathcal{N}_2)$ respectively, for either JC, K2P, or K3P. Then $\mathcal{M}_1\cap\mathcal{M}_2=\emptyset$.
\end{lemma}
\begin{proof}
    By Table~\ref{tab:trinets}, there exists a subset of 3 taxa $S$ for which $\mathcal{N}_1|_S$ and $\mathcal{N}_2|_S$ are a 3-star tree and a 3-sunlet (possibly not in that order). By the appropriate trinet inequality (\cite[Proposition~2.15]{englander2025identifiability}, Lemma~\ref{lem:trinetk2p}, or Lemma~\ref{lem:trinetk3p}) we have $\mathcal{M}_1|_S\cap\mathcal{M}_2|_S=\emptyset$. Then the result follows by Theorem~\ref{thm:restriction_network}.
\end{proof}

\begin{table}[h!]
    \centering
    \begin{tabular}{c|c|c|c|c}
        & Quartet Tree & Single Triangle & Double Triangle & 4-Sunlet \\
        \hline
    3-Star     & 4 & 2 & 0 & 1  \\
    3-Sunlet   & 0 & 2 & 4 & 3
    \end{tabular}
    \caption{The number of each type of induced trinet from the four level-1 quarnet topologies. Note that each column is distinct.}
    \label{tab:trinets}
\end{table}

Observe that the only level-1 quarnet without a non-trivial split is the 4-sunlet. Then by Lemma~\ref{lem:quarnettopologies}, we are able to determine whether or not a level-1 quarnet has a non-trivial split. By \cite[Proposition~2.9]{englander2025identifiability} and Lemma~\ref{lem:k3pquartets} we are able to determine the non-trivial split in a quartet tree. In our next result, we show that we can also determine the non-trivial split in a single or double triangle quarnet.

\begin{lemma}\label{lem:k3psplit}
    Let $\mathcal{N}_1$ be a level-1 quarnet with split $12|34$, and let $\mathcal{N}_2$ be a level-1 quarnet without the split $12|34$. Let $\mathcal{M}_1$ and $\mathcal{M}_2$ be the corresponding models for either the JC, K2P, or K3P models on the parameter sets $\Theta_0(\mathcal{N}_1)$ and $\Theta_0(\mathcal{N}_2)$ respectively. Then $\mathcal{M}_1\cap\mathcal{M}_2=\emptyset$.
\end{lemma}
\begin{proof}
   For the JC and K2P models, the result is a special case of \cite[Theorem~2.11]{englander2025identifiability}. We will prove the result for K3P. First, if $\mathcal{N}_1$ and $\mathcal{N}_2$ have different underlying topologies, then the result is given by Lemma~\ref{lem:quarnettopologies}, so we may assume that $\mathcal{N}_1$ and $\mathcal{N}_2$ have the same underlying topology. Since we are assuming that $\mathcal{N}_1$ has a non-trivial split, we may exclude the 4-sunlet topology.

   We show that the polynomial $Q_{12|34}=q_{\rm AAAA}q_{\rm TTTT} - q_{\rm AATT}q_{\rm TTAA}$ is an invariant for $\mathcal{N}_1$. We restrict our attention to the subgroup of $\mathbb{Z}/2\mathbb{Z}\times\mathbb{Z}/2\mathbb{Z}$ generated by  ${\rm A} = (0,0)$ and ${\rm T}=(1,1)$ (isomorphic to $\mathbb{Z}/2\mathbb{Z}$), and the subalgebra $\mathcal{C} = \mathbb{C}[q_{x_1x_2x_3x_4}\ |\ x_i \in \{{\rm A,T}\},\, x_1+x_2+x_3+x_4={\rm A}]$. The ideal for the general group-based model for $\mathbb{Z}/2\mathbb{Z}$, known as the Cavender-Farris-Neyman (CFN) model, on $\mathcal{N}_1$ embeds naturally into this subalgebra. Furthermore, we have that 
   $$I^{\text{CFN}}_{\mathcal{N}_1} \cong I^{\text{K3P}}_{\mathcal{N}_1}\cap\mathcal{C}.$$
   Thus, it is sufficient to show that $Q_{12|34}\in I^{\text{CFN}}_{\mathcal{N}_1}$. Consider the flattening matrix
   \[ \text{Flat}(12,34)= 
\begin{blockarray}{ccccc}
    & {\rm AA} & {\rm TT} & {\rm AT} & {\rm TA} \\
    \begin{block}{c(cccc)}
        {\rm AA} & q_{\rm AAAA} & q_{\rm AATT} & 0 & 0 \\
        {\rm TT} & q_{\rm TTAA} & q_{\rm TTTT} & 0 & 0 \\
        {\rm AT} & 0 & 0 & q_{\rm ATAT} & q_{\rm ATTA} \\
        {\rm TA} & 0 & 0 & q_{\rm TAAT} & q_{\rm TATA} \\
    \end{block}
\end{blockarray}\ .
\]
Since $\mathcal{N}_1$ has split $12|34$, a single edge disconnects $\mathcal{N}_1$ into two subgraphs with leaf sets $\{1,2\}$ and $\{3,4\}$ respectively. Then by \cite[Theorem~7.8]{sullivant2025phylogenetic}, for any point $p\in\mathcal{M}_1$ we have $\text{rank}\,\text{Flat}(12,34)(p) \leq 2$. It follows that each block diagonal has rank $1$. In particular, the determinant of the upper left block is zero, i.e., $Q_{12|34}(p)=0$.

It remains to show that for any point $p\in\mathcal{M}_2$ we have $Q_{12|34}(p) \neq 0$. We have already shown this for the quartet trees with splits $13|24$ and $14|23$ in the proof of Lemma~\ref{lem:k3pquartets}, so we may assume that $\mathcal{N}_2$ is either a single-triangle or double-triangle quarnet. There are 4 single-triangle quarnets and 2 double-triangle quarnets without the split $12|34$. We will show the result for one single-triangle quarnet, and defer the remaining cases to the appendix.

\begin{figure}[h!]

    \centering
    \begin{tikzpicture}[scale = .5]
        
        \node[shape=circle,draw=black,scale=0.3] (B) at (0,2) {};
        \node[shape=circle,draw=black,scale=0.3] (C) at (1.8,5) {};
        \node[shape=circle,draw=black,scale=0.3] (D) at (-1.8,5) {};
        \node[shape=circle,draw=black,scale=0.3] (E) at (0,0) {};
        \node[shape=circle,draw=black,scale=0.3,fill=black] (3) at (3,6.5) {};
        \node[shape=circle,draw=black,scale=0.3,fill=black] (1) at (-3,6.5) {}; 
        \node[shape=circle,draw=black,scale=0.3,fill=black] (4) at (3,-2) {};
        \node[shape=circle,draw=black,scale=0.3,fill=black] (2) at (-3,-2) {};
        \node at (3.2,7.2) {$3$};
        \node at (-3.2,7.2) {$1$};
        \node at (3.5,-2) {$4$};
        \node at (-3.5,-2) {$2$};

        \draw [->](B)--(C) node[midway,right=0.1] {$h$};
        \draw [dashed,->](B)--(D) node[midway,left=0.1] {$k$};
        \draw [dashed,->](C)--(D) node[midway,above=0.05] {$f$};
        \draw [->](C)--(3) node[midway,below right] {$c$};
        \draw [->](D)--(1) node[midway,below left] {$a$};
        \draw [->](B)--(E) node[midway,left=0.05] {$e$};
        \draw [->](E)--(2) node[midway,above left] {$b$};
        \draw [->](E)--(4) node[midway,above right] {$d$};

    \end{tikzpicture}
    \caption{A leaf and edge-labelled directed single-triangle quarnet. }
    \label{fig:singletri}
\end{figure}
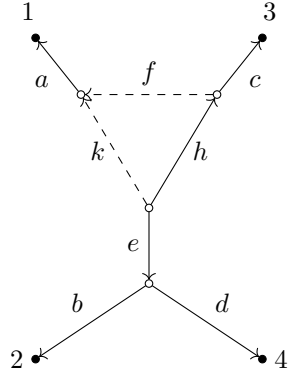

Let $\mathcal{N}_2$ be the single-triangle quarnet with split $13|24$ pictured in Figure~\ref{fig:singletri}, and label the edges of $\mathcal{N}_2$ therein. The parameterization of $\mathcal{N}_2$ is given by
$$q_{x_1x_2x_3x_4} = a_{x_1}b_{x_2}c_{x_3}d_{x_4}e_{x_1+x_3}(\delta f_{x_1}h_{x_1+x_3} + (1-\delta)k_{x_1}h_{x_3}).$$
Then for a point $p\in\mathcal{M}_2$ we have
\begin{align*}
    Q_{12|34}(p) &= a_{\rm T}b_{\rm T}c_{\rm T}d_{\rm T}(\delta f_{\rm T} + (1-\delta)k_{\rm T}h_{\rm T}) - c_{\rm T}d_{\rm T}e_{\rm T}h_{\rm T}\times a_{\rm T}b_{\rm T}e_{\rm T}(\delta f_{\rm T}h_{\rm T} + (1-\delta)k_{\rm T})  \\
    & = a_{\rm T}b_{\rm T}c_{\rm T}d_{\rm T}(\delta f_{\rm T} + (1-\delta)k_{\rm T}h_{\rm T}) - a_{\rm T}b_{\rm T}c_{\rm T}d_{\rm T}e_{\rm T}^2h_{\rm T}(\delta f_{\rm T}h_{\rm T} + (1-\delta)k_{\rm T}) \\
    &=a_{\rm T}b_{\rm T}c_{\rm T}d_{\rm T}(\delta f_{\rm T} + (1-\delta)k_{\rm T}h_{\rm T} - e_{\rm T}^2h_{\rm T}(\delta f_{\rm T}h_{\rm T} + (1-\delta)k_{\rm T})) \\
    &= a_{\rm T}b_{\rm T}c_{\rm T}d_{\rm T}(\delta f_{\rm T}(1-e_{\rm T}^2h_{\rm T}^2) + (1-\delta)h_{\rm T}k_{\rm T}(1-e_{\rm T}^2h_{\rm T})) > 0.
\end{align*}
\end{proof}

By permuting the leaves as in Remark \ref{rem:quartetpermutation}, given any two quarnets $\mathcal{N}_1$ and $\mathcal{N}_2$ with different splits, we have $\mathcal{M}_1\cap\mathcal{M}_2=\emptyset$. Thus, in this section we have so far shown that we can identify whether a quarnet is a 4-sunlet or has a split (Lemma~\ref{lem:quarnettopologies}), and if it has a split we can identify which one (Lemma~\ref{lem:k3psplit}). This is enough to identify the tree-of-blobs. In order to identify a semi-directed phylogenetic network, however, we must also be able to identify the blobs themselves. For level-1 phylogenetic networks, the blobs correspond to sunlet networks.

We now show that we can identify $n$-sunlets for $n\geq 4$. We start with the case $n=4$. For the JC and K2P substitution models, we already have the results to show this.
\begin{lemma}\label{lem:JCK2P4sunlet}
    Let $\mathcal{N}_1$ and $\mathcal{N}_2$ be two distinct $4$-sunlet networks. Let $\mathcal{M}_1$ and $\mathcal{M}_2$ be the corresponding models for either the JC or K2P substitution models on the parameter sets $\Theta_0(\mathcal{N}_1)$ and $\Theta_0(\mathcal{N}_2)$ respectively. Then $\mathcal{M}_1\cap\mathcal{M}_2=\emptyset$.
\end{lemma}
\begin{proof}
    First, assume that $\mathcal{N}_1$ and $\mathcal{N}_2$ have different leaves under the reticulation vertex. Without loss of generality, assume that $1$ is the leaf under the reticulation vertex for $\mathcal{N}_1$, and that leaves $1,2$ and $3$ are not under the reticulation vertex for $\mathcal{N}_2$. Restricting to the leaf set $\mathcal{Y}=\{1,2,3\}$ we have that $\mathcal{N}_1|_{\mathcal{Y}}$ is a 3-sunlet and $\mathcal{N}_2|_{\mathcal{Y}}$ is an unrooted 3-leaf tree. Let $\mathcal{M}_1'$ and $\mathcal{M}_2'$ be the corresponding models. By Lemma~\ref{lem:trinetk2p} we have $\mathcal{M}_1'\cap\mathcal{M}_2'=\emptyset$. Then the result follows by Theorem~\ref{thm:restriction_network}.

    Next, assume that $\mathcal{N}_1$ and $\mathcal{N}_2$ have the same leaf under the reticulation vertex. Then since $\mathcal{N}_1$ and $\mathcal{N}_2$ are distinct, they must have different displayed quartets. The result follows by \cite[Theorem~2.11]{englander2025identifiability}.
\end{proof}

Whilst we cannot use results on displayed trees for the K3P model, the result still holds.

\begin{lemma}\label{lem:k3p4sunlet}
    Let $\mathcal{N}_1$ and $\mathcal{N}_2$ be two distinct $4$-sunlet networks. Let $\mathcal{M}_1$ and $\mathcal{M}_2$ be the corresponding models for the K3P model on the parameter sets $\Theta_0(\mathcal{N}_1)$ and $\Theta_0(\mathcal{N}_2)$ respectively. Then $\mathcal{M}_1\cap\mathcal{M}_2=\emptyset$.
\end{lemma}
\begin{proof}
    In the case that $\mathcal{N}_1$ and $\mathcal{N}_2$ have different leaves under the reticulation vertex, we can restrict to a suitable set of three leaves as in the proof of Lemma \ref{lem:JCK2P4sunlet}. 

    We therefore assume that $\mathcal{N}_1$ and $\mathcal{N}_2$ have the same leaf under the reticulation vertex, and thus have different circular orders (see e.g. \cite{rhodes2025identifying}). We will assume that the leaf under the reticulation vertex is $1$. The other cases can be obtained by symmetry. Here, we will show that if $\mathcal{N}_1$ has circular order $(1,2,3,4)$ and $\mathcal{N}_2$ has circular order $(1,3,2,4)$, as in Figure~\ref{fig:4Sunlets}, then $\mathcal{M}_1\cap\mathcal{M}_2=\emptyset$. The remaining cases can be reduced to this one as follows. The case where $\mathcal{N}_1$ has circular order $(1,2,3,4)$ and $\mathcal{N}_2$ has circular order $(1,2,4,3)$ can be reduced to the above case by applying the permutation $(24)$ (in cycle notation). This leaves $\mathcal{N}_1$ unchanged (it merely changes the presentation of the graph) and changes $\mathcal{N}_2$ to the network in Figure~\ref{fig:4Sunlets}(b). The final case where $\mathcal{N}_1$ has circular order $(1,3,2,4)$ and $\mathcal{N}_2$ has circular order $(1,2,4,3)$ can be reduced to the above case by applying the permutation $(34)$.

    \begin{figure}[h!]
        \centering
        \begin{subfigure}[b]{0.4\linewidth}
            \centering
            \begin{tikzpicture}[scale = .9]
    
            \node[shape=circle,draw=black,scale=0.3] (A) at (0,4) {};
            \node[shape=circle,draw=black,scale=0.3] (B) at (1,3) {};
            \node[shape=circle,draw=black,scale=0.3] (C) at (-1,3) {};
            \node[shape=circle,draw=black,scale=0.3] (D) at (0,2) {};
            \node[shape=circle,draw=black,scale=0.3,fill=black] (1) at (0,5) {};
            \node[shape=circle,draw=black,scale=0.3,fill=black] (2) at (2,3) {}; 
            \node[shape=circle,draw=black,scale=0.3,fill=black] (3) at (-2,3) {};
            \node[shape=circle,draw=black,scale=0.3,fill=black] (4) at (0,1) {}; 

            \node (t1) at (0,5.4) {$3$};
            \node (t2) at (2.4,3) {$2$}; 
            \node (t3) at (-2.4,3) {$4$};
            \node (t4) at (0,0.6) {$1$};

            \draw (1)--(A);
            \draw (A)--(B);
            \draw (A)--(C);
            \draw [dashed,->] (B)--(D);
            \draw [dashed,->] (C)--(D);
            \draw (B)--(2);
            \draw (C)--(3);
            \draw (D)--(4);
            
            \end{tikzpicture}
            \caption{$\mathcal{N}_1$}
        \end{subfigure}
        \begin{subfigure}[b]{0.4\linewidth}
            \centering
            \begin{tikzpicture}[scale = .9]
    
            \node[shape=circle,draw=black,scale=0.3] (A) at (0,4) {};
            \node[shape=circle,draw=black,scale=0.3] (B) at (1,3) {};
            \node[shape=circle,draw=black,scale=0.3] (C) at (-1,3) {};
            \node[shape=circle,draw=black,scale=0.3] (D) at (0,2) {};
            \node[shape=circle,draw=black,scale=0.3,fill=black] (1) at (0,5) {};
            \node[shape=circle,draw=black,scale=0.3,fill=black] (2) at (2,3) {}; 
            \node[shape=circle,draw=black,scale=0.3,fill=black] (3) at (-2,3) {};
            \node[shape=circle,draw=black,scale=0.3,fill=black] (4) at (0,1) {}; 

            \node (t1) at (0,5.4) {$2$};
            \node (t2) at (2.4,3) {$3$}; 
            \node (t3) at (-2.4,3) {$4$};
            \node (t4) at (0,0.6) {$1$};

            \draw [->](A)--(1) node[midway,right] {$c$};
            \draw [->](A)--(B) node[midway,above right] {$f$};;
            \draw [->](A)--(C) node[midway,above left] {$k$};;
            \draw [dashed,->] (B)--(D) node[midway,below right] {$e$};;
            \draw [dashed,->] (C)--(D) node[midway,below left] {$h$};;
            \draw [->](B)--(2) node[midway,above] {$b$};;
            \draw [->](C)--(3) node[midway,above] {$d$};;
            \draw [->](D)--(4) node[midway,right] {$a$};;
            
            \end{tikzpicture}
            \caption{$\mathcal{N}_2$}
        \end{subfigure}   
        \caption{Two $4$-sunlets with circular orders \textbf{(a)} (1,2,3,4) and \textbf{(b)} (1,3,2,4). Note that $\mathcal{N}_2$ is directed and edge-labelled to determine the parameterization.}
        \label{fig:4Sunlets}
    \end{figure}
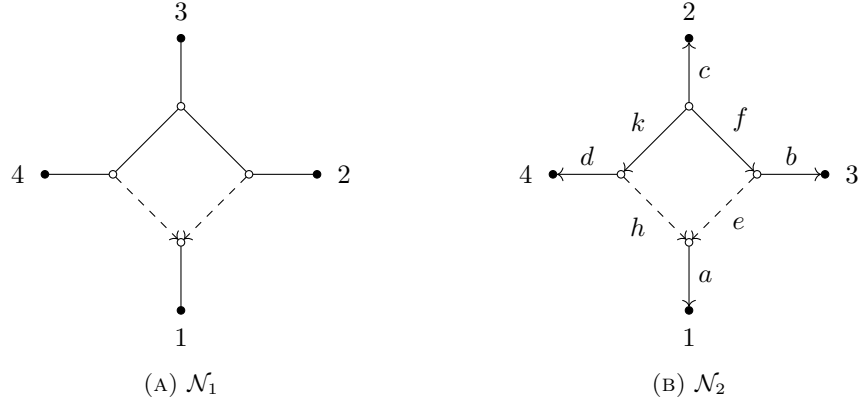

    Now, let $\mathcal{N}_1$ be the 4-sunlet with circular order $(1,2,3,4)$ and let $\mathcal{N}_2$ be the 4-sunlet with circular order $(1,3,2,4)$, as in Figure \ref{fig:4Sunlets}. Consider the polynomial
    \begin{equation}\label{eqn:4sunletQ}
        Q = q_{\rm AGAG}q_{\rm TCTC} - q_{\rm AGTC}q_{\rm TCAG} + q_{\rm ACAC}q_{\rm TGTG} - q_{\rm ACTG}q_{\rm TGAC}.
    \end{equation}
    The reader can check that $Q \in I_{\mathcal{N}_1}$. As usual, we substitute the parameterization of $\mathcal{N}_2$ into $Q$ and show that the resulting expression is strictly positive. The parameterization of $\mathcal{N}_2$ is given by
    \begin{equation*}
        q_{x_1x_2x_3x_4} = a_{x_1}b_{x_3}c_{x_2}d_{x_4}(\delta e_{x_1}f_{x_1+x_3}k_{x_4} + (1-\delta)f_{x_3}k_{x_1+x_4}h_{x_1}),
    \end{equation*}
    with variables corresponding to edges as shown in Figure \ref{fig:4Sunlets}.b. Thus we have
    \begin{align*}
        q_{\rm AGAG} &= c_{\rm G}d_{\rm G}k_{\rm G}, \\ 
        q_{\rm TCTC} &= a_{\rm T}b_{\rm T}c_{\rm C}d_{\rm C}(\delta e_{\rm T}k_{\rm C} + (1-\delta)f_{\rm T}k_{\rm G}h_{\rm T}), \\
        q_{\rm AGTC} &= b_{\rm T}c_{\rm G}d_{\rm C}f_{\rm T}k_{\rm C}, \\
        q_{\rm TCAG} &= a_{\rm T}c_{\rm C}d_{\rm G}(\delta e_{\rm T}f_{\rm T}k_{\rm G} + (1-\delta)k_{\rm C}h_{\rm T}), \\
        q_{\rm ACAC} &= c_{\rm C}d_{\rm C}k_{\rm C},\\
        q_{\rm TGTG} &= a_{\rm T}b_{\rm T}c_{\rm G}d_{\rm G}(\delta e_{\rm T}k_{\rm G} + (1-\delta)f_{\rm T}k_{\rm C}h_{\rm T}), \\
        q_{\rm ACTG} &= b_{\rm T}c_{\rm C}d_{\rm G}f_{\rm T}k_{\rm G}, \\
        q_{\rm TGAC} &= a_{\rm T}c_{\rm G}d_{\rm C}(\delta e_{\rm T}f_{\rm T}k_{\rm C} + (1-\delta)k_{\rm G}h_{\rm T})
    \end{align*}
    Substituting these into equation (\ref{eqn:4sunletQ}) gives
    \begin{align*}
        Q = &a_{\rm T}b_{\rm T}c_{\rm C}c_{\rm G}d_{\rm C}d_{\rm G}k_{\rm G}(\delta e_{\rm T}k_{\rm C} + (1-\delta)f_{\rm T}k_{\rm G}h_{\rm T}) \\
          &- a_{\rm T}b_{\rm T}c_{\rm C}c_{\rm G}d_{\rm C}d_{\rm G}f_{\rm T}k_{\rm C}(\delta e_{\rm T}f_{\rm T}k_{\rm G} + (1-\delta)k_{\rm C}h_{\rm T}) \\
          &+ a_{\rm T}b_{\rm T}c_{\rm C}c_{\rm G}d_{\rm C}d_{\rm G}k_{\rm C}(\delta e_{\rm T}k_{\rm G} + (1-\delta)f_{\rm T}k_{\rm C}h_{\rm T}) \\
          &- a_{\rm T}b_{\rm T}c_{\rm C}c_{\rm G}d_{\rm C}d_{\rm G}f_{\rm T}k_{\rm G}(\delta e_{\rm T}f_{\rm T}k_{\rm C} + (1-\delta)k_{\rm G}h_{\rm T}) \\ 
          = &2a_{\rm T}b_{\rm T}c_{\rm C}c_{\rm G}d_{\rm C}d_{\rm G}e_{\rm T}k_{\rm C}k_{\rm G}\delta(1-f_{\rm T}^2) > 0. 
    \end{align*}
    
\end{proof}

\subsection{Level-1 Phylogenetic Networks}\label{sec:identifiability}
Using the results developed in the previous two sections, we are now ready to prove our main result on the full identifiability of level-1 phylogenetic networks under the JC, K2P and K3P models.

\begin{theorem}\label{thm:level1identifiability}
    Let $\mathcal{N}_1$ and $\mathcal{N}_2$ be two level-1, semi-directed phylogenetic networks  on a set of $n$ taxa $\mathcal{X}$, with $n\geq 3$. Let $\mathcal{M}_1$ and $\mathcal{M}_2$ be the corresponding models for either the JC, K2P, or K3P substitution models on the  parameter sets $\Theta_0(\mathcal{N}_1)$ and $\Theta_0(\mathcal{N}_2)$ respectively. If $\mathcal{N}_1$ and $\mathcal{N}_2$ are distinct, modulo the placement of the reticulation vertices in any triangles, then $\mathcal{M}_1\cap\mathcal{M}_2=\emptyset$.
\end{theorem}
\begin{proof}
    We may assume that $n > 3$, otherwise the result follows directly from \cite[Proposition~2.15]{englander2025identifiability} (JC), Lemma~\ref{lem:trinetk2p} (K2P), or Lemma~\ref{lem:trinetk3p} (K3P).

    We will show that for $\mathcal{N}_1$ and $\mathcal{N}_2$ to be distinct (modulo the placement of the reticulation vertices in the triangles), there is a set~$\mathcal{Y} \subseteq \mathcal{X}$ with $|\mathcal{Y}| \in \{3, 4\}$ such that $\mathcal{M}'_1\cap\mathcal{M}'_2=\emptyset$, where $\mathcal{M}'_i$ is the model corresponding to $\mathcal{N}_i|_{\mathcal{Y}}$, $i \in \{1, 2\}$. The result then follows from Theorem~\ref{thm:restriction_network}.
    
    First suppose that $\mathcal{N}_1$ and $\mathcal{N}_2$ are distinct, modulo the placement of the reticulation vertices in any triangles \emph{and} contracting triangles to single vertices. It then follows from \cite[Theorem~14(b)]{frohn2025reconstructing}, that there is some $\mathcal{Y} \subseteq \mathcal{X}$ with $|\mathcal{Y}|=4$ such that one of the following holds:
    \begin{enumerate}
        \item $\mathcal{N}_1|_{\mathcal{Y}}$ and $\mathcal{N}_2|_{\mathcal{Y}}$ induce a different non-trivial split;
        \item one of $\mathcal{N}_1|_{\mathcal{Y}}$ and $\mathcal{N}_2|_{\mathcal{Y}}$ induces a non-trivial split and the other one does not;
        \item $\mathcal{N}_1|_{\mathcal{Y}}$ and $\mathcal{N}_2|_{\mathcal{Y}}$ both do not induce a non-trivial split; in particular, they are distinct 4-sunlets.
    \end{enumerate}
    In the first two cases, $\mathcal{M}'_1\cap\mathcal{M}'_2=\emptyset$ by Lemma~\ref{lem:k3psplit}. In the last case, this follows from Lemma~\ref{lem:JCK2P4sunlet} (JC and K2P) or Lemma~\ref{lem:k3p4sunlet} (K3P).

    It remains to consider the case where $\mathcal{N}_1$ and $\mathcal{N}_2$ are isomorphic when considered modulo the placement of the reticulation vertices in any triangles \emph{and} contracting triangles to single vertices, but they become distinct when only comparing them modulo the placement of the reticulation vertices in the triangles. Then, there is a triangle in one of $\mathcal{N}_1$ and $\mathcal{N}_2$ that corresponds to a single vertex in the other network. Thus, there is some set $\mathcal{Y}$ with $|\mathcal{Y}| = 3$, such that one of $\mathcal{N}_1$ and $\mathcal{N}_2$ is a 3-sunlet and the other is a 3-star unrooted tree. The result again follows from \cite[Proposition~2.15]{englander2025identifiability} (JC), Lemma~\ref{lem:trinetk2p} (K2P), or Lemma~\ref{lem:trinetk3p} (K3P).
\end{proof}

For the next result we make the following definitions. Let $\mathcal{N}_1$ and $\mathcal{N}_2$ be two semi-directed phylogenetic networks with $n$ leaves. A semi-directed phylogenetic network $\mathcal{N}$ is a \emph{maximal shared network} of $\mathcal{N}_1$ and $\mathcal{N}_2$ if it is displayed by both $\mathcal{N}_1$ and $\mathcal{N}_2$ and there is no other semi-directed phylogenetic network $\mathcal{N}'$ such that $\mathcal{N}'$ is also displayed by both $\mathcal{N}_1$ and $\mathcal{N}_2$ and $\mathcal{N} < \mathcal{N}'$. We denote the set of all such networks as ${\rm max}(\mathcal{N}_1,\mathcal{N}_2)$.

For a semi-directed phylogenetic network $\mathcal{N}$, let $\Theta_0^+(\mathcal{N})$ be the parameter set $\Theta_0(\mathcal{N})$ expanded to allow mixing parameters $\delta_i$ to be in the range $[0,1]$, and let $\mathcal{M}_{\mathcal{N}}^+=\psi_{\mathcal{N}}(\Theta_0^+)$. We observe that if $\mathcal{N}'$ is a displayed network of $\mathcal{N}$, then $\mathcal{M}_{\mathcal{N}'}^+\subset\mathcal{M}_{\mathcal{N}}^+$ (got by setting the corresponding mixing parameters to either 0 or 1 in the parameterization $\psi_{\mathcal{N}}$).

\begin{corollary}\label{cor:level1identifiability}
    Let $\mathcal{N}_1$ and $\mathcal{N}_2$ be two $n$-leaf, level-1, semi-directed phylogenetic networks for $n\geq 3$ (considered modulo the placement of the reticulation vertex in any triangles). Let $\mathcal{M}^+_1$ and $\mathcal{M}^+_2$ be the corresponding models for either the JC, K2P, or K3P substitution models on the respective parameter sets $\Theta_0^+(\mathcal{N}_1)$ and $\Theta_0^+(\mathcal{N}_2)$. Then 
    $$\mathcal{M}^+_1\cap\mathcal{M}^+_2 = \bigcup_{\mathcal{N}\in {\rm max}(\mathcal{N}_1,\mathcal{N}_2)}\mathcal{M}^+_{\mathcal{N}}$$
    where ${\rm max}(\mathcal{N}_1,\mathcal{N}_2)$ is the set of maximal shared displayed networks of $\mathcal{N}_1$ and $\mathcal{N}_2$.
\end{corollary}
\begin{proof}
    First, observe that if $\mathcal{N}_1=\mathcal{N}_2$ then ${\rm max}(\mathcal{N}_1,\mathcal{N}_2)=\mathcal{N}_1$ and so the result holds. We may therefore assume that $\mathcal{N}_1\neq\mathcal{N}_2$.
    
    Let $\mathcal{N}_j$ have $r_j$ reticulation vertices for $j=1,2$. We proceed by induction on ${\rm max}(r_1,r_2)$. For the base case, suppose ${\rm max}(r_1,r_2) = 0$, so that $r_1=r_2 =0$ and $\mathcal{N}_1$ and $\mathcal{N}_2$ are unrooted trees. Since neither $\mathcal{N}_1$ nor $\mathcal{N}_2$ has any non-trivial displayed networks and $\mathcal{N}_1\neq\mathcal{N}_2$, we have ${\rm max}(\mathcal{N}_1,\mathcal{N}_2) = \emptyset$. Since these trees are distinct, there exists a subset of 4 taxa $\mathcal{Y}$ such that $\mathcal{N}_1|_{\mathcal{Y}}$ and $\mathcal{N}_2|_{\mathcal{Y}}$ are distinct quartets. By either \cite[Proposition~2.9]{englander2025identifiability} (JC, K2P) or Remark~\ref{rem:quartetpermutation} (K3P) and Proposition~\ref{prop:restriction} we have $\mathcal{M}_1\cap\mathcal{M}_2 = \emptyset$. Now since neither $\mathcal{N}_1$ nor $\mathcal{N}_2$ has any reticulation vertices, we have $\Theta_0(\mathcal{N}_i) = \Theta_0^+(\mathcal{N}_i)$ for $i=1,2$. It follows that $\mathcal{M}_1^+\cap\mathcal{M}_2^+ = \emptyset$.

    For the induction step, suppose ${\rm max}(r_1,r_2) = r$ and assume, without loss of generality, that $r_1 = r$, so $r_2\leq r$. Let $v_1,\ldots, v_r$ be the reticulation vertices of $\mathcal{N}_1$, and let $\mathcal{N}^i_{11}$ and $\mathcal{N}^i_{12}$ be the displayed networks of $\mathcal{N}_1$ obtained by removing one of the pair of reticulation edges incident to $v_i$ for $i=1,\ldots, r$. For each reticulation vertex $v_i$, by Lemma \ref{lem:displayedNetworkParam} we can write the parameterization of $\mathcal{N}_1$ as
    $$\psi_{\mathcal{N}_1} = \delta_i\psi_{\mathcal{N}^i_{11}} + (1-\delta_i)\psi_{\mathcal{N}^i_{12}},$$ where $\delta_i$ is the mixing parameter associated to $v_i$. Then it is clear that
    \begin{equation}\label{eqn:m1plus}
    \mathcal{M}_1^+ = \mathcal{M}_1\cup\bigcup_{i = 1}^r\big((\mathcal{M}^i_{11})^+\cup(\mathcal{M}^i_{12})^+\big).
    \end{equation}
    Here, the $(\mathcal{M}^i_{1j})^+$ are the models of the displayed networks $\mathcal{N}^i_{1j}$ on the parameter set $\Theta^+_0(\mathcal{N}^i_{1j})$, viewed as a subset of $\Theta^+_0(\mathcal{N}_1)$. They are the points in $\mathcal{M}_1^+$ got when $\delta_i$ is 0 or 1. It follows that
    \begin{equation}\label{eqn:m1pluscapm2plus}
        \mathcal{M}_1^+\cap\mathcal{M}_2^+=(\mathcal{M}_1\cap\mathcal{M}_2^+) \cup \bigcup_{i = 1}^r\big(((\mathcal{M}^i_{11})^+\cap\mathcal{M}_2^+)\cup((\mathcal{M}^i_{12})^+\cap\mathcal{M}_2^+)\big).
    \end{equation}
    We claim that $\mathcal{M}_1\cap\mathcal{M}_2^+=\mathcal{M}_1\cap\mathcal{M}_2$. We will prove this by induction on $r_2$. For the base case $r_2 = 0$, we have that $\mathcal{N}_2$ is a tree and therefore $\mathcal{M}_2^+ =\mathcal{M}_2$.
    For the induction step (on~$r_2$) we write
    \begin{equation}\label{eqn:m2plus}
    \mathcal{M}_2^+ = \mathcal{M}_2\cup\bigcup_{i = 1}^{r_2}\big((\mathcal{M}^i_{21})^+\cup(\mathcal{M}^i_{22})^+\big),
    \end{equation}
    analogously to equation (\ref{eqn:m1plus}). Then we have
    $$\mathcal{M}_1\cap\mathcal{M}_2^+=(\mathcal{M}_1\cap\mathcal{M}_2)\cup\bigcup_{i=1}^{r_2}\big((\mathcal{M}_1\cap(\mathcal{M}_{21}^i)^+)\cup(\mathcal{M}_1\cap(\mathcal{M}_{22}^i)^+)\big).$$
    Now, each $\mathcal{N}_{2j}^i$ has $r_2-1$ reticulation vertices, for $i=1,\ldots, r_2$ and $j=1,2$. By induction then, we have $\mathcal{M}_1\cap(\mathcal{M}_{2j}^i)^+=\mathcal{M}_1\cap\mathcal{M}_{2j}^i$. Since $\mathcal{N}_{2j}^i$ has strictly fewer reticulation vertices than $\mathcal{N}_1$, they must be distinct. Then by Theorem \ref{thm:level1identifiability} we have $\mathcal{M}_1\cap\mathcal{M}_{2j}^i=\emptyset$. It follows that $\mathcal{M}_1\cap\mathcal{M}_2^+=\mathcal{M}_1\cap\mathcal{M}_2$, so the claim is proved.
    
    Returning to the main proof, we can now write equation (\ref{eqn:m1pluscapm2plus}) as
    $$\mathcal{M}_1^+\cap\mathcal{M}_2^+=(\mathcal{M}_1\cap\mathcal{M}_2) \cup \bigcup_{i = 1}^r\big(((\mathcal{M}^i_{11})^+\cap\mathcal{M}_2^+)\cup((\mathcal{M}^i_{12})^+\cap\mathcal{M}_2^+)\big).$$
    By Theorem \ref{thm:level1identifiability}, since $\mathcal{N}_1$ and $\mathcal{N}_2$ are distinct, we have $\mathcal{M}_1\cap\mathcal{M}_2=\emptyset$. It follows that 
    \begin{equation}\label{eqn:m1pluscapm2plus2}
        \mathcal{M}_1^+\cap\mathcal{M}_2^+=\bigcup_{i = 1}^r\big(((\mathcal{M}^i_{11})^+\cap\mathcal{M}_2^+)\cup((\mathcal{M}^i_{12})^+\cap\mathcal{M}_2^+)\big).
    \end{equation}
    We now consider two cases. First, suppose that $r_2<r_1=r$. Since for each $i=1,\ldots,r$ the networks $\mathcal{N}_{11}^i$ and $\mathcal{N}_{12}^i$ have $r-1$ reticulation vertices,  by the induction hypothesis we have
    $$ (\mathcal{M}^i_{11})^+\cap\mathcal{M}_2^+= \bigcup_{\mathcal{N}\in {\rm max}(\mathcal{N}_{11}^i,\mathcal{N}_2)}\mathcal{M}_{\mathcal{N}}^+$$
    and similarly for $\mathcal{N}^i_{12}$. It follows then that
    \begin{align*}
        \mathcal{M}_1^+\cap\mathcal{M}_2^+ &= \bigcup_{i=1}^r\Big(\big(\bigcup_{\mathcal{N}\in {\rm max}(\mathcal{N}_{11}^i,\mathcal{N}_2)}\mathcal{M}_{\mathcal{N}}^+\big)\cup\big(\bigcup_{\mathcal{N}\in {\rm max}(\mathcal{N}_{12}^i,\mathcal{N}_2)}\mathcal{M}_{\mathcal{N}}^+\big)\Big) \\
        &=\bigcup_{\mathcal{N}\in \cup_{i=1}^r({\rm max}(\mathcal{N}_{11}^i,\mathcal{N}_2)\cup{\rm max}(\mathcal{N}_{12}^i,\mathcal{N}_2))}\mathcal{M}_{\mathcal{N}}^+.
    \end{align*}
    Next we observe that since $\mathcal{N}_1\neq\mathcal{N}_2$, any shared maximal displayed network of $\mathcal{N}_1$ and $\mathcal{N}_2$ is strictly smaller than $\mathcal{N}_1$. It must therefore be a displayed network of either $\mathcal{N}_{11}^i$ or $\mathcal{N}_{12}^i$ for some $i\in\{1,\ldots,r\}$, so we have 
    $${\rm max}(\mathcal{N}_1, \mathcal{N}_2) \subseteq \bigcup_{i=1}^r\big({\rm max}(\mathcal{N}_{11}^i,\mathcal{N}_2)\cup{\rm max}(\mathcal{N}_{12}^i,\mathcal{N}_2)\big),$$
    and it follows that 
    $$ \bigcup_{\mathcal{N}\in {\rm max}(\mathcal{N}_1,\mathcal{N}_2)}\mathcal{M}^+_{\mathcal{N}} \subseteq \bigcup_{\mathcal{N}\in \cup_{i=1}^r({\rm max}(\mathcal{N}_{11}^i,\mathcal{N}_2)\cup{\rm max}(\mathcal{N}_{12}^i,\mathcal{N}_2))}\mathcal{M}_{\mathcal{N}}^+.$$
    On the other hand, if $\mathcal{N}\in{\rm max}(\mathcal{N}_{1j}^i,\mathcal{N}_2)$ for $j\in\{1,2\}$ then $\mathcal{N}$ is a displayed network of $\mathcal{N}_1$ and $\mathcal{N}_2$ so there must exist some $\mathcal{N}'\in {\rm max}(\mathcal{N}_{1},\mathcal{N}_2)$ with $\mathcal{N}\leq\mathcal{N}'$ (possibly $\mathcal{N}'=\mathcal{N}$). It follows then that
    $$\mathcal{M}_{\mathcal{N}}^+\subseteq\bigcup_{\mathcal{N}'\in {\rm max}(\mathcal{N}_1,\mathcal{N}_2)}\mathcal{M}^+_{\mathcal{N}'},$$
    for all $\mathcal{N}\in{\rm max}(\mathcal{N}_{1j}^i,\mathcal{N}_2)$ for $j\in\{1,2\}$, thereby proving the result for the case $r_2<r_1$.

    For the case that $r_1=r_2=r$, we substitute equation (\ref{eqn:m2plus}) into equation (\ref{eqn:m1pluscapm2plus2}) to obtain
    \begin{align*}
        \mathcal{M}_1^+\cap\mathcal{M}_2^+ = &\bigcup_{i_1=1}^r\Big[\big((\mathcal{M}_{11}^{i_1})^+\cap\mathcal{M}_2\big)\cup\big((\mathcal{M}_{12}^{i_1})^+\cap\mathcal{M}_2\big)\Big]\cup \\
        &\bigcup_{i_1=1}^r\bigcup_{i_2=1}^{r}\Big[\big((\mathcal{M}_{11}^{i_1})^+\cap(\mathcal{M}_{21}^{i_2})^+\big)\cup\big((\mathcal{M}_{11}^{i_1})^+\cap(\mathcal{M}_{22}^{i_2})^+\big) \\
        &\cup\big((\mathcal{M}_{12}^{i_1})^+\cap(\mathcal{M}_{21}^{i_2})^+\big)\cup\big((\mathcal{M}_{12}^{i_1})^+\cap(\mathcal{M}_{22}^{i_2})^+\big)\Big].
    \end{align*}
    By our previous claim we have that $(\mathcal{M}_{1j}^{i_1})^+\cap\mathcal{M}_2=\mathcal{M}_{1j}^{i_1}\cap\mathcal{M}_2$ for $i_1=1,\ldots,r$ and $j=1,2$. Now in each case $\mathcal{N}_{1j}^{i_1}$ has strictly fewer reticulation vertices than $\mathcal{N}_2$, and so these networks must be distinct. By Theorem \ref{thm:level1identifiability} it follows that $(\mathcal{M}_{1j}^{i_1})^+\cap\mathcal{M}_2=\emptyset$. Next we apply the induction hypothesis to all terms of the form $(\mathcal{M}_{1j_1}^{i_1})^+\cap(\mathcal{M}_{2j_2}^{i_2})^+$ (since in all cases the two networks have strictly fewer than $r$ reticulation vertices). We thus have
    $$(\mathcal{M}_{1j_1}^{i_1})^+\cap(\mathcal{M}_{2j_2}^{i_2})^+=\bigcup_{\mathcal{N}\in {\rm max}(\mathcal{N}_{1j_1}^{i_1},\mathcal{N}_{2j_2}^{i_2})}\mathcal{M}_{\mathcal{N}}^+.$$
    We finally observe that since $\mathcal{N}_1\neq\mathcal{N}_2$, any maximal shared displayed network of $\mathcal{N}_1$ and $\mathcal{N}_2$ is neither $\mathcal{N}_1$ nor $\mathcal{N}_2$. Furthermore, since both $\mathcal{N}_1$ and $\mathcal{N}_2$ have $r$ reticulation vertices, any maximal shared displayed network of $\mathcal{N}_1$ and $\mathcal{N}_2$ must have at most $r-1$ reticulation vertices, and is therefore a maximal displayed network of $\mathcal{N}_{1j_1}^{i_1}$ and $\mathcal{N}_{2j_2}^{i_2}$ for some  $i_1,i_2\in\{1,\ldots,r\}$ and $j_1,j_2\in\{1,2\}$. It follows that the set of maximal displayed networks of $\mathcal{N}_1$ and $\mathcal{N}_2$ is a subset of the union of the set of maximal displayed networks of $\mathcal{N}_{1j_1}^{i_1}$ and $\mathcal{N}_{2j_2}^{i_2}$ over all $i_1,i_2=1,\ldots,r$ and $j_1,j_2=1,2$, that is,
    $$ \bigcup_{\mathcal{N}\in {\rm max}(\mathcal{N}_1,\mathcal{N}_2)}\mathcal{M}^+_{\mathcal{N}} \subseteq \bigcup_{i_1,i_2,j_1,j_2}\bigcup_{\mathcal{N}\in {\rm max}(\mathcal{N}_{1j_1}^{i_1},\mathcal{N}_{2j_2}^{i_2})}\mathcal{M}_{\mathcal{N}}^+.$$
    As in the previous case, for any $\mathcal{N}\in{\rm max}(\mathcal{N}_{1j_1}^{i_1},\mathcal{N}_{2j_2}^{i_2})$ we also have
    $$\mathcal{M}_{\mathcal{N}}^+\subseteq\bigcup_{\mathcal{N}'\in {\rm max}(\mathcal{N}_1,\mathcal{N}_2)}\mathcal{M}^+_{\mathcal{N}'}$$
    thereby proving the result.
\end{proof}

\section{Tree-Network Distinguishability}\label{sec:levelktrinet}

In this section we extend our view to more general phylogenetic networks by removing the level-1 restriction. In particular, we extend the trinet inequality from Section~\ref{sec:trinet} to trinets of any level, under the JC model. We then prove a lemma that allows us to lift this result from trinets to networks with an arbitrary number of leaves, ultimately yielding a general tree-network distinguishability result: no phylogenetic tree and non-tree phylogenetic network (modulo suppressing 2-blobs) can produce identical leaf-pattern distributions. To this end, we first show that we can disregard certain substructures in a network called 2-sub-blobs~\cite{ane2024anomalous}.

\subsection{Suppressing 2-sub-blobs}
Let $\mathcal{N} = (V,E)$ be a semi-directed phylogenetic network, let $W \subseteq V$, and let $B = \mathcal{N}[W]$ be the subgraph induced by~$W$. We call $B$ a \emph{2-sub-blob} of $\mathcal{N}$ if 
\begin{enumerate}[label={(\roman*)}, noitemsep]
    \item $B$ is connected;
    \item $B$ contains no edge that is a cut-edge of $\mathcal{N}$;
    \item $W$ contains exactly two vertices that are adjacent to a vertex in $V \setminus W$ in $\mathcal{N}$.
\end{enumerate}
Observe that a $2$-blob is also a 2-sub-blob, but a 2-sub-blob need not necessarily be a 2-blob, as it can be contained in a $k$-blob, $k\geq 2$. See for example Figure~\ref{fig:2subblobs}.

A 2-sub-blob~$B$ of a semi-directed phylogenetic network $\mathcal{N}$ \emph{traps the root} if in any rooting of $\mathcal{N}$, the root is contained in the 2-sub-blob (or in a larger blob that is suppressed into the $2$-sub-blob when obtaining the semi-directed network).
%\todo{I think the root be contained ``in'' parallel edges that are suppressed when obtaining the semi-directed network.}

By \emph{suppressing} a 2-sub-blob~$B$, we mean contracting all vertices in $B$ to a single vertex, and then suppressing that vertex. Note that the mixed graph obtained by this operation is also a semi-directed network. See again Figure~\ref{fig:2subblobs}.

\begin{figure}[h!]
    \centering
    \begin{tikzpicture}[scale=0.27]
    \tikzstyle{node}=[circle, draw=black, fill=black, scale=0.25]
    \tikzstyle{edge}=[]
    \tikzstyle{special_edge}=[thick]
    \tikzstyle{dashed_edge}=[dashed]
    \tikzstyle{arc}=[-{Latex[scale=1.2]}]
    \tikzstyle{dashed_arc}=[-{Latex[scale=1.2]}, dashed]
    \tikzstyle{none}=[]
    \tikzstyle{medium_label}=[none,scale=1.0]
    \tikzstyle{internal_node}=[shape=circle,draw=black,scale=0.3]
    \tikzstyle{special_node}=[shape=circle,draw=black,scale=0.3, thick]
    \tikzstyle{leaf_node}=[shape=circle,draw=black,scale=0.3,fill=black] %0.275
    \tikzstyle{ret_arc}=[-{>[scale=.8]}, dashed]
    \tikzstyle{special_edge2}=[-{>[scale=.8]}, dashed, thick]
    \tikzstyle{small_label}=[scale=0.85]

	\begin{pgfonlayer}{nodelayer}
		\node [style={medium_label}] (70) at (12.75, 7.25) {$\mathcal{N}_2$};
		\node [style={internal_node}] (453) at (12.75, 14.25) {};
		\node [style={internal_node}] (454) at (17, 7.25) {};
		\node [style={internal_node}] (455) at (21.25, 14.25) {};
		\node [style={internal_node}] (456) at (19.5, 14.25) {};
		\node [style={internal_node}] (457) at (20.25, 12.75) {};
		\node [style={internal_node}] (458) at (17.75, 14.25) {};
		\node [style={internal_node}] (459) at (19.5, 11.5) {};
		\node [style={special_node}] (460) at (13.5, 12.75) {};
		\node [style={special_node}] (461) at (16, 8.75) {};
		\node [style={special_node}] (462) at (14.5, 9.5) {};
		\node [style={special_node}] (463) at (16, 10.5) {};
		\node [style={special_node}] (464) at (13.5, 11) {};
		\node [style={special_node}] (465) at (15, 12) {};
		\node [style={leaf_node}] (466) at (22.5, 15.5) {};
		\node [style={leaf_node}] (467) at (17, 5.75) {};
		\node [style={leaf_node}] (468) at (11.5, 15.5) {};
		\node [style={internal_node}] (487) at (-2.5, 14.25) {};
		\node [style={internal_node}] (488) at (1.75, 7.25) {};
		\node [style={internal_node}] (489) at (6, 14.25) {};
		\node [style={internal_node}] (490) at (4.25, 14.25) {};
		\node [style={internal_node}] (491) at (5, 12.75) {};
		\node [style={internal_node}] (492) at (2.5, 14.25) {};
		\node [style={internal_node}] (493) at (4.25, 11.5) {};
		\node [style={special_node}] (494) at (-1.75, 12.75) {};
		\node [style={special_node}] (495) at (0.75, 8.75) {};
		\node [style={special_node}] (496) at (-0.75, 9.5) {};
		\node [style={special_node}] (497) at (0.75, 10.5) {};
		\node [style={special_node}] (498) at (-1.75, 11) {};
		\node [style={special_node}] (499) at (-0.25, 12) {};
		\node [style={leaf_node}] (500) at (7.25, 15.5) {};
		\node [style={leaf_node}] (501) at (1.75, 5.75) {};
		\node [style={leaf_node}] (502) at (-3.75, 15.5) {};
		\node [style={small_label}] (503) at (11, 16) {$y$};
		\node [style={small_label}] (504) at (23.25, 16) {$z$};
		\node [style={small_label}] (505) at (17, 5) {$x$};
		\node [style={small_label}] (506) at (-4.25, 16) {$y$};
		\node [style={small_label}] (507) at (8, 16) {$z$};
		\node [style={small_label}] (508) at (1.75, 5) {$x$};
		\node [style={internal_node}] (509) at (28, 14.25) {};
		\node [style={internal_node}] (510) at (32.25, 7.25) {};
		\node [style={internal_node}] (511) at (36.5, 14.25) {};
		\node [style={internal_node}] (512) at (34.75, 14.25) {};
		\node [style={internal_node}] (513) at (35.5, 12.75) {};
		\node [style={internal_node}] (514) at (33, 14.25) {};
		\node [style={internal_node}] (515) at (34.75, 11.5) {};
		\node [style={leaf_node}] (522) at (37.75, 15.5) {};
		\node [style={leaf_node}] (523) at (32.25, 5.75) {};
		\node [style={leaf_node}] (524) at (26.75, 15.5) {};
		\node [style={small_label}] (525) at (26.25, 16) {$y$};
		\node [style={small_label}] (526) at (38.5, 16) {$z$};
		\node [style={small_label}] (527) at (32.25, 5) {$x$};
		\node [style={medium_label}] (528) at (29, 7.25) {$\mathcal{N}_3$};
		\node [style={medium_label}] (529) at (-2.5, 7.25) {$\mathcal{N}_1$};
	\end{pgfonlayer}
	\begin{pgfonlayer}{edgelayer}
		\draw [style={ret_arc}] (461) to (454);
		\draw [style={ret_arc}] (459) to (454);
		\draw [style={special_edge2}] (462) to (461);
		\draw [style={special_edge2}] (463) to (461);
		\draw [style={special_edge2}] (464) to (460);
		\draw [style={special_edge2}] (465) to (460);
		\draw [style={special_edge2}] (464) to (465);
		\draw [style={special_edge2}] (463) to (465);
		\draw [style={special_edge}] (462) to (464);
		\draw [style={special_edge}] (462) to (463);
		\draw [style=edge] (460) to (453);
		\draw [style=edge] (453) to (458);
		\draw [style=edge] (458) to (459);
		\draw [style={ret_arc}] (459) to (457);
		\draw [style={ret_arc}] (456) to (457);
		\draw [style={ret_arc}] (457) to (455);
		\draw [style={ret_arc}] (456) to (455);
		\draw [style=edge] (458) to (456);
		\draw [style=edge] (466) to (455);
		\draw [style=edge] (453) to (468);
		\draw [style=edge] (454) to (467);
		\draw [style={ret_arc}] (495) to (488);
		\draw [style={ret_arc}] (493) to (488);
		\draw [style={special_edge2}] (496) to (495);
		\draw [style={special_edge2}] (497) to (495);
		\draw [style={special_edge}] (496) to (498);
		\draw [style=edge] (494) to (487);
		\draw [style=edge] (487) to (492);
		\draw [style=edge] (492) to (493);
		\draw [style={ret_arc}] (493) to (491);
		\draw [style={ret_arc}] (490) to (491);
		\draw [style={ret_arc}] (491) to (489);
		\draw [style={ret_arc}] (490) to (489);
		\draw [style=edge] (492) to (490);
		\draw [style=edge] (500) to (489);
		\draw [style=edge] (487) to (502);
		\draw [style=edge] (488) to (501);
		\draw [style={special_edge2}] (496) to (497);
		\draw [style={special_edge2}] (499) to (497);
		\draw [style={special_edge2}] (494) to (498);
		\draw [style={special_edge2}] (499) to (498);
		\draw [style={special_edge}] (494) to (499);
		\draw [style={ret_arc}] (515) to (510);
		\draw [style=edge] (509) to (514);
		\draw [style=edge] (514) to (515);
		\draw [style={ret_arc}] (515) to (513);
		\draw [style={ret_arc}] (512) to (513);
		\draw [style={ret_arc}] (513) to (511);
		\draw [style={ret_arc}] (512) to (511);
		\draw [style=edge] (514) to (512);
		\draw [style=edge] (522) to (511);
		\draw [style=edge] (509) to (524);
		\draw [style=edge] (510) to (523);
		\draw [style={ret_arc}] (509) to (510);
	\end{pgfonlayer}
\end{tikzpicture}

    \caption{Three trinets with leaf set $\{x, y, z\}$. $\mathcal{N}_1$ and $\mathcal{N}_2$ both have a 2-sub-blob (in bold). The 2-sub-blob in $\mathcal{N}_2$ traps the root and the one in $\mathcal{N}_1$ does not. The trinet $\mathcal{N}_3$ is the one obtained from $\mathcal{N}_1$ or $\mathcal{N}_2$ by suppressing their respective 2-sub-blobs.}
    \label{fig:2subblobs}
\end{figure}
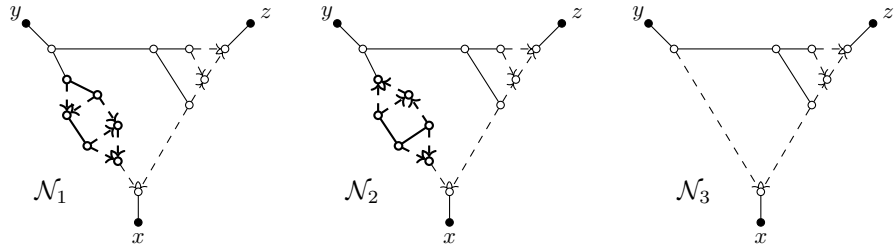

The following result follows by a straightforward generalization of the results in Section~5 of~\cite{sullivant2025phylogenetic} from 2-blobs to 2-sub-blobs. The arguments given there do not rely on the considered subgraph being a blob itself. In particular, they apply equally to a 2-sub-blob that is contained within a larger $k$-blob, for some $k \geq 2$. We note that our models and restricted parameter set $\Theta_0$ also satisfy the assumptions required for those results (multiplicative closure, closure under marginal convex combinations, time-reversibility, splittability, and being a semi-algebraic set; see also Section~\ref{sec:prelim}).
\begin{lemma}\label{lem:suppres_2sub-blobs}
    Let $\mathcal{N}$ be a semi-directed phylogenetic network with a 2-sub-blob~$B$ and let $\mathcal{N}'$ be the semi-directed phylogenetic network obtained from $\mathcal{N}$ by suppressing $B$. Let $\mathcal{M}$ and $\mathcal{M}'$ be the corresponding models under %either
    JC, K2P, or K3P %models
    on the parameter set $\Theta_0(\mathcal{N})$.
    \begin{enumerate}[label={(\roman*)}, noitemsep]
        \item If $B$ traps the root, then $\mathcal{M} \subseteq \mathcal{M}'$ and $\mathcal{M}, \mathcal{M}'$ have the same topological closure.
        \item If $B$ does not trap the root, then $\mathcal{M} = \mathcal{M}'$.
    \end{enumerate}
    \xqed{\proofbox}
\end{lemma}
It remains an open question whether part (i) can be strengthened to yield equality~\cite{sullivant2025phylogenetic}. However, the lemma suffices to prove the following corollary, which lets us consider only networks with suppressed 2-sub-blobs for identifiability purposes. See again Section~5 of \cite{sullivant2025phylogenetic} for the 2-blob analogue.

\begin{corollary}\label{cor:safely_suppres_2blobs}
    Let $\mathcal{N}_1$ and $\mathcal{N}_2$ be two $n$-leaf, semi-directed phylogenetic networks for $n\geq 3$ and let $\mathcal{N}'_1$ and $\mathcal{N}'_2$ be the corresponding networks obtained by suppressing all 2-sub-blobs. Let $\mathcal{M}_1$, $\mathcal{M}_2$, $\mathcal{M}'_1$, and $\mathcal{M}'_2$ be the corresponding models under
    %for either the 
    JC, K2P, or K3P %models
    on the  parameter sets $\Theta_0(\mathcal{N}_1)$ and $\Theta_0(\mathcal{N}_2)$ respectively. If $\mathcal{M}'_1\cap\mathcal{M}'_2=\emptyset$, then $\mathcal{M}_1\cap\mathcal{M}_2=\emptyset$.
\end{corollary}
\begin{proof}
By Lemma~\ref{lem:suppres_2sub-blobs}, if $\mathcal{N}'$ is obtained from a semi-directed network $\mathcal{N}$ by suppressing a single 2-sub-blob, then $\mathcal{M}_\mathcal{N} \subseteq \mathcal{M}_{\mathcal{N}'}$ and the two models have the same topological closure. Consequently, if a third network $\mathcal{N}''$ satisfies $\mathcal{M}_{\mathcal{N}'} \cap \mathcal{M}_{\mathcal{N}''} = \emptyset$, then also $\mathcal{M}_\mathcal{N} \cap \mathcal{M}_{\mathcal{N}''} = \emptyset$.

Applying this argument iteratively over all suppressed 2-sub-blobs yields the result.
\end{proof}

\subsection{Trinet Inequality (Level-$k$)}

We first prove the following lemma. The proof uses Lemma~1 and Theorem~1 of \cite{holtgrefe2026characterizing}, which provide a characterization of when a mixed graph is a semi-directed phylogenetic network. In particular, they imply that any semi-directed phylogenetic network either contains a leaf adjacent to a reticulation or contains a \emph{cherry}, that is, a pair of leaves adjacent to the same vertex.
\begin{lemma}\label{lem:knet_reticulation_leaf}
    Let $\mathcal{N}$ be a semi-directed phylogenetic network on $\mathcal{X}$ with $n = |\mathcal{X}|$. If $\mathcal{N}$ has a non-trivial $n$-blob, then there exists a reticulation vertex that is adjacent to a leaf.
\end{lemma}
\begin{proof}
    Since $\mathcal{N}$ has a non-trivial $n$-blob, the $n$ leaves of $\mathcal{N}$ are each adjacent to distinct vertices. In particular, $\mathcal{N}$ contains no cherry. Hence, by Lemma~1 and Theorem~1 of \cite{holtgrefe2026characterizing}, $\mathcal{N}$ must contain a leaf adjacent to a reticulation vertex.
\end{proof}

Before we prove a lemma which will be crucial in our inductive proof of the two main results of this section, we require some further terminology.

Let $\mathcal{N}$ be a semi-directed phylogenetic network with leaf set $\mathcal{X} = \{x, y, z\}$ and with a reticulation vertex~$r$ that is adjacent to leaf~$x \in \mathcal{X}$. Let $W \subseteq V$ be the set of all vertices on at least one up-down path between $y$ and~$z$; set $U = V \setminus W$. We call the induced subgraph $\mathcal{N}[W]$ the \emph{basin (of~$x$)} and---following~\cite{allman2025beyond}---the induced subgraph $\mathcal{N}[U]$ the \emph{funnel (of~$x$)}. The \emph{attachment points} of the basin are its vertices adjacent to a vertex in the funnel. See Figure~\ref{fig:funnel} for an example. Observe that $\mathcal{N} [W]$ is a semi-directed network with leaf set~$\{y, z\}$ in a slightly more general sense, since it may have degree-2 vertices. For the following proof, we extend the definition of 2-(sub)-blobs naturally to these slightly more general networks.

    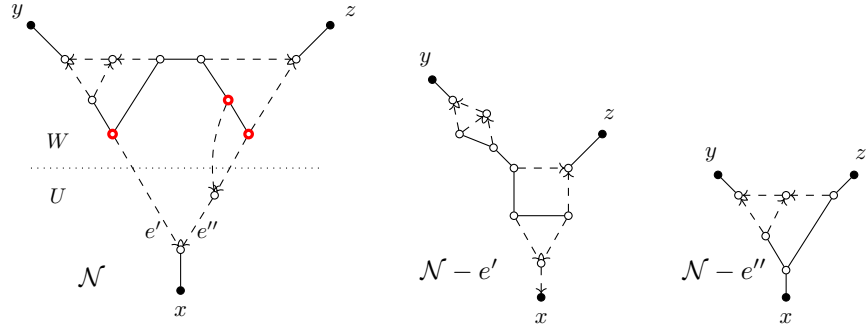
\begin{figure}[h!]
        \centering
    \begin{tikzpicture}[scale=0.36]
    \tikzstyle{node}=[circle, draw=black, fill=black, scale=0.25]
    \tikzstyle{edge}=[]
    \tikzstyle{special_edge}=[thick]
    \tikzstyle{special_edge3}=[dotted]
    \tikzstyle{dashed_edge}=[dashed]
    \tikzstyle{arc}=[-{Latex[scale=1.2]}]
    \tikzstyle{dashed_arc}=[-{Latex[scale=1.2]}, dashed]
    \tikzstyle{none}=[]
    \tikzstyle{medium_label}=[none,scale=1.0]
    \tikzstyle{main_label}=[none,scale=0.8]
    \tikzstyle{internal_node}=[shape=circle,draw=black,scale=0.3]
    \tikzstyle{special_node}=[shape=circle,draw=red,scale=0.3, very thick]
    \tikzstyle{leaf_node}=[shape=circle,draw=black,scale=0.3,fill=black] %0.275
    \tikzstyle{ret_arc}=[-{>[scale=.8]}, dashed]
    \tikzstyle{special_edge2}=[-{>[scale=.8]}, dashed, thick]
    \tikzstyle{small_label}=[scale=0.85]

	\begin{pgfonlayer}{nodelayer}
		\node [style={internal_node}] (509) at (27.75, 14.5) {};
		\node [style={internal_node}] (510) at (32, 7.5) {};
		\node [style={internal_node}] (511) at (36.25, 14.5) {};
		\node [style={internal_node}] (512) at (32.75, 14.5) {};
		\node [style={special_node}] (513) at (34.5, 11.75) {};
		\node [style={internal_node}] (515) at (33.25, 9.5) {};
		\node [style={leaf_node}] (522) at (37.5, 15.75) {};
		\node [style={leaf_node}] (523) at (32, 6) {};
		\node [style={leaf_node}] (524) at (26.5, 15.75) {};
		\node [style={small_label}] (525) at (26, 16.25) {$y$};
		\node [style={small_label}] (526) at (38.25, 16.25) {$z$};
		\node [style={small_label}] (527) at (32, 5.25) {$x$};
		\node [style={internal_node}] (530) at (29.5, 14.5) {};
		\node [style={internal_node}] (531) at (31.25, 14.5) {};
		\node [style={internal_node}] (532) at (28.75, 13) {};
		\node [style={special_node}] (533) at (29.5, 11.75) {};
		\node [style={special_node}] (534) at (33.75, 13) {};
		\node [style={small_label}] (535) at (31, 8.25) {$e'$};
		\node [style={small_label}] (536) at (33, 8.25) {$e''$};
		\node [style={internal_node}] (537) at (42, 13) {};
		\node [style={internal_node}] (539) at (46.25, 10.5) {};
		\node [style={internal_node}] (540) at (44.25, 10.5) {};
		\node [style={internal_node}] (541) at (46.25, 8.75) {};
		\node [style={internal_node}] (542) at (45.25, 7) {};
		\node [style={leaf_node}] (543) at (47.5, 11.75) {};
		\node [style={leaf_node}] (544) at (45.25, 5.75) {};
		\node [style={leaf_node}] (545) at (41.25, 13.75) {};
		\node [style={small_label}] (546) at (41, 14.5) {$y$};
		\node [style={small_label}] (547) at (47.75, 12.5) {$z$};
		\node [style={small_label}] (548) at (45.25, 5) {$x$};
		\node [style={internal_node}] (549) at (43.25, 12.5) {};
		\node [style={internal_node}] (550) at (43.5, 11.25) {};
		\node [style={internal_node}] (551) at (42.25, 11.75) {};
		\node [style={internal_node}] (553) at (44.25, 8.75) {};
		\node [style={internal_node}] (554) at (52.5, 9.5) {};
		\node [style={leaf_node}] (560) at (56.75, 10.25) {};
		\node [style={leaf_node}] (561) at (54.25, 5.75) {};
		\node [style={leaf_node}] (562) at (51.75, 10.25) {};
		\node [style={small_label}] (563) at (51.5, 11) {$y$};
		\node [style={small_label}] (564) at (57, 11) {$z$};
		\node [style={small_label}] (565) at (54.25, 5) {$x$};
		\node [style={internal_node}] (566) at (54.25, 9.5) {};
		\node [style={internal_node}] (567) at (56, 9.5) {};
		\node [style={internal_node}] (568) at (53.5, 8) {};
		\node [style={internal_node}] (569) at (54.25, 6.75) {};
		\node [style={medium_label}] (570) at (28.75, 6.5) {$\mathcal{N}$};
		\node [style={medium_label}] (571) at (42.25, 6.75) {$\mathcal{N} - e'$};
		\node [style={medium_label}] (572) at (52, 6.75) {$\mathcal{N} - e''$};
		\node [style=none] (573) at (26.5, 10.5) {};
		\node [style=none] (574) at (37.25, 10.5) {};
		\node [style={main_label}] (575) at (27.5, 11.5) {$W$};
		\node [style={main_label}] (576) at (27.5, 9.5) {$U$};
	\end{pgfonlayer}
	\begin{pgfonlayer}{edgelayer}
		\draw [style={ret_arc}] (515) to (510);
		\draw [style={ret_arc}] (513) to (511);
		\draw [style={ret_arc}] (512) to (511);
		\draw [style=edge] (522) to (511);
		\draw [style=edge] (509) to (524);
		\draw [style=edge] (510) to (523);
		\draw [style={ret_arc}] (532) to (509);
		\draw [style={ret_arc}] (530) to (509);
		\draw [style={ret_arc}] (532) to (530);
		\draw [style={ret_arc}] (531) to (530);
		\draw [style=edge] (532) to (533);
		\draw [style=edge] (533) to (531);
		\draw [style={ret_arc}] (533) to (510);
		\draw [style={ret_arc}] (513) to (515);
		\draw [style={ret_arc}, bend right=15] (534) to (515);
		\draw [style=edge] (513) to (534);
		\draw [style=edge] (534) to (512);
		\draw [style=edge] (512) to (531);
		\draw [style={ret_arc}] (541) to (539);
		\draw [style={ret_arc}] (540) to (539);
		\draw [style=edge] (543) to (539);
		\draw [style=edge] (537) to (545);
		\draw [style={ret_arc}] (551) to (537);
		\draw [style={ret_arc}] (549) to (537);
		\draw [style={ret_arc}] (551) to (549);
		\draw [style={ret_arc}] (550) to (549);
		\draw [style={ret_arc}] (541) to (542);
		\draw [style={ret_arc}] (553) to (542);
		\draw [style=edge] (541) to (553);
		\draw [style=edge] (553) to (540);
		\draw [style=edge] (540) to (550);
		\draw [style=edge] (551) to (550);
		\draw [style={ret_arc}] (542) to (544);
		\draw [style=edge] (554) to (562);
		\draw [style={ret_arc}] (568) to (554);
		\draw [style={ret_arc}] (566) to (554);
		\draw [style={ret_arc}] (568) to (566);
		\draw [style={ret_arc}] (567) to (566);
		\draw [style=edge] (568) to (569);
		\draw [style=edge] (569) to (567);
		\draw [style=edge] (560) to (567);
		\draw [style=edge] (561) to (569);
		\draw [style={special_edge3}] (573.center) to (574.center);
	\end{pgfonlayer}
\end{tikzpicture}
        \caption{\emph{Left:} A level-5 trinet~$\mathcal{N}$ with the vertices~$W$ in the basin for~$x$ and the vertices~$U$ in the funnel for~$x$. The three thick red vertices are the attachment points of the basin. \emph{Right:} The two displayed networks $\mathcal{N}''=\mathcal{N} - e'$ and $\mathcal{N}'=\mathcal{N}-e''$.}
        \label{fig:funnel}
    \end{figure}

\begin{lemma}\label{lem:trinet_display_3blob}
    Let $\mathcal{N}$ be a trinet with a non-trivial 3-blob. Let $r$ be a reticulation vertex that is adjacent to a leaf and let $e'$ and $e''$ be the corresponding reticulation edges. If $\mathcal{N}$ is strict level-$k$, $k\geq 2$, and it has no 2-sub-blobs, then at least one of $\mathcal{N} - e'$ and $\mathcal{N} - e''$ has a non-trivial 3-blob.
\end{lemma}
\begin{proof}
Let $x$ be the leaf adjacent to $r$. Let $y$ and $z$ denote the other two leaves of $\mathcal{N} = (V, E)$. Set $\mathcal{N}' := \mathcal{N} - e''$ and $\mathcal{N}'' := \mathcal{N} - e'$, where $e' = (u', r)$ and $e'' = (u'', r)$. Let $W \subseteq V$ be the vertices in the basin of~$x$ and $U = V\setminus W$ the vertices in the funnel of~$x$. See also Figure~\ref{fig:funnel}. We now distinguish two cases.

\textbf{Case 1:} \emph{$W$ contains a reticulation vertex of $\mathcal{N}$.} In this case, the basin $\mathcal{N}[W]$ has at least one non-trivial 2-blob~$Q$. We claim that $Q$ has at least one attachment point $a$. Otherwise, $Q$ would form a 2-sub-blob in $\mathcal{N}$, a contradiction. Let $\pi$ be an up-down path between $a$ and $x$ in $\mathcal{N}[U\cup\{a\}]$; assume without loss of generality
that it uses $e'$. Now consider $\mathcal{N}'$ (which retains $e'$) and any up-down path~$P$ between~$y$ and~$z$ that contains at least one vertex (and hence at least one edge) of~$Q$. Then~$P$ contains vertex-disjoint up-down paths between a vertex of~$Q$ and~$y$ and between a different vertex of~$Q$ and~$z$. Together with~$\pi$, this gives
vertex-disjoint up-down paths from $Q$ to all three leaves. Thus, $Q$ is part of some non-trivial 3-blob~$B$ of~$\mathcal{N}'$. (Note that~$B$ contains~$Q$ but is not necessarily equal to~$Q$.)

\textbf{Case 2:} \emph{$W$ contains no reticulation vertex of $\mathcal{N}$.} In this case $\mathcal{N}[W]$ is a path between~$y$ and~$z$. Since~$k \geq 2$, there exist at least two reticulation vertices in the funnel. Consider a lowest reticulation not equal to~$r$. This must be a parent of~$r$, say~$u'$ (without loss of generality).

First suppose that there exist at least two attachment points with semi-directed paths to~$u'$, not using any edges of the basin. Then the blob of $\mathcal{N}'$ containing $u'$ is a non-trivial $3$-blob.

Now suppose that there exists only one attachment point~$a$ with a semi-directed path to~$u'$, not traversing any edges of the basin. Let~$A$ be the union of all such paths. Since~$A$ is not a $2$-sub-blob, there is a semi-directed path~$P_a$ from a vertex of~$A$ to~$r$ traversing~$e''$.

Now observe that there exists a different attachment point~$b\neq a$ because otherwise the funnel edge incident to~$a$ would be a cut edge and hence there would be a $2$-blob. Moreover, there exists at least one semi-directed path~$P_b$ from~$b$ to~$r$ traversing~$e''$ but no basin edges. The first vertex that is on both paths~$P_a,P_b$ is a reticulation vertex~$r^*$. It follows that, in this case the blob of $\mathcal{N}''$ containing $r^*$ is a non-trivial $3$-blob.
\end{proof}

We now give the level-$k$ trinet inequality %considering first
under the JC model. Note that for the case where $\mathcal{N}_2$ is level-2, the following lemma was proved in \cite{englander2025identifiability}. Our generalization positively answers their Conjecture~2.16.
\begin{lemma}\label{lem:trinetJC_levk}
    Let $\mathcal{N}_1$ be a trinet without a non-trivial 3-blob and $\mathcal{N}_2$ a trinet with a non-trivial 3-blob, and let $\mathcal{M}_1$ and $\mathcal{M}_2$ be the associated models under the JC evolutionary model with parameter values in $\Theta_0(\mathcal{N}_1)$ and $\Theta_0(\mathcal{N}_2)$ respectively. Then %\correction{there exists an ordering of the leaves for which}
    the polynomial
    $$ Q = q_{\rm ACC}q_{\rm CAC}q_{\rm CCA} - q_{\rm AAA} q_{\rm TCG}^2$$
    is zero on $\mathcal{M}_1$ and strictly positive on $\mathcal{M}_2$. Thus $\mathcal{M}_1\cap\mathcal{M}_2=\emptyset$.
\end{lemma}
\begin{proof}
    For simplicity, we explicitly make the substitution $q_{\rm AAA} = 1$. By Corollary~\ref{cor:safely_suppres_2blobs}, we may assume that $\mathcal{N}_1$ has no 2-sub-blobs, and in particular no 2-blobs, so it is an unrooted 3-leaf tree (i.e., a 3-star tree). It was shown in \cite[Proposition~2.15]{englander2025identifiability} that $Q$ evaluates to zero on $\mathcal{M}_1$. For completeness, if $\mathcal{N}_1$ is a 3-star tree with pendant edges $a$, $b$ and $c$, then under the JC Fourier parameterization $q_{\rm TCG}=a_{\rm T} b_{\rm C} c_{\rm G} = a_{\rm C} b_{\rm C} c_{\rm C}$, $q_{\rm CCA}=a_{\rm C} b_{\rm C}$, $q_{\rm CAC}=a_{\rm C} c_{\rm C}$, and $q_{\rm ACC}=b_{\rm C} c_{\rm C}$. Under the JC model, we have $e_{\rm C} = e_{\rm T} = e_{\rm G}$ for every edge~$e$, so $Q = 0$ on $\mathcal{M}_1$.

    Let $x$ be a leaf adjacent to a reticulation vertex $r$ of $\mathcal{N}_2$, with reticulation edges $e'$ and $e''$, whose existence is guaranteed by Lemma~\ref{lem:knet_reticulation_leaf}.
    We also assume that the leaf~$x$ corresponds to the last position in the indices of the Fourier coordinates. We can assume this without loss of generality because, in the JC model, the polynomial $Q$ is invariant under leaf permutations (recall that for JC all leaf permutations of $q_{\rm TCG}$ all lie in the same equivalence class). Let~$\delta'$ be the mixing parameter corresponding to $e'$ and let $\delta'' := 1-\delta'$ be the mixing parameter corresponding to $e''$. Define $\mathcal{N}'_2 := \mathcal{N}_2 - e''$ and $\mathcal{N}''_2 := \mathcal{N}_2 - e'$. See Figure~\ref{fig:funnel} for an example. By Corollary~\ref{cor:safely_suppres_2blobs}, we may assume that $\mathcal{N}_2$, $\mathcal{N}'_2$, and $\mathcal{N}''_2$ contain no 2-sub-blobs.

    Let
    \[
    Q' = q'_{\rm ACC} q'_{\rm CAC} q'_{\rm CCA} - (q'_{\rm TCG})^2,
    \qquad
    Q'' = q''_{\rm ACC} q''_{\rm CAC} q''_{\rm CCA} - (q''_{\rm TCG})^2
    \]
    be the corresponding polynomials for $\mathcal{N}'_2$ and $\mathcal{N}''_2$. Each Fourier coordinate decomposes as
    $q_{uvw} = \delta' q'_{uvw} + \delta'' q''_{uvw}$.
    Moreover, by Lemma \ref{lem:paramlemma} we have $q_{\rm CCA} = q'_{\rm CCA} = q''_{\rm CCA}$.
    Using this and expanding $Q$ gives
    \begin{align}\label{eq:trinetJC:decomp}
    Q
    &= q_{\rm ACC} q_{\rm CAC} q_{\rm CCA} - q_{\rm TCG}^2 \\
    &= (\delta' q'_{\rm ACC} + \delta'' q''_{\rm ACC})
       (\delta' q'_{\rm CAC} + \delta'' q''_{\rm CAC})
       (\delta' q'_{\rm CCA} + \delta'' q''_{\rm CCA}) \notag\\
    &\qquad
       - (\delta' q'_{\rm TCG} + \delta'' q''_{\rm TCG})^2 \notag\\
    &= (\delta')^2 Q' + (\delta'')^2 Q'' + \delta'\delta''\, C \notag
    \end{align}
    where
    \begin{equation}
    C =
    q_{\rm CCA}(q'_{\rm ACC} q''_{\rm CAC} + q''_{\rm ACC} q'_{\rm CAC})
    - 2 q'_{\rm TCG} q''_{\rm TCG}.
    \end{equation}
    We now prove a claim regarding the cross term~$C$, which will form the basis of our inductive proof below.

    \textbf{Claim.} Suppose that $Q', Q'' \ge 0$. Then,
    \begin{enumerate}[label={(\roman*)}, noitemsep]
        \item $C \geq 0$; and
        \item $C > 0$, if $q'_{\rm ACC} q''_{\rm CAC} \neq q''_{\rm ACC} q'_{\rm CAC}$.
    \end{enumerate}
    
    \noindent
    \emph{Proof of claim.}
       From the definitions of $Q'$ and $Q''$ we obtain
    \[
    q_{\rm CCA} q'_{\rm ACC} q'_{\rm CAC} \ge (q'_{\rm TCG})^2,
    \qquad
    q_{\rm CCA} q''_{\rm ACC} q''_{\rm CAC} \ge (q''_{\rm TCG})^2 ,
    \]
    if $Q', Q'' \geq 0$.
    Applying the AM--GM inequality\footnote{The \emph{AM--GM inequality} states that for $x,y \geq 0$ we have $x+y \geq 2 \sqrt{xy}$, where equality holds if and only if $x=y$.} to $q'_{\rm ACC}q''_{\rm CAC} + q''_{\rm ACC}q'_{\rm CAC}$ gives
    \begin{align*}
        q_{\rm CCA}(q'_{\rm ACC}q''_{\rm CAC} + q''_{\rm ACC}q'_{\rm CAC})
    &\ge
    2\sqrt{(q_{\rm CCA} q'_{\rm ACC} q'_{\rm CAC})(q_{\rm CCA} q''_{\rm ACC} q''_{\rm CAC})} \\
    &\geq 2 q'_{\rm TCG} q''_{\rm TCG},
    \end{align*}
    with equality holding in the first line if and only if $q'_{\rm ACC} q''_{\rm CAC} = q''_{\rm ACC} q'_{\rm CAC}$. %\todo{Is this right? NH: now it should be :) SM: I think the primes are still wrong!}
    Using the bounds above, and the fact that the Fourier coordinates are positive, this concludes the proof of the claim. \xqed{$\diamond$}

    We now prove that $Q$ is strictly positive on $\mathcal{M}_2$ by induction on the number of reticulation vertices in $\mathcal{N}_2$. Since $\mathcal{N}_2$ is a trinet with a 3-blob, the number of reticulation vertices is equal to the level, $k \geq 1$. The base case (and the case $k=2$) was already established in \cite{englander2025identifiability}, but we provide an alternative proof based on the above Claim.

    \textbf{Case: $k=1$.} Then, $\mathcal{N}_2$ is a 3-sunlet, with $\mathcal{N}'_2$ and $\mathcal{N}''_2$ being 3-star trees, and hence $Q' = Q'' = 0$. Since $\delta',\delta'' \in (0,1)$, it suffices by \ref{eq:trinetJC:decomp} to show that $C > 0$. By the Claim above, we only need to show $q'_{\rm ACC} q''_{\rm CAC} \neq q''_{\rm ACC} q'_{\rm CAC}$.
    We now use the edge labelling from Figure~\ref{fig:3sunlet}, where we identify edges $d$ and $e'$, and edges $e$ and $e''$ (i.e. the mixing parameter $\delta'$ corresponds to $d = e'$ and $\delta''$ to $e = e''$). We then have $q'_{\rm ACC} = b_{\rm C} c_{\rm C} d_{\rm C} f_{\rm C}$, $q''_{\rm ACC} = b_{\rm C} c_{\rm C} e_{\rm C}$, $q'_{\rm CAC} = a_{\rm C} c_{\rm C} d_{\rm C}$ and $q''_{\rm CAC} =  a_{\rm C} c_{\rm C} e_{\rm C} f_{\rm C}$. So, $q'_{\rm ACC} q''_{\rm CAC} = ( q''_{\rm ACC} q'_{\rm CAC} ) \cdot f^2_{\rm C}$, and the result follows since the Fourier parameters are in $(0,1)$.
    
    \textbf{Case: $k \geq 2$.} By Lemma~\ref{lem:trinet_display_3blob}, we may assume that $\mathcal{N}'_2$ has a non-trivial 3-blob (with strictly fewer reticulation vertices than $\mathcal{N}_2$) and thus $Q' > 0$ by the induction hypothesis. Moreover, $Q'' \ge 0$, as either the induction hypothesis also applies (in which case $Q'' > 0$) or $\mathcal{N}''_2$ is a tree (in which case $Q''=0$). 
    Since $\delta',\delta'' \in (0,1)$, and since $C \geq 0$ by the claim, the result follows from \eqref{eq:trinetJC:decomp}.
\end{proof}

Note that the proof of Lemma~\ref{lem:trinetJC_levk} cannot readily be adapted to K2P since the polynomial $Q$ in Lemma \ref{lem:trinetk2p} is not invariant under all leaf permutations in K2P (take, for example, $\sigma = (123)$ in cycle notation). Therefore, even though one could assume that the last leaf in the ordering is incident to a reticulation in~$\mathcal{N}_2$, this cannot again be assumed for the network~$\mathcal{N}'_2$, so the induction hypothesis would not apply.

% \begin{lemma}\label{lem:trinetK2P_levk}
%     Let $\mathcal{N}_1$ be a trinet without a non-trivial 3-blob and $\mathcal{N}_2$ a trinet with a non-trivial 3-blob, and let $\mathcal{M}_1$ and $\mathcal{M}_2$ be the associated models under the K2P evolutionary model with parameter values in $\Theta_0(\mathcal{N}_1)$ and $\Theta_0(\mathcal{N}_2)$ respectively. Then the polynomial
%     $$ Q = q_{\rm AGG}q_{\rm GAG}q_{\rm CCA}^2 - q_{\rm AAA} q_{\rm GGA}q_{\rm TCG}^2$$
%     is zero on $\mathcal{M}_1$ and strictly positive on $\mathcal{M}_2$. Thus $\mathcal{M}_1\cap\mathcal{M}_2=\emptyset$.
% \end{lemma}
% The proof of Lemma~\ref{lem:trinetK2P_levk} is analogous to the proof of Lemma~\ref{lem:trinetJC_levk} and given in the appendix.

\subsection{Distinguishing Phylogenetic Trees from Phylogenetic Networks}

Before proving the next proposition, we require the following terminology and auxiliary results. 

A cycle~$C$ in a semi-directed network is called \emph{good} if there exists a unique reticulation vertex~$r$, called the \emph{sink} of~$C$, such that~$C$ contains both reticulation edges entering~$r$. A good cycle is called \emph{excellent} if its sink is adjacent to a leaf. Call a semi-directed network \emph{simple} if it contains at most one $m$-blob with $m\geq 3$.

By Lemma~5.4 of \cite{huber2025quarnets}, after suppressing all 2-blobs, every non-cut-edge of a simple, semi-directed network is contained in an excellent cycle.

\begin{lemma}\label{lem:reticulate_trinet}
Let $\mathcal{N}$ be a semi-directed network on $\mathcal{X}$ containing a non-trivial $m$-blob, where $m\geq 3$. Then there exist leaves $a,b,c\in \mathcal{X}$ such that the trinet $\mathcal{N}|_{\{a,b,c\}}$ contains a non-trivial $3$-blob.
\end{lemma}

\begin{proof}
Since $\mathcal{N}$ contains a non-trivial $m$-blob $B$, there exists a subset
$\mathcal{Y}\subseteq \mathcal{X}$ of %size~$k$
size~$m$ such that $\mathcal{N}|_{\mathcal{Y}}$ is simple and contains a non-trivial
$m$-blob. Indeed, removing $B$ disconnects $\mathcal{N}$ into
exactly~$m$ %$k$
components, each containing at least one leaf. By choosing one leaf
from each such component and letting $\mathcal{Y}$ be the resulting set of leaves, the
subnetwork $\mathcal{N}|_{\mathcal{Y}}$ retains the blob $B$ as a non-trivial $m$-blob.
%$k$-blob.
Thus, replacing $\mathcal{N}$ by $\mathcal{N}|_{\mathcal{Y}}$ if necessary, we may assume
that $\mathcal{N}$ is simple.

We next note that only the structure of the network in and around the
non-trivial $m$-blob $B$ is relevant for the argument. Any $2$-blobs attached to~$B$ do not affect the presence or absence of non-trivial 3-blobs in trinets. Hence, we may assume without loss of
generality that $\mathcal{N}$ has no $2$-blobs.

Let $e$ be an edge in the non-trivial $m$-blob $B$. By \cite[Lemma~17]{banos2019identifying},
there exist leaves $x,y \in \mathcal{X}$ such that $e$ lies on an up-down path
between $x$ and $y$. This lemma is stated for rooted networks, but the result for semi-directed networks follows directly by considering any rooting. By \cite[Lemma~5.4]{huber2025quarnets} (mentioned above), it then follows that the edge~$e$ is contained in an excellent cycle $C$. Let~$a\in\mathcal{X}$ be the leaf incident to the sink of~$C$.
%whose sink is a reticulation vertex $v$ adjacent to a leaf $a \in \mathcal{X}$.

We claim that there exist two disjoint up-down paths from~$C$ to leaves not equal to~$a$.
The up-down path between $x$ and $y$ traversing $e$ contains disjoint up-down paths from~$C$ to~$x$ and~$y$. This proves the claim if $a \notin \{x,y\}$.

Hence, assume $a \in \{x,y\}$, say $a=x$. If every up-down path from~$y$
%$x$
to a leaf in $\mathcal{X}\setminus\{x,y\}$ avoids $C$, this would
imply that the blob containing $C$ is a $2$-blob, contradicting our assumption
that $B$ is a non-trivial $m$-blob, $m\geq 3$. Hence there exists a leaf
$c \in \mathcal{X}\setminus\{x,y\}$ such that some up-down path from
%$x$
$y$ to $c$ uses an
edge of $C$. This path then contains two disjoint up-down paths from~$C$ to~$y$ and~$c$ with~$y,c\neq a$. This concludes the proof of the claim.

Thus, we have shown that~$B$ contains an excellent cycle~$C$ with a leaf~$a$ incident to the unique sink of~$C$ and with disjoint up-down paths from~$C$ to leaves~$b,c\neq a$. Observe that each vertex~$v$ of~$C$ lies on at least one up-down path between leaves in~$\{a,b,c\}$. To see this, note that there exist  up-down paths from~$v$ to~$b$ and~$c$ (by~\cite[Corollary~1]{holtgrefe2026characterizing}) and a semi-directed path from~$v$ to~$a$. Hence, the trinet $\mathcal{N}|_{\{a,b,c\}}$ contains a cycle with disjoint paths between this cycle and~$a,b$ and~$c$. Hence, $\mathcal{N}|_{\{a,b,c\}}$ contains a nontrivial $3$-blob.
\end{proof}

Clearly, a network without a non-trivial $m$-blob, $m\geq 3$, cannot induce a trinet that has a non-trivial $m$-blob. Hence, by combining Theorem~\ref{thm:restriction_network} with Lemma~\ref{lem:trinetJC_levk}, %and~\ref{lem:trinetK2P_levk},
the previous lemma implies the following result.

\begin{corollary}\label{cor:trees_vs_networks}
    Let $\mathcal{N}_1$ and $\mathcal{N}_2$ be two semi-directed phylogenetic networks  on a set of $n$ taxa $\mathcal{X}$, with $n\geq 3$. Let $\mathcal{M}_1$ and $\mathcal{M}_2$ be the corresponding models under the JC model on the  parameter sets $\Theta_0(\mathcal{N}_1)$ and $\Theta_0(\mathcal{N}_2)$ respectively. If $\mathcal{N}_1$ contains at least one non-trivial $m$-blob, $m\geq 3$, and $\mathcal{N}_2$ does not, then $\mathcal{M}_1\cap\mathcal{M}_2=\emptyset$. \xqed{\proofbox}
\end{corollary}

Observe that the presence or absence of a non-trivial $m$-blob, $m\geq 3$, is preserved when passing from a rooted phylogenetic network to its associated semi-directed network. Hence Corollary~\ref{cor:trees_vs_networks} also holds for rooted phylogenetic networks (see Figure~\ref{fig:tree_vs_network}).

 %-------------------------------------------------------------------------------------------------
%
%                                    Discussion
%
%-------------------------------------------------------------------------------------------------

\section{Discussion}\label{sec:discussion}

We have shown that, under the JC, K2P, and K3P models, the semi-directed network parameter of level-1 phylogenetic networks is fully identifiable from the leaf-pattern distribution, provided all mixing parameters are non-trivial. A key structural insight underlying this result is that two level-1 network models overlap exactly in the union of models of their maximal shared displayed networks---the largest networks displayed by both. As an example, the intersection of two sunlet models is exactly the model of their common displayed tree. This proves that the method for identifying 4-leaf unrooted trees used in~\cite{martin2025algebraic} is statistically consistent. More generally, our result shows that phylogenetic network reconstruction methods based on leaf-pattern probabilities
should be statistically consistent on level-1 phylogenetic networks, 
up to the semi-directed phylogenetic network and placement of reticulations in triangles, under the JC, K2P, and K3P models. Furthermore, our result is a necessary condition for the statistical consistency of maximum likelihood as an estimator for a level-1 semi-directed phylogenetic network topology under those models. To the best of our knowledge, it is not known whether different phylogenetic networks have different leaf-pattern probability distributions for more popular but more complex models, such as the generalized time-reversible (GTR) model.

Our identifiability results are up to the placement of the reticulation vertex in any triangles in the network. This is because, for group-based models, the leaf-pattern distribution of 3-sunlets is the same regardless of the placement of the reticulation vertex \cite{gross2018distinguishing}. In recent work, however, semi-algebraic conditions were given under which the placement of the reticulation vertex in a triangle is identifiable from the leaf-pattern distribution for the JC model \cite{currie2026semialgebraic}. These conditions are valid on the same parameter set $\Theta_0$ that we consider here.

The second contribution of this paper is tree-network distinguishability under the JC model. %and K2P models. 
Lemma~\ref{lem:trinetJC_levk} %and~\ref{lem:trinetK2P_levk}
establishes that trinets of arbitrary level can be distinguished from trees under this model. The extension %to K2P and
to arbitrary level is new: \cite{englander2025identifiability} established the trinet inequality only for %the JC model and for
level-1 and level-2 networks. This shows that, modulo suppressing 2-blobs, any phylogenetic network can be distinguished from any phylogenetic tree---that is, no network and tree can produce the same leaf-pattern distributions when the mixing parameters of the network are non-trivial. Beyond substitution models, Lemma~\ref{lem:reticulate_trinet}, in combination with the proof of the \textup{`A-3blob'} assumption from \cite{allman2025beyond}, can be used to establish tree-network distinguishability under several coalescent-based models as well.

The trinet inequality has several further identifiability consequences. One concerns a class of galled tree-child networks satisfying topological restrictions on cycle sizes and hybrid attachment; see Remark~5 of~\cite{allman2025beyond} for precise conditions. Quartet distinguishability~\cite[Proposition~8]{englander2025identifiability} verifies assumptions \textup{`A-4circ'} and \textup{`A-ToB'} of~\cite{allman2025beyond}, and our trinet results verify \textup{`A-3blob'}, so that Lemma~4.2, Theorem~4.3, and Remark~5 of~\cite{allman2025beyond} yield identifiability for a class of networks of arbitrary level under JC. % and K2P. Our extension to K2P is also a natural first step toward lifting the level-2 identifiability results of~\cite{englander2025identifiability} from JC to K2P.
In a related direction, \cite{holtgrefe2025distinguishing} showed identifiability of outer-labelled planar level-2 networks from quartets, subject to cycle-size restrictions---notably excluding 4-cycles. Since our trinet results may provide a means to locate smaller cycles, it is conceivable that some of these restrictions can be relaxed.

Several questions remain open. It is natural to ask whether the high-level trinet inequality extends to the K2P and K3P models, whether quartet distinguishability extends to the K3P model, and whether the identifiability results for level-2 networks in~\cite{englander2025identifiability,holtgrefe2025distinguishing} can then be fully lifted to K2P and K3P. %This seems likely to be true given that we have
Even though we have already shown the level-1 trinet inequality for K2P and K3P, %However,
it is not %immediately
clear how to adapt the proof of Lemma \ref{lem:trinetJC_levk} %and \ref{lem:trinetK2P_levk}
to the cases of K2P and K3P,
since this proof relies on the single invariant found for the level-1 trinet inequality in that case, and that this invariant is itself invariant under permuting the leaves.
%while for K3P, the proof of the level-1 trinet inequality (Lemma \ref{lem:trinetk3p}) requires several invariants. 

More generally, we could ask whether our results hold for the more general class of equivariant substitution models. This class includes the group-based models, but also includes the strand-symmetric model and the general Markov model. These models maintain many of the nice properties of the models that we have used here (see \cite{sullivant2025phylogenetic}), so it seems plausible that our results could extend to these models.

Here we considered the parameter space $\Theta_0$, and extended this to $\Theta_0^+$, allowing mixing parameters to be either 0 or 1. This enabled a description of the intersection of two level-1 phylogenetic networks in terms of their displayed networks. In $\Theta_0$ the transition matrix parameters $a^e_x$ lie in the open interval $(0,1)$. We could also consider extending this to the interval $(0,1]$. Having $a^e_x = 1$ for all $x\in\{{\rm A,C,G,T}\}$ corresponds to a branch of length 0. In phylogenetic applications branches of length 0 can correspond to uncertainty in the order of branching events. We would then expect to find further intersections between phylogenetic network models, since, for example, the variety corresponding to a single triangle quarnet is contained in four $4$-sunlet varieties, and these containments occur when a branch in the $4$-sunlet is shrunk to have length 0 (see \cite[Example~4.10]{gross2018distinguishing}).

Along similar lines, another direction is the extension of our results to non-binary trees and networks. These are phylogenetic trees and networks in which the in-degree and out-degree of the internal vertices are not restricted (recall that in this work, the sum of in-degree and out-degree for each non-root internal vertex is three). Such vertices are often called multifurcations, and can represent uncertainty in the order of events represented by a phylogenetic tree or network. In early work for this project, we found that on the parameter space $\Theta_0$, we can distinguish between a 4-star tree and a 4-leaf unrooted binary tree (result not shown), suggesting that further full identifiability results may be possible in this direction.

\section*{Acknowledgments}
We thank Seth Sullivant, Vincent Moulton and Guuske Kouwenhoven for interesting discussions related to the questions discussed in this paper. We thank Alec Kriebel for pointing us at a mistake in an earlier version of this paper.
%We thank Seth Sullivant for discussions on suppressing 2-blobs, see also \cite{sullivant2025phylogenetic}, and Guuske Kouwenhoven for discussions on trinet inequalities.

\section*{Funding}
NH \& LvI were supported by grant OCENW.M.21.306 of the Dutch Research Council (NWO).
SM was supported by the European Molecular Biology Laboratory (EMBL).

\bibliographystyle{plain}
\bibliography{identifiability}
\newpage

\appendix
\section{Further Proofs}
In this appendix we give proofs left out of the main document, and complete proofs for some results in the main document that were left incomplete.
\renewcommand\thetheorem{3.3}
\begin{proposition}
	Let $\mathcal{T}$ be an unrooted binary phylogenetic tree on a set of taxa $\mathcal{X}$, and let $\mathcal{M}$ be the corresponding model for a multiplicatively closed substitution model. For a leaf $x\in\mathcal{X}$, let $\mathcal{T}' = \mathcal{T}|_{\mathcal{X}\setminus\{x\}}$ be the restricted tree obtained from $\mathcal{T}$ by pruning the leaf $x$ and suppressing the resulting degree two vertex, and let $\mathcal{M}'$ be the corresponding model for the same substitution model. If $\mathbf{p}\in\mathcal{M}$ then the point $\mathbf{p}'$ obtained from $\mathbf{p}$ by marginalizing over the leaf $x$ is in $\mathcal{M}'$. \qed%\todo{downgrade to lemma/proposition?}
\end{proposition}
\begin{proof}
We assume, without loss of generality, that $\mathcal{X} = \{1,\ldots, n\}$ and that we restrict to the set $\mathcal{X}' = \{1,\ldots, n-1\}$. Consider an element $\mathbf{p} \in\mathcal{M}$. We can write the coordinates of $\mathbf{p}$ as 
	\begin{align*}
		p_{x_1, \ldots, x_n} :&= P(X_1 = x_1, \ldots, X_n = x_n) \\
		&= \sum_{x \in X(x_1,\ldots, x_n)}\pi_{x_\rho}\prod_{(u,v) \in E(\mathcal{T})}M_{x_u, x_v}^{(u,v)}, 
	\end{align*}
	for some choice of matrices $M^{(u,v)}$ for each edge $(u,v) \in E(\mathcal{T})$. Marginalizing over states at the leaf $n$ we obtain 
	\begin{align*}
		\big(\mathrm{m}_n(\mathbf{p})\big)_{x_1, \ldots, x_{n-1}} :&=  \sum_{x_n\in\Omega} p_{x_1, x_2, \ldots, x_{n}} \\
		&= \sum_{x_n\in\Omega}\sum_{x \in X(x_1,\ldots, x_n)}\pi_{x_\rho}\Big(\prod_{(u,v) \in E(\mathcal{T})\setminus{\{e\}}}M_{x_u, x_v}^{(u,v)}\Big)M^e_{x_{p(n)}, x_n}\\
		&= \sum_{x \in X(x_1,\ldots, x_{n-1})}\pi_{x_\rho}\Big(\prod_{(u,v) \in E(\mathcal{T})\setminus{\{e\}}}M_{x_u, x_v}^{(u,v)}\Big)\sum_{x_n\in\Omega}M^e_{x_{p(n)}, x_n}
	\end{align*}
	where the edge $e$ is the edge leading to leaf $n$, and $x_{p(n)}$ is the state at the vertex $p(n)$, the parent of $n$, in the assignment $x$ (see Figure \ref{fig:T1}(A)). In the third line we sum over all assignments of states at the vertices of $\mathcal{T}$ for which leaves $1,\ldots,n-1$ are assigned states $x_1,\ldots, x_{n-1}$ respectively. Now since $M^e$ is a Markov matrix, we have $\sum_{x_b\in\Omega}M^e_{x_a, x_b} = 1$, so
	\begin{align*}
		\big(\mathrm{m}_n(\mathbf{p})\big)_{x_1, \ldots, x_{n-1}} &= \sum_{x \in X(x_1,\ldots, x_{n-1})}\pi_{x_\rho}\prod_{(u,v) \in E(\mathcal{T})\setminus{\{e\}}}M_{x_u, x_v}^{(u,v)}.
	\end{align*}
	Next, let $X'(x_1,\ldots, x_{n-1})$ be the set of all assignments of states at the vertices of $\mathcal{T}'$ for which leaves $1,\ldots,n-1$ are assigned states $x_1,\ldots, x_{n-1}$ respectively. Observe that the vertices of  $\mathcal{T}'$ are exactly those of $\mathcal{T}$, except for $n$ and $p(n)$ (see Figure \ref{fig:T1}), so we may write
	\begin{align*}
		\big(\mathrm{m}_n(\mathbf{p})\big)_{x_1, \ldots, x_{n-1}} &= \sum_{x \in X'(x_1,\ldots, x_{n-1})}\pi_{x_\rho}\Big(\prod_{(u,v) \in E(\mathcal{T})\setminus{\{e,e_1, e_2\}}}M_{x_u, x_v}^{(u,v)}\Big)\sum_{x_{a} \in \Omega}M_{x_{p(e_1)}, x_a}^{e_1}M_{x_a, x_{p(e_2)}}^{e_2}.
	\end{align*}
	
	Now, since the substitution model is multiplicatively closed, the product of transition matrices $M^{e_1}M^{e_2}$ is also a transition matrix in the model. Let $M^{e'}$ be this matrix. Then we have 
	\begin{align*}
		\big(\mathrm{m}_n(\mathbf{p})\big)_{x_1, \ldots, x_{n-1}} &= \sum_{x \in X'(x_1,\ldots, x_{n-1})}\pi_{x_\rho}\Big(\prod_{(u,v) \in E(\mathcal{T})\setminus{\{e,e_1, e_2\}}}M_{x_u, x_v}^{(u,v)}\Big)M^{e'}_{x_{p(e_1)},x_{p(e_2)}} \\
			&=  \sum_{x \in X'(x_1,\ldots, x_{n-1})}\pi_{x_\rho}\prod_{(u,v) \in E(\mathcal{T}')}M_{x_u, x_v}^{(u,v)} \\
			&= p_{x_1, x_2, \ldots, x_{n-1}},
	\end{align*}
	
	where we write $e'$ for the edge in $\mathcal{T}'$ got by suppressing the vertex $p(n)$ in $\mathcal{T}$. Thus we have shown that by marginalizing over leaf $n$, we obtain a point in the model of the restricted tree $\mathcal{T}'$.
\end{proof}

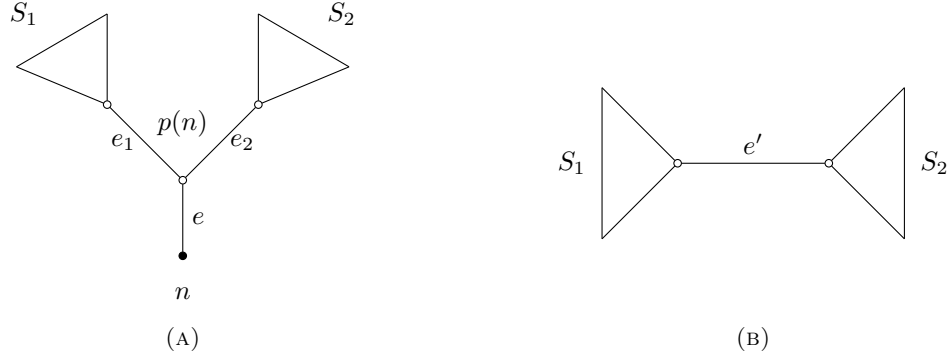
\begin{figure}[h!]
    \centering
    \begin{subfigure}[b]{0.45\linewidth}
        \centering
        \begin{tikzpicture}[scale = 1]
            \node[shape=circle,draw=black,scale=0.3] (a) at (0,0) {};
            \node[shape=circle,draw=black,scale=0.3] (b) at (-1,1) {};
            \node[shape=circle,draw=black,scale=0.3] (c) at (1,1) {};
            \node[shape=circle,draw=black,scale=0.3,fill=black] (d) at (0,-1) {};
            
            \draw[] (a)--(b) node[midway,left=0.05] {$e_{1}$};
            \draw[] (b)--(-2.2,1.5);
            \draw[] (-1,2.2)--(-2.2,1.5);
            \draw[] (b)--(-1,2.2);
            \draw[] (a)--(c)  node[midway,right=0.05] {$e_{2}$};
            \draw[] (c)--(2.2,1.5);
            \draw[] (1,2.2)--(2.2,1.5);
            \draw[] (c)--(1,2.2);
            \draw[] (a)--(d) node[midway, right=0.05] {$e$};
            
            \node (v1) at (0,.75) {$p(n)$};
            \node (v2) at (0,-1.5) {$n$};
            \node (v3) at (2.1,2.2) {$S_2$};
            \node (v4) at (-2.1,2.2) {$S_1$};
        
        \end{tikzpicture}
        \caption{}
    \end{subfigure}
    \begin{subfigure}[b]{0.45\linewidth}
        \centering
        \begin{tikzpicture}[scale = 1]
            \node[shape=circle,draw=black,scale=0.3] (b) at (-1,0) {};
            \node[shape=circle,draw=black,scale=0.3] (a) at (1,0) {};
            
            \draw[] (a)--(b) node[midway,above=0.05] {$e'$};
            \draw[] (b)--(-2,1);
            \draw[] (-2,1)--(-2,-1);
            \draw[] (b)--(-2,-1);
            \draw[] (a)--(2,1);
            \draw[] (2,1)--(2,-1);
            \draw[] (a)--(2,-1);
            
            \node (v1) at (2.4,0) {$S_2$};
            \node (v2) at (-2.4,0) {$S_1$};
            \node (v3) at (0,-1.8) {};
        
        \end{tikzpicture}
        \caption{}
    \end{subfigure}
    \caption{\textbf{(A)} A phylogenetic tree $\mathcal{T}$, focussed on the leaf $n$, its parent vertex $p(n)$, and the incident edges leading to the subtrees $S_1$ and $S_2$. \textbf{(B)} The corresponding restricted tree $\mathcal{T}' = \mathcal{T}|_{[n-1]}$. Subtrees $S_1$ and $S_2$ remain unchanged.}
    \label{fig:T1}
\end{figure}

\renewcommand\thetheorem{4.1}
\begin{lemma}
    Let $\mathcal{N}_1$ be the 3-star unrooted tree and $\mathcal{N}_2$ be a 3-sunlet semi-directed network, and let $\mathcal{M}_1$ and $\mathcal{M}_2$ be the associated models under the K2P substitution model with parameter values in $\Theta_0(\mathcal{N}_1)$ and $\Theta_0(\mathcal{N}_2)$ respectively. Then the polynomial
    $$ Q = q_{\rm AGG}q_{\rm GAG}q_{\rm CCA}^2 - q_{\rm AAA}q_{\rm GGA}q_{\rm TCG}^2$$
    is zero on $\mathcal{M}_1$ and strictly positive on $\mathcal{M}_2$. In particular, $\mathcal{M}_1\cap\mathcal{M}_2=\emptyset$.
\end{lemma}
\begin{proof}
    First, we have that $Q\in I_{\mathcal{N}_1}$. This is clear by looking at the multiset of labels for each leaf position under the K2P labelling (see e.g. \cite{sturmfels2005toric}). For example, in the first leaf position we have $\{{\rm A,G,C,C}\}$ for the monomial on the left and $\{{\rm A,G,T,T}\}$ for the monomial on the right. For a star tree, these multisets describe the parameters for a particular edge in the parameterization of each Fourier coordinate. Since for K2P we identify the Fourier parameters corresponding to ${\rm C}$ and ${\rm T}$ for each edge, the parameters corresponding to these multisets are the same.

    To show that $Q$ is strictly positive on $\mathcal{N}_2$, we label the edges of the 3-sunlet as in Figure~\ref{fig:3sunlet} of the main document. For each edge $a$, we have the associated parameters $a_{\rm A}, a_{\rm C}, a_{\rm G}, a_{\rm T}$ which are subject to the constraints $a_{\rm A} = 1$ and $a_{\rm C}=a_{\rm T}$. We substitute a generic point into $Q$ using the parameterization of the model, which is given by (see Example \ref{ex:3sunlet})
    $$q_{x_1x_2x_3}=a_{x_1}b_{x_2}c_{x_3}(\delta d_{x_3}f_{x_2} + (1-\delta)e_{x_3}f_{x_2+x_3}),$$
    for $x_1, x_2, x_3 \in G$ with $x_1+x_2+x_3 = 0$. We have
 \begin{align*}
    q_{\rm AGG} &= b_{\rm G}c_{\rm G}(\delta d_{\rm G} f_{\rm G} + (1-\delta)e_{\rm G}), \\
    q_{\rm CCA} &= a_{\rm C}b_{\rm C}f_{\rm C}, \\
    q_{\rm GAG} &= a_{\rm G}c_{\rm G}(\delta d_{\rm G} + (1-\delta)e_{\rm G}f_{\rm G}), \\
    q_{\rm GGA} &= a_{\rm G}b_{\rm G}f_{\rm G}, \\
    q_{\rm TCG} &= a_{\rm C}b_{\rm C}c_{\rm G}f_{\rm C}(\delta d_{\rm G} + (1-\delta)e_{\rm G}).
\end{align*}    
Expanding monomials in $Q$ we have
$$q_{\rm AGG}q_{\rm GAG}q_{\rm CCA}^2 = a_{\rm C}^2a_{\rm G}b_{\rm C}^2b_{\rm G}c_{\rm G}^2f_{\rm C}^2(\delta^2d_{\rm G}^2f_{\rm G} + \delta(1-\delta)(d_{\rm G}e_{\rm G} + d_{\rm G}e_{\rm G}f_{\rm G}^2) + (1-\delta)^2e_{\rm G}^2f_{\rm G})$$
and
$$q_{\rm GGA}q_{\rm TCG}^2 = a_{\rm C}^2a_{\rm G}b_{\rm C}^2b_{\rm G}c_{\rm G}^2f_{\rm C}^2f_{\rm G}(\delta^2 d_{\rm G}^2 + 2\delta(1-\delta)d_{\rm G}e_{\rm G} + (1-\delta)^2e_{\rm G}^2).$$
Taking the difference we obtain
\begin{align*}
    q_{\rm AGG}q_{\rm GAG}q_{\rm CCA}^2 - &q_{\rm GGA}q_{\rm TCG}^2\\
     &= a_{\rm C}^2a_{\rm G}b_{\rm C}^2b_{\rm G}c_{\rm G}^2f_{\rm C}^2\delta(1-\delta)(d_{\rm G}e_{\rm G} + d_{\rm G}e_{\rm G}f_{\rm G}^2 - 2d_{\rm G}e_{\rm G}f_{\rm G}) \\
    &= a_{\rm C}^2a_{\rm G}b_{\rm C}^2b_{\rm G}c_{\rm G}^2f_{\rm C}^2\delta(1-\delta)d_{\rm G}e_{\rm G}(1+f_{\rm G}^2 - 2f_{\rm G}) \\
    &= a_{\rm C}^2a_{\rm G}b_{\rm C}^2b_{\rm G}c_{\rm G}^2d_{\rm G}e_{\rm G}f_{\rm C}^2\delta(1-\delta)(1-f_{\rm G})^2.
\end{align*}
Since all parameters lie in $(0,1)$, this is strictly positive.
\end{proof}

\renewcommand\thetheorem{4.6}
\begin{lemma}
    Let $\mathcal{N}_1$ be a level-1 quarnet with split $12|34$, and let $\mathcal{N}_2$ be a level-1 quarnet without the split $12|34$. Let $\mathcal{M}_1$ and $\mathcal{M}_2$ be the corresponding models for either the JC, K2P, or K3P models on the restricted parameter set. Then $\mathcal{M}_1\cap\mathcal{M}_2=\emptyset$.
\end{lemma}
\begin{proof}
    In the main document we reduced the proof to showing that the polynomial $Q_{12|34}=q_{\rm AAAA}q_{\rm TTTT} - q_{\rm AATT}q_{\rm TTAA}$ is strictly positive for the single-triangle quarnets and double-triangle quarnets without the split $12|34$. We showed this for the single-triangle quarnet with split $13|24$ in Figure \ref{fig:singletris} (a). It remains to show that $Q_{12|34}$ is strictly positive for the single-triangle quarnets in Figure \ref{fig:singletris} (b), (c), and (d), and the double-triangle quarnets in Figure \ref{fig:doubletris} (a) and (b). Note that the placement of the reticulation vertex in the triangle does not affect the underlying model \cite{gross2018distinguishing}, so these are all remaining quarnets without the split $12|34$, for which the model is distinct.

    \begin{figure}[h!]
    \centering
        \begin{subfigure}{0.24\textwidth}
            \centering
            \begin{tikzpicture}[scale = .35]             
                \node[shape=circle,draw=black,scale=0.3] (B) at (0,2) {};
                \node[shape=circle,draw=black,scale=0.3] (C) at (1.8,5) {};
                \node[shape=circle,draw=black,scale=0.3] (D) at (-1.8,5) {};
                \node[shape=circle,draw=black,scale=0.3] (E) at (0,0) {};
                \node[shape=circle,draw=black,scale=0.3,fill=black] (3) at (3,6.5) {};
                \node[shape=circle,draw=black,scale=0.3,fill=black] (1) at (-3,6.5) {}; 
                \node[shape=circle,draw=black,scale=0.3,fill=black] (4) at (3,-2) {};
                \node[shape=circle,draw=black,scale=0.3,fill=black] (2) at (-3,-2) {};
                \node at (3.2,7.4) {$3$};
                \node at (-3.2,7.4) {$1$};
                \node at (3.2,-2.8) {$4$};
                \node at (-3.2,-2.8) {$2$};                
                \draw [->](B)--(C) node[midway,right=0.1] {};
                \draw [dashed,->](B)--(D) node[midway,left=0.1] {};
                \draw [dashed,->](C)--(D) node[midway,above=0.05] {};
                \draw [->](C)--(3) node[midway,below right] {};
                \draw [->](D)--(1) node[midway,below left] {};
                \draw [->](B)--(E) node[midway,left=0.05] {};
                \draw [->](E)--(2) node[midway,above left] {};
                \draw [->](E)--(4) node[midway,above right] {};          
            \end{tikzpicture}
        \caption{}
        \end{subfigure}
        \hfill
        \begin{subfigure}{0.24\textwidth}
            \centering
            \begin{tikzpicture}[scale = .35]           
                \node[shape=circle,draw=black,scale=0.3] (B) at (0,2) {};
                \node[shape=circle,draw=black,scale=0.3] (C) at (1.8,5) {};
                \node[shape=circle,draw=black,scale=0.3] (D) at (-1.8,5) {};
                \node[shape=circle,draw=black,scale=0.3] (E) at (0,0) {};
                \node[shape=circle,draw=black,scale=0.3,fill=black] (3) at (3,6.5) {};
                \node[shape=circle,draw=black,scale=0.3,fill=black] (1) at (-3,6.5) {}; 
                \node[shape=circle,draw=black,scale=0.3,fill=black] (4) at (3,-2) {};
                \node[shape=circle,draw=black,scale=0.3,fill=black] (2) at (-3,-2) {};
                \node at (3.2,7.4) {$4$};
                \node at (-3.2,7.4) {$1$};
                \node at (3.2,-2.8) {$3$};
                \node at (-3.2,-2.8) {$2$};                
                \draw [->](B)--(C) node[midway,right=0.1] {};
                \draw [dashed,->](B)--(D) node[midway,left=0.1] {};
                \draw [dashed,->](C)--(D) node[midway,above=0.05] {};
                \draw [->](C)--(3) node[midway,below right] {};
                \draw [->](D)--(1) node[midway,below left] {};
                \draw [->](B)--(E) node[midway,left=0.05] {};
                \draw [->](E)--(2) node[midway,above left] {};
                \draw [->](E)--(4) node[midway,above right] {};           
            \end{tikzpicture}
        \caption{}
        \end{subfigure}
        \hfill
        \begin{subfigure}{0.24\textwidth}
            \centering
            \begin{tikzpicture}[scale = .35]      
                \node[shape=circle,draw=black,scale=0.3] (B) at (0,2) {};
                \node[shape=circle,draw=black,scale=0.3] (C) at (1.8,5) {};
                \node[shape=circle,draw=black,scale=0.3] (D) at (-1.8,5) {};
                \node[shape=circle,draw=black,scale=0.3] (E) at (0,0) {};
                \node[shape=circle,draw=black,scale=0.3,fill=black] (3) at (3,6.5) {};
                \node[shape=circle,draw=black,scale=0.3,fill=black] (1) at (-3,6.5) {}; 
                \node[shape=circle,draw=black,scale=0.3,fill=black] (4) at (3,-2) {};
                \node[shape=circle,draw=black,scale=0.3,fill=black] (2) at (-3,-2) {};
                \node at (3.2,7.4) {$4$};
                \node at (-3.2,7.4) {$2$};
                \node at (3.2,-2.8) {$3$};
                \node at (-3.2,-2.8) {$1$};                
                \draw [->](B)--(C) node[midway,right=0.1] {};
                \draw [dashed,->](B)--(D) node[midway,left=0.1] {};
                \draw [dashed,->](C)--(D) node[midway,above=0.05] {};
                \draw [->](C)--(3) node[midway,below right] {};
                \draw [->](D)--(1) node[midway,below left] {};
                \draw [->](B)--(E) node[midway,left=0.05] {};
                \draw [->](E)--(2) node[midway,above left] {};
                \draw [->](E)--(4) node[midway,above right] {};         
            \end{tikzpicture}
        \caption{}
        \end{subfigure}
        \hfill
        \begin{subfigure}{0.24\textwidth}
            \centering
            \begin{tikzpicture}[scale = .35]              
                \node[shape=circle,draw=black,scale=0.3] (B) at (0,2) {};
                \node[shape=circle,draw=black,scale=0.3] (C) at (1.8,5) {};
                \node[shape=circle,draw=black,scale=0.3] (D) at (-1.8,5) {};
                \node[shape=circle,draw=black,scale=0.3] (E) at (0,0) {};
                \node[shape=circle,draw=black,scale=0.3,fill=black] (3) at (3,6.5) {};
                \node[shape=circle,draw=black,scale=0.3,fill=black] (1) at (-3,6.5) {}; 
                \node[shape=circle,draw=black,scale=0.3,fill=black] (4) at (3,-2) {};
                \node[shape=circle,draw=black,scale=0.3,fill=black] (2) at (-3,-2) {};
                \node at (3.2,7.4) {$3$};
                \node at (-3.2,7.4) {$2$};
                \node at (3.2,-2.8) {$4$};
                \node at (-3.2,-2.8) {$1$};                
                \draw [->](B)--(C) node[midway,right=0.1] {};
                \draw [dashed,->](B)--(D) node[midway,left=0.1] {};
                \draw [dashed,->](C)--(D) node[midway,above=0.05] {};
                \draw [->](C)--(3) node[midway,below right] {};
                \draw [->](D)--(1) node[midway,below left] {};
                \draw [->](B)--(E) node[midway,left=0.05] {};
                \draw [->](E)--(2) node[midway,above left] {};
                \draw [->](E)--(4) node[midway,above right] {};      
            \end{tikzpicture}
        \caption{}
        \end{subfigure}
    \caption{All single-triangle quarnets without the split $12|34$.}
    \label{fig:singletris}
\end{figure}
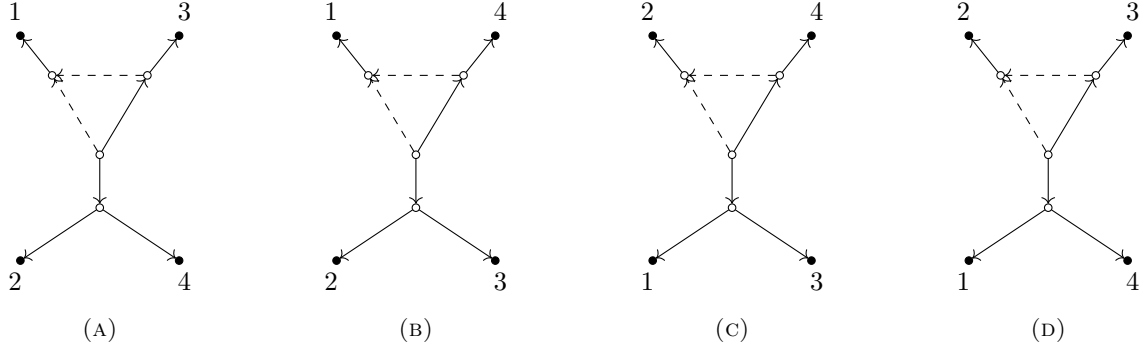

First we show that $Q_{12|34}$ is strictly positive on the quarnets in Figure~\ref{fig:singletris} (b), (c), and (d). We begin with (b). Observe that (b) can be obtained from (a) via the leaf-permutation $(34)$ (written in cycle notation). Now $Q_{12|34}$ is invariant under $(34)$ (where for a permutation $\theta\in S_4$ we define $\theta(q_{g_1g_2g_3g_4}) = q_{g_{\theta(1)}g_{\theta(2)}g_{\theta(3)}g_{\theta(4)}}$ and extend as a ring homomorphism), so it follows that $Q_{12|34}$ is strictly positive on (b). For the quarnets (c) and (d) we observe that (a) can be obtained from (c) via the leaf permutation $(12)(34)$ and (b) can be obtained from (d) by the same permutation. Since $Q_{12|34}$ is invariant under this permutation too, it follows that $Q_{12|34}$ is strictly positive on (c) and (d).

Next we consider the double-triangle quarnets shown in Figure \ref{fig:doubletris}.

    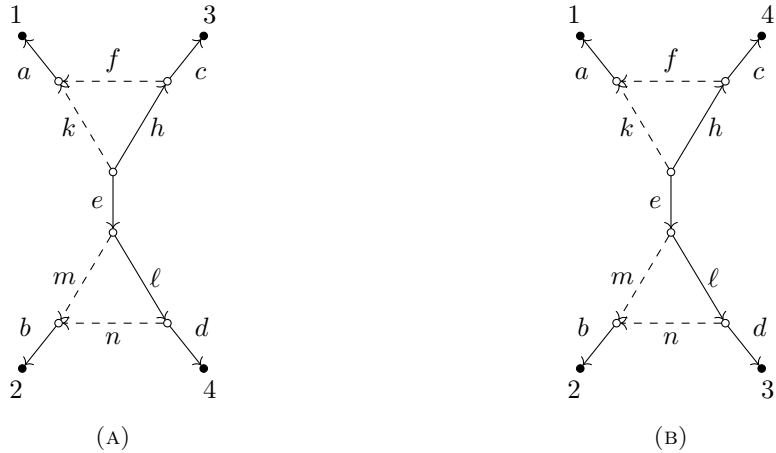
\begin{figure}[h!]
    \centering
        \begin{subfigure}{0.44\textwidth}
            \centering
            \begin{tikzpicture}[scale = .4]             
                \node[shape=circle,draw=black,scale=0.3] (B) at (0,2) {};
                \node[shape=circle,draw=black,scale=0.3] (C) at (1.8,5) {};
                \node[shape=circle,draw=black,scale=0.3] (D) at (-1.8,5) {};
                \node[shape=circle,draw=black,scale=0.3] (E) at (0,0) {};
                \node[shape=circle,draw=black,scale=0.3] (F) at (1.8,-3) {};
                \node[shape=circle,draw=black,scale=0.3] (G) at (-1.8,-3) {};
                \node[shape=circle,draw=black,scale=0.3,fill=black] (3) at (3,6.5) {};
                \node[shape=circle,draw=black,scale=0.3,fill=black] (1) at (-3,6.5) {}; 
                \node[shape=circle,draw=black,scale=0.3,fill=black] (4) at (3,-4.5) {};
                \node[shape=circle,draw=black,scale=0.3,fill=black] (2) at (-3,-4.5) {};
                \node at (3.2,7.2) {$3$};
                \node at (-3.2,7.2) {$1$};
                \node at (3.2,-5.2) {$4$};
                \node at (-3.2,-5.2) {$2$};                
                \draw [->](B)--(C) node[midway,right=0.1] {$h$};
                \draw [dashed,->](B)--(D) node[midway,left=0.1] {$k$};
                \draw [dashed,->](C)--(D) node[midway,above=0.05] {$f$};
                \draw [->](C)--(3) node[midway,below right] {$c$};
                \draw [->](D)--(1) node[midway,below left] {$a$};
                \draw [->](B)--(E) node[midway,left=0.05] {$e$};
                \draw [->](E)--(F) node[midway,right=0.05] {$\ell$};
                \draw [dashed,->](E)--(G) node[midway,left=0.1] {$m$};
                \draw [dashed,->](F)--(G) node[midway,below=0.05] {$n$};               
                \draw [->](G)--(2) node[midway,above left] {$b$};
                \draw [->](F)--(4) node[midway,above right] {$d$};          
            \end{tikzpicture}
        \caption{}
        \end{subfigure}
        \begin{subfigure}{0.44\textwidth}
            \centering
            \begin{tikzpicture}[scale = .4]           
                \node[shape=circle,draw=black,scale=0.3] (B) at (0,2) {};
                \node[shape=circle,draw=black,scale=0.3] (C) at (1.8,5) {};
                \node[shape=circle,draw=black,scale=0.3] (D) at (-1.8,5) {};
                \node[shape=circle,draw=black,scale=0.3] (E) at (0,0) {};
                \node[shape=circle,draw=black,scale=0.3] (F) at (1.8,-3) {};
                \node[shape=circle,draw=black,scale=0.3] (G) at (-1.8,-3) {};
                \node[shape=circle,draw=black,scale=0.3,fill=black] (3) at (3,6.5) {};
                \node[shape=circle,draw=black,scale=0.3,fill=black] (1) at (-3,6.5) {}; 
                \node[shape=circle,draw=black,scale=0.3,fill=black] (4) at (3,-4.5) {};
                \node[shape=circle,draw=black,scale=0.3,fill=black] (2) at (-3,-4.5) {};
                \node at (3.2,7.2) {$4$};
                \node at (-3.2,7.2) {$1$};
                \node at (3.2,-5.2) {$3$};
                \node at (-3.2,-5.2) {$2$};                
                \draw [->](B)--(C) node[midway,right=0.1] {$h$};
                \draw [dashed,->](B)--(D) node[midway,left=0.1] {$k$};
                \draw [dashed,->](C)--(D) node[midway,above=0.05] {$f$};
                \draw [->](C)--(3) node[midway,below right] {$c$};
                \draw [->](D)--(1) node[midway,below left] {$a$};
                \draw [->](B)--(E) node[midway,left=0.05] {$e$};
                \draw [->](E)--(F) node[midway,right=0.05] {$\ell$};
                \draw [dashed,->](E)--(G) node[midway,left=0.1] {$m$};
                \draw [dashed,->](F)--(G) node[midway,below=0.05] {$n$};               
                \draw [->](G)--(2) node[midway,above left] {$b$};
                \draw [->](F)--(4) node[midway,above right] {$d$};      
            \end{tikzpicture}
        \caption{}
        \end{subfigure}
    \caption{All double-triangle quarnets without the split $12|34$.}
    \label{fig:doubletris}
\end{figure}
As in the single-triangle case, we have that (a) can be obtained from (b) with the leaf-permutation $(34)$, so it is sufficient to prove that $Q_{12|34}$ is strictly positive on the quarnet in (a). The parameterization is given by
$$q_{g_1g_2g_3g_4} = a_{g_1}b_{g_2}c_{g_3}d_{g_4}e_{g_2+g_4}(\delta_1 f_{g_1}h_{g_1+g_3} + (1-\delta_1)k_{g_1}h_{g_3})(\delta_2 n_{g_2}\ell_{g_2+g_4} + (1-\delta_2)m_{g_2}\ell_{g_4}).$$
We have
\begin{align*}
    q_{\rm AAAA} &= 1, \\
    q_{\rm TTTT} &= a_{\rm T}b_{\rm T}c_{\rm T}d_{\rm T}(\delta_1 f_{\rm T} + (1-\delta_1)k_{\rm T}h_{\rm T})(\delta_2 n_{\rm T} + (1-\delta_2)m_{\rm T}\ell_{\rm T}), \\
    q_{\rm AATT} &= c_{\rm T}d_{\rm T}e_{\rm T}h_{\rm T}\ell_{\rm T}, \\
    q_{\rm TTAA} &= a_{\rm T}b_{\rm T}e_{\rm T}(\delta_1 f_{\rm T}h_{\rm T} + (1-\delta_1)k_{\rm T})(\delta_2 n_{\rm T}\ell_{\rm T} + (1-\delta_2)m_{\rm T}).
\end{align*}
This gives
\begin{align*}
    Q_{12|34} &= a_{\rm T}b_{\rm T}c_{\rm T}d_{\rm T}(\delta_1 f_{\rm T} + (1-\delta_1)k_{\rm T}h_{\rm T})(\delta_2 n_{\rm T} + (1-\delta_2)m_{\rm T}\ell_{\rm T}) \\
                &\qquad - a_{\rm T}b_{\rm T}c_{\rm T}d_{\rm T}e_{\rm T}^2h_{\rm T}\ell_{\rm T}(\delta_1 f_{\rm T}h_{\rm T} + (1-\delta_1)k_{\rm T})(\delta_2 n_{\rm T}\ell_{\rm T} + (1-\delta_2)m_{\rm T}) \\
              &= a_{\rm T}b_{\rm T}c_{\rm T}d_{\rm T}(\delta_1 f_{\rm T}(1-e_{\rm T}h_{\rm T}^2) + (1-\delta_1)k_{\rm T}h_{\rm T}(1-e_{\rm T})) \\
              &\qquad\times (\delta_2 n_{\rm T}(1-e_{\rm T}\ell_{\rm T}^2) + (1-\delta_2)m_{\rm T}\ell_{\rm T}(1-e_{\rm T})) > 0.
\end{align*}

\end{proof}

\end{document}